%% file: main.tex
\documentclass[appendix,runningheads,orivec,envcountsame]{llncs}

\usepackage{microtype}
\usepackage{mathpartir}
\usepackage{breakcites}
\usepackage{enumitem}
\usepackage{style/notation}
\usepackage{style/appendix}
\usepackage[pass]{geometry}

\begin{document}

\title{Monitoring Blackbox Implementations \texorpdfstring{\\}{} of  Multiparty Session Protocols\thanks{%
This paper is an extension of~\cite{conf/rv/vdHeuvelPD23} with appendices. \\
This research has been supported by the Dutch Research Council (NWO) under project No.\ 016.Vidi.189.046 (Unifying
Correctness for Communicating Software).}}
\titlerunning{Monitoring Blackbox Implementations of Multiparty Session Types}
\author{Bas van den Heuvel\orcidID{0000-0002-8264-7371} \and Jorge A.\ Pérez\orcidID{0000-0002-1452-6180} \and Rares A.\ Dobre}
\institute{University of Groningen, Groningen, The Netherlands}

\maketitle

\begin{abstract}
    We present a framework for the distributed monitoring of networks of components that coordinate by  message-passing, following multiparty session protocols specified as global types.
   We improve over prior works by (i)~supporting components whose exact specification is unknown (``blackboxes'') and (ii)~covering protocols that cannot be analyzed by existing techniques.
    We first give a procedure for synthesizing monitors for blackboxes from global types, and precisely define when a  blackbox correctly \emph{satisfies} its global type.
    Then, we prove that  monitored blackboxes are  \emph{sound} (they correctly follow the protocol) and \emph{transparent} (blackboxes with and without monitors are behaviorally equivalent).
    \keywords{distributed monitoring \and message-passing \and concurrency \and  multiparty session types }
\end{abstract}

\section{Introduction}
\label{s:introduction}

%

Runtime verification excels at analyzing systems with components that cannot be (statically) checked, such as closed-source and third-party components with unknown/partial specifications~\cite{book/BartocciFFR18,book/FrancalanzaPS18}.
In this spirit, we present a monitoring framework for networks of  communicating components.
We adopt \emph{global types} from \emph{multiparty session types}~\cite{conf/popl/HondaYC08,journal/acm/HondaYC16} both to specify protocols and to synthesize monitors.
As we explain next, rather than process implementations, we consider \emph{``blackboxes''}---components whose exact structure is unknown.
Also, aiming at wide applicability, we cover networks of monitored components that implement global types that go beyond the scope of existing techniques.

Session types provide precise specifications of the protocols that components should respect.
It is then natural to use session types as references for distributed monitoring~\cite{journal/tcs/BocchiCDHY17,conf/ecoop/BartoloBurloFS21,conf/popl/JiaGP16,journal/jlamp/GommerstadtJP22,conf/tgc/Thiemann14,conf/icfp/IgarashiTVW17}.
In particular, Bocchi \etal~\cite{journal/tcs/BocchiCDHY17,conf/forte/BocchiCDHY13,conf/tgc/ChenBDHY11} use multiparty session types to monitor networks of \picalc processes.
Leveraging notions originally conceived for static verification (such as \emph{global types} and their \emph{projection} onto \emph{local types}), their framework guarantees the correctness of monitored networks with statically and dynamically checked components.

However, existing monitoring techniques based on multiparty session types have two limitations.
One concerns the class of protocols they support; the other is their reliance on fully specified components, usually given as (typed) processes.
That is, definitions of networks assume that a component can be inspected---an overly strong assumption in many cases.
There is then a tension between (i)~the assumptions on component structure and (ii)~the strength of formal guarantees: the less we know about components, the harder it is to establish such guarantees.

\paragraph{Our approach.}
We introduce a new approach to monitoring based on multiparty session types that relies on minimal assumptions on a component's structure.
As key novelty, we consider \emph{blackboxes}---components with unknown structure but observable behavior---and \emph{networks} of monitored blackboxes that use asynchronous message-passing to implement multiparty protocols specified as global types.

As running example, let us consider the global type $\gtc{G_{\mathsf{a}}}$ (inspired by an example by Scalas and Yoshida~\cite{conf/popl/ScalasY19}), which expresses an authorization protocol between three participants: server ($s$), client ($c$), and authorization service ($a$).
\begin{equation}
    \gtc{G_{\mathsf{a}}} := \rec X \gtc. s \send c \gtBraces{ \msg \mathsf{login}<> \gtc. c \send a \gtBraces{ \msg \mathsf{pwd}<\mathsf{str}> \gtc. a \send s \gtBraces{ \msg \mathsf{succ}<\mathsf{bool}> \gtc. X } } \gtc, \msg \mathsf{quit}<> \gtc. \pend }
    \label{eq:auth}
\end{equation}
This recursive global type ($\rec X$) specifies that $s$ sends to $c$ ($s \send c$) a choice between labels $\mathsf{login}$ and $\mathsf{quit}$.
In the $\mathsf{login}$-branch, $c$ sends to $a$ a password of type $\<\mathsf{str}\>$ and $a$ notifies $s$ whether it was correct, after which the protocol repeats ($X$).
In the $\mathsf{quit}$-branch, the protocol simply ends ($\pend$).
As explained in~\cite{conf/popl/ScalasY19}, $\gtc{G_{\mathsf{a}}}$ is not supported by most theories of multiparty sessions, including those in~\cite{journal/tcs/BocchiCDHY17,conf/forte/BocchiCDHY13,conf/tgc/ChenBDHY11}.

\begin{figure}[t]
    \begin{center}
        \begin{tikzpicture}[every loop/.style={min distance=5mm, looseness=10}]
            \begin{scope}[local bounding box=LTSc, node distance=4mm, every node/.style={inner sep=.5mm}]
                \node (n0) {$P_c$};
                \node [below=of n0] (n1) {};
                \node [right=of n1] (n2) {};
                \node [below=of n2] (n3) {};
                \draw [-to] (n0) to (n1);
                \draw [-to] (n1) to (n3);
                \draw [-to] (n0) to (n2);
                \draw [-to] (n2) to (n3);
                \draw [-to] (n2) to [loop right] ();
            \end{scope}
            \node [draw,fit=(LTSc), inner sep=.5mm] (LTScBB) {};
            \node [draw,right=-.4pt of LTScBB.north east, anchor=north west] (Mc) {$\mc{M_c}$};

            \node [draw,right=2cm of Mc.east] (Ms) {$\mc{M_s}$};
            \begin{scope}[local bounding box=LTSs, shift={($(Ms.north east)+(2.5mm+.5pt,-2.5mm+.5pt)$)}, node distance=4mm, every node/.style={inner sep=.5mm}]
                \node (n0) {$P_s$};
                \node [below=of n0] (n1) {};
                \node [right=of n1] (n2) {};
                \node [below=of n1] (n3) {};
                \draw [-to] (n0) to (n1);
                \draw [-to] (n1) to (n2);
                \draw [-to] (n0) to (n2);
                \draw [-to] (n1) to (n3);
                \draw [-to] (n3) to [loop right] ();
            \end{scope}
            \node [draw,fit=(LTSs),inner sep=.5mm] (LTSsBB) {};

            \node [right=1cm of Mc.east, anchor=center, yshift=-1.0cm] (G) {$\gtc{G_{\mathsf{a}}}$};

            \node [draw, below=6.5mm of G.south, anchor=north] (Ma) {$\mc{M_a}$};
            \begin{scope}[local bounding box=LTSa, shift={($(Ma.south west)+(2.7mm+.5pt,-2.4mm+.5pt)$)}, node distance=4mm, every node/.style={inner sep=.5mm}]
                \node (n0) {$P_a$};
                \node [right=of n0] (n2) {};
                \node [below=of n2] (n1) {};
                \node [right=of n2] (n4) {};
                \node [below=of n4] (n3) {};
                \draw [-to] (n0) to (n2);
                \draw [-to] (n2) to (n4);
                \draw [-to] (n4) to (n3);
                \draw [-to] (n3) to (n1);
                \draw [-to] (n1) to (n2);
                \node [shift={($(n3.south east)+(1pt,-1pt)$)}] {};
            \end{scope}
            \node [draw,fit=(LTSa), inner sep=.5mm] (LTSaBB) {};

            \node [left=8mm of G, yshift=-2mm, anchor=center, inner sep=0mm] (Rca) {$\rtc{R_{c,a}}$};
            \node [right=8mm of G, yshift=-2mm, anchor=center, inner sep=0mm] (Rsa) {$\rtc{R_{s,a}}$};
            \node [above=8.0mm of G, anchor=center] (Rcs) {$\rtc{R_{c,s}}$};

            \draw [gray,-to,line width=1pt] (G) to (Rca);
            \draw [gray,-to,line width=1pt] (G) to (Rsa);
            \draw [gray,-to,line width=1pt] (G) to (Rcs);

            \draw [-to,line width=2pt] (G) to (Mc);
            \draw [-to,line width=2pt] (G) to (Ms);
            \draw [-to,line width=2pt] (G) to (Ma);

            \draw [-to,dotted] ($(Rca.south)+(2mm,-.5mm)$) to (Ma.north west);
            \draw [-to,dotted] ($(Rsa.south)+(-2mm,-.5mm)$) to (Ma.north east);
            \draw [-to,dotted] ($(Rca.north)+(0mm,1mm)$) to (Mc);
            \draw [-to,dotted] (Rcs) to (Mc);
            \draw [-to,dotted] ($(Rsa.north)+(0mm,1mm)$) to (Ms);
            \draw [-to,dotted] (Rcs) to (Ms);
        \end{tikzpicture}%
    \end{center}

    \caption{ \label{f:setup}
        Monitoring setup based on the global type (multiparty protocol) $\gtc{G_{\mathsf{a}}}$~\eqref{eq:auth}.
        Each protocol participant has a blackbox (an LTS), attached to a monitor (e.g. $P_c$ and~$\mc{M_c}$).
        The monitors are synthesized from $\gtc{G_{\mathsf{a}}}$ (thick arrows).
        Relative types (e.g. $\rtc{R_{c,s}}$) obtained by projection from $\gtc{G_{\mathsf{a}}}$ (thin gray arrows) are used in this synthesis (dotted arrows).
    }
\end{figure}

\Cref{f:setup} illustrates our approach to monitoring global types such as $\gtc{G_{\mathsf{a}}}$.
There is a blackbox per participant, denoted $P_s$, $P_c$, and $P_a$, whose behavior is given by a labeled transition system (LTS).
Each blackbox implements a participant as dictated by $\gtc{G_{\mathsf{a}}}$
while coupled with a monitor ($\mc{M_s}$, $\mc{M_c}$, and $\mc{M_a}$ in \Cref{f:setup}).
Monitors are \emph{synthesized} from $\gtc{G_{\mathsf{a}}}$ by relying on \emph{relative types}~\cite{journal/scico/vdHeuvelP22}, which provide local views of the global type: they
specify protocols between \emph{pairs} of participants; hence, in the case of $\gtc{G_{\mathsf{a}}}$, we have three relative types:  $\rtc{R_{c,s}}$, $\rtc{R_{c,a}}$, and $\rtc{R_{s,a}}$.

Introduced in~\cite{journal/scico/vdHeuvelP22} for type-checking communicating components, relative types are instrumental to our approach.
They give a fine-grained view of protocols that is convenient for monitor synthesis.
Relative types explicitly specify \emph{dependencies} between participants, e.g., when the behavior of a participant $p$ is the result of a prior choice made by some other participants $q$ and $r$.
Treating dependencies as explicit messages is key to ensuring the distributed implementability of protocols that usual multiparty theories cannot support (e.g., $\gtc{G_{\mathsf{a}}}$~\eqref{eq:auth}). 
Our algorithm for monitor synthesis mechanically computes these dependencies from relative types, and exploits them to coordinate monitored blackboxes.

A central ingredient in our technical developments is the notion of \emph{satisfaction} (\cref{d:satisfaction}), which defines when a monitored blackbox conforms to the role of a specific protocol participant.
Building upon satisfaction, we prove \emph{soundness} and \emph{transparency} for networks of monitored blackboxes.
Soundness (\cref{t:soundnessMain}) ensures that if each monitored blackbox in a network behaves correctly (according to a global type), then the entire network behaves correctly too.
Transparency (\cref{t:transparencyMain}) ensures that monitors do not interfere with the (observable) behavior of their contained blackboxes; it is given in terms of a (weak) behavioral equivalence, which is suitably informed by the actions of a given global type.


\paragraph{Related work.}
The literature on distributed runtime verification is vast.
In this setting, the survey by Francalanza \etal~\cite{book/FrancalanzaPS18} proposes several classification criteria.
Phrased in terms of their criteria,  our work concerns  {distributed} monitoring for  {asynchronous} message-passing.
We work with blackboxes, whose monitors are minimally {intrusive}: they do not alter behavior, but do contribute to coordination.

The works by Bocchi \etal~\cite{journal/tcs/BocchiCDHY17,conf/forte/BocchiCDHY13,conf/tgc/ChenBDHY11} and by Scalas and Yoshida~\cite{conf/popl/ScalasY19}, mentioned above, are a main source of inspiration to us.
The work~\cite{conf/popl/ScalasY19} highlights the limitations of techniques based on the projection of a global type onto local types:  many practical protocols, such as $\gtc{G_{\mathsf{a}}}$, cannot be analyzed because their projection onto local types is undefined.
With respect to~\cite{journal/tcs/BocchiCDHY17,conf/forte/BocchiCDHY13,conf/tgc/ChenBDHY11}, there are three major differences.
First,  Bocchi \etal rely on precise specifications of components ($\pi$-calculus processes), whereas we monitor blackboxes (LTSs).
Second, we resort to relative types, whereas they rely on {local types}; this is a limitation, as just mentioned.
Third, their monitors drop incorrect messages (cf.~\cite{conf/concur/AcetoCFI18}) instead of signaling errors, as we do.
Their framework ensures transparency (akin to \cref{t:transparencyMain}) and safety, i.e., monitored components do not misbehave.
In contrast, we establish soundness, which is different and more technically involved than safety: our focus is on monitoring blackboxes rather than fully specified components, and soundness concerns correct behavior rather than the absence of misbehavior.

We mention runtime verification techniques based on \emph{binary} session types, a sub-class of multiparty session types.
Bartolo Burlò \etal~\cite{conf/ecoop/BartoloBurloFS21} monitor sequential processes that communicate synchronously, prove that ill-typed processes raise errors, and consider also probabilistic session types~\cite{journal/scico/BartoloBurloFSTT22,conf/coordination/BartoloBurloFSTT21}.
Other works couple monitoring session types with \emph{blame assignment} upon protocol violations~\cite{conf/popl/JiaGP16,journal/jlamp/GommerstadtJP22,conf/tgc/Thiemann14,journal/jlamp/IgarashiTTVW19}.
Jia \etal~\cite{conf/popl/JiaGP16} monitor asynchronous session-typed processes.
Gommerstadt \etal~\cite{journal/jlamp/GommerstadtJP22,conf/esop/GommerstadtJP18} extend~\cite{conf/popl/JiaGP16} with rich refinement-based contracts.
We do not consider blame assignment, but it can conceivably be added by enhancing error signals.

\paragraph{Outline.}
\Cref{s:networks} defines networks of monitored blackboxes and their behavior.
\Cref{s:globalTypesAsMonitors} defines how to synthesize monitors from global types.
\Cref{s:correctMonitoredBlackboxes} defines correct   monitored blackboxes, and establishes soundness and transparency.
\cref{s:conclusion} concludes the paper. We use colors to improve readability.

\ifappendix
The appendix (\cpageref{a:exampleBocchi}) includes
\else
The full version of this paper~\cite{report/vdHeuvelPD23} includes an appendix with
\fi
additional examples (including the running example from~\cite{journal/tcs/BocchiCDHY17}), a description of a practical toolkit based on this paper, and omitted proofs.

\section{Networks of Monitored Blackboxes}
\label{s:networks}

We write $P,Q,\ldots$ to denote \emph{blackbox processes} (simply \emph{blackboxes}) that implement protocol \emph{participants} (denoted $p,q,\ldots$).
We assume that a blackbox $P$ is associated with an LTS that specifies its behavior.
Transitions are denoted~$P \ltrans{\alpha} P'$.
Actions $\alpha$, defined in \cref{f:netGrammars} (top), encompass messages $m$, which can be labeled data but also \emph{dependency messages} (simply \emph{dependencies}).
As we will see, dependencies are useful to ensure the coordinated implementation of choices.
Messages abstract away from values, and include only their type.


\begin{figure}[t]
    Let $\ell$ and $T$ denote a label and a data type, respectively.
    \begin{align*}
        & \text{Actions}
        & \alpha,\beta &
        \begin{array}[t]{@{}l@{}lr@{}lr@{}lr@{}}
            {} ::= {}&
            \tau & \text{(silent)}
            \sepr* &
            m & \text{(message)}
            \sepr* &
            \pend & \text{(end)}
        \end{array}
        \\
        & \text{Messages}
        & m,n &
        \begin{array}[t]{@{}l@{}lr@{}lr@{}}
            {} ::= {}&
            p \send q (\msg \ell<T>) & \text{(output)}
            \sepr* &
            p \send q \Parens{\ell} & \text{(dep. output,  {only for networks})}
            \\ \sepr* &
            p \recv q (\msg \ell<T>) & \text{(input)}
            \sepr* &
            p \recv q \Parens{\ell} & \text{(dep. input)}
        \end{array}
    \end{align*}
    Given $\alpha$, the \emph{recipient in $\alpha$} is defined as follows:
    \begin{align*}
        \recip(p \send q(\msg \ell<T>)) &:= \recip(p \send q\Parens{\ell}) := q
        \\
        \recip(p \recv q(\msg \ell<T>)) &:= \recip(p \recv q\Parens{\ell}) := \recip(\tau) := \recip(\pend) := \text{undefined}
    \end{align*}

    \hrulefill

    \smallskip
    Let $D,E,\ldots$ denote sets of participants; $I,J,\ldots$ denote non-empty sets of labels; $X,Y,\ldots$ denote recursion variables.
    \begin{align*}
        & \text{Networks}
        & \net P,\net Q &
        \begin{array}[t]{@{}l@{}lr@{}}
            {} ::= {} &
            \monImpl{ \bufImpl{ p }{ P }{ \vect m } }{ \mc{M} }{ \vect n } & \text{(monitored blackbox)}
            \\ \sepr* &
            \net P \| \net Q & \text{(parallel composition)}
            \\ \sepr* &
            \perror_D & \text{(error signal)}
        \end{array}
        \\
        & \text{Monitors}
        & \mc{M} &
        \begin{array}[t]{@{}l@{}lr@{}lr@{}}
            {} ::= {} &
            p \send q \mBraces{ \msg i<T_i> \mc. \mc{M} }_{i \in I} & \text{(output)}
            \sepr* &
            p \send D (\ell) \mc. \mc{M} & \text{(dep.\ output)}
            \\ \sepr* &
            p \recv q \mBraces{ \msg i<T_i> \mc. \mc{M} }_{i \in I} & \text{(input)}
            \sepr* &
            p \recv q \mBraces{ i \mc. \mc{M} }_{i \in I} & \text{(dep.\ input)}
            \\ \sepr* &
            \rec X \mc. \mc{M} \sepr X & \text{(recursion)}
            \sepr* &
            \pend & \text{(end)}
            \\ \sepr* &
            \perror & \text{(error)}
            \sepr* &
            \checkmark & \text{(finished)}
        \end{array}
    \end{align*}
    The set of \emph{subjects of a network} is as follows:
    \begin{align*}
        \subjs( \monImpl{ \bufImpl{ p }{ P }{ \vect m } }{ \mc{M} }{ \vect n } ) &:= \braces{ p }
        &
        \subjs(\net P {\|} \net Q) &:= \subjs(\net P) \cup \subjs(\net Q)
        &
        \subjs(\perror_D) &:= D
    \end{align*}

    \caption{Actions, messages, networks, and monitors.}\label{f:netGrammars}
\end{figure}

A silent transition $\tau$ denotes an internal computation.
Transitions $p \send q(\msg \ell<T>)$ and $p \recv q(\msg \ell<T>)$ denote the output and input of a message of type $T$ with label $\ell$ between $p$ and $q$, respectively.
If a message carries no data, we write $\msg \ell<>$ (i.e., the data type is empty).
Dependency outputs are used for monitors, defined below.

We adopt minimal assumptions about the behavior of blackboxes:

\begin{definition}[Assumptions]\label{d:blackboxLTSassumptions}
    We assume the following about LTSs of blackboxes:
    \begin{itemize}

        \item\label{i:bTau}
            \textbf{(Finite $\tau$)}
            Sequences of $\tau$-transitions are finite.

        \item\label{i:bInputOutput}
            \textbf{(Input/Output)}
            There are never input- and output-transitions available at the same time.

        \item\label{i:bEnd}
            \textbf{(End)}
            There are never transitions after an $\pend$-transition.

    \end{itemize}
\end{definition}

\begin{example}\label{ex:cImpl}
    The blackboxes $P_c$, $P_s$, $P_a$ implement $c$, $s$, $a$, respectively, in $\gtc{G_{\mathsf{a}}}$~\eqref{eq:auth} with the following LTSs:
    \begin{center}
        \begin{tikzpicture}[node distance=24mm,every path/.style={->},every edge quotes/.style={font=\scriptsize,auto=left},tight/.style={inner sep=1pt},loose/.style={inner sep=5pt}]
            \node (ci) {$P_c$};
            \node[right=of ci] (cq) {$P_c^{\mathsf{q}}$};
            \node[right=of cq] (ce) {$P_c^{\mathsf{e}}$};
            \node[left=of ci] (cl) {$P_c^{\,\mathsf{l}}$};
            \draw   (ci) edge["$c \recv s (\msg \mathsf{quit}<>)$" tight] (cq)
                    (cq) edge["$\pend$"] (ce)
                    (ci) edge["$c \recv s (\msg \mathsf{login}<>)$"' tight] (cl)
                    (cl) edge[bend right=35,"$c \send a (\msg \mathsf{pwd}<\mathsf{str}>)$" loose] (ci);

            \node[below=10mm of ci] (si) {$P_s$};
            \node[right=of si] (sq) {$P_s^{\mathsf{q}}$};
            \node[right=of sq] (se) {$P_s^{\mathsf{e}}$};
            \node[left=of si] (sl) {$P_s^{\,\mathsf{l}}$};
            \draw   (si) edge["$s \send c (\msg \mathsf{quit}<>)$" tight] (sq)
                    (sq) edge["$\pend$"] (se)
                    (si) edge["$s \send c (\msg \mathsf{login}<>)$"' tight] (sl)
                    (sl) edge[bend right=35,"$s \recv a (\msg \mathsf{succ}<\mathsf{bool}>)$" loose] (si);

            \node[below=10mm of si] (ai) {$P_a$};
            \node[right=of ai] (aqs) {$P_a^{\mathsf{q}_s}$};
            \node[right=of aqs] (aqc) {$P_a^{\mathsf{q}_c}$};
            \node[below=8mm of aqc] (ae) {$P_a^{\mathsf{e}}$};
            \node[left=of ai] (als) {$P_a^{\,\mathsf{l}_s}$};
            \node[below=8mm of als] (alc) {$P_a^{\,\mathsf{l}_c}$};
            \node[below=8.7mm of ai] (ap) {$P_a^{\mathsf{p}}$};
            \draw   (ai)  edge["$a \recv s \Parens{\mathsf{quit}}$" tight] (aqs)
                    (aqs) edge["$a \recv c \Parens{\mathsf{quit}}$" tight] (aqc)
                    (aqc) edge["$\pend$" left] (ae)
                    (ai)  edge["$a \recv s \Parens{\mathsf{login}}$"' tight] (als)
                    (als)  edge["$a \recv c \Parens{\mathsf{login}}$"' tight] (alc)
                    (alc)  edge["$a \recv c (\msg \mathsf{pwd}<\mathsf{str}>)$"' above] (ap)
                    (ap)  edge["$a \send s (\msg \mathsf{succ}<\mathsf{bool}>)$"' tight] (ai);
        \end{tikzpicture}
    \end{center}
    All three LTSs above respect the assumptions in \Cref{d:blackboxLTSassumptions}.
    On the other hand, the following LTS violates all three assumptions; in particular, there are an input- and an output-transition simultaneously enabled at $Q$:
    \begin{center}
        \begin{tikzpicture}[node distance=24mm,every path/.style={->},every edge quotes/.style={font=\scriptsize,auto=left},tight/.style={inner sep=1pt},loose/.style={inner sep=5pt}]
            \node (ci) {$Q$};
            \node[right=of ci] (cq) {$Q^{\mathsf{q}}$};
            \node[left=of ci] (cp) {$Q^{\mathsf{p}}$};
            \draw   (ci) edge["$c \recv s (\msg \mathsf{quit}<>)$" tight] (cq)
                    (cq) edge[bend left=35,"$\pend$"] (ci)
                    (ci) edge["$c \send a (\msg \mathsf{pwd}<\mathsf{str}>)$"' tight] (cp)
                    (cp) edge[loop left,"$\tau$"] (cp);
        \end{tikzpicture}
    \end{center}
\end{example}

Blackboxes communicate {asynchronously}, using buffers (denoted $\vect m$): ordered sequences of messages, with the most-recently received message on the left.
The empty buffer is denoted $\epsi$.
When a blackbox does an input transition, it attempts to read the message from its buffer.
An output transition places the message in the recipient's buffer; to accommodate this, we mark each blackbox with the participant they implement.
The result is a \emph{buffered blackbox}, denoted $\bufImpl{ p }{ P }{ \vect m }$.

By convention, the buffer of $p$ contains output messages with recipient $p$.
We allow the silent reordering of messages with different senders; this way,  e.g., given $q \neq r$, $\vect m , q \send p (\msg \ell<T>) , r \send p (\msg \ell'<T'>) , \vect n$ and $\vect m , r \send p (\msg \ell'<T'>) , q \send p (\msg \ell<T>) , \vect n$ are the same.

Having defined standalone (buffered) blackboxes, we now define how they interact in \emph{networks}.
We couple each buffered blackbox with a \emph{monitor} $\mc{M}$, which has its own buffer~$\vect n$.
The result is a \emph{monitored blackbox}, denoted $\monImpl{ \bufImpl{ p }{ P }{ \vect m } }{ \mc{M} }{ \vect n }$.

Monitors define finite state machines that accept  sequences of incoming and outgoing messages, as stipulated by some protocol.
An \emph{error} occurs when a message exchange does not conform to such protocol.
Additionally, monitors support the dependencies mentioned earlier: when a blackbox sends or receives a message, the monitor broadcasts the message's label to other monitored blackboxes such that they can receive the chosen label and react accordingly.

Networks, defined in \cref{f:netGrammars} (bottom), are  compositions of monitored blackboxes and \emph{error signals}.
An error signal $\perror_D$ replaces a monitored blackbox when its monitor detects an error involving participants in the set $D$.
Indeed, a participant's error will propagate to the other monitored blackboxes in a network.

\noindent
Output and (dependency) input monitors check  outgoing and incoming (dependency) messages, respectively.
Output dependency monitors $p \send D (\ell) \mc. \mc{M}$ broadcast $\ell$ to the participants in $D$.
Recursive monitors are encoded by recursive definitions ($\rec X \mc. \mc{M}$) and recursive calls ($X$). 
The $\pend$ monitor waits for the buffered blackbox to end the protocol.
The $\perror$ monitor denotes an inability to process received messages; it will be useful when the sender and recipient of an exchange send different dependency messages.
The finished monitor $\checkmark$ is self-explanatory.

We now define the behavior of monitored blackboxes in networks:

\begin{definition}[LTS for Networks]\label{d:ltsNetworks}
    We define an \emph{LTS for networks}, denoted $\net P \ltrans{\alpha} \net Q$, by the rules in \cref{f:ltsNetworks} (\cpageref{f:ltsNetworks}) with actions $\alpha$ as in \cref{f:netGrammars} (top).

    \noindent
    We write $\net P \trans* \net Q$ to denote a sequence of zero or more $\tau$-transitions $\net P \ltrans{\tau} \ldots \ltrans{\tau} \net Q$, and
    we write $\net P \ntrans$ to denote that there do not exist $\alpha,\net Q$ such that $\net P \ltrans{\alpha} \net Q$.
\end{definition}

\setlength{\textfloatsep}{\textfloatsepsave}
\begin{figure}[p]
    \def\MathparLineskip{\lineskip=2.3mm}
    \def\defaultHypSeparation{\hskip 1mm}
    \begin{mathpar}
        \begin{bussproof}[buf-out]
            \bussAssume{
                \quad\smash[t]{P \ltrans{ p \send q (\msg \ell<T>) } P'}
            }
            \bussUn{
                \bufImpl{ p }{ P }{ \vect m }
                \mathrel{\raisebox{-2.5pt}[0ex][0ex]{$\ltrans{ p \send q (\msg \ell<T>) }$}}
                \bufImpl{ p }{ P' }{ \vect m }
            }
        \end{bussproof}
        \and
        \begin{bussproof}[buf-in]
            \bussAssume{
                \smash[t]{P \ltrans{ p \recv q (\msg \ell<T>) } P'}
            }
            \bussUn{
                \bufImpl{ p }{ P }{ \vect m , q \send p (\msg \ell<T>) }
                \ltrans{ \tau }
                \bufImpl{ p }{ P' }{ \vect m }
            }
        \end{bussproof}
        \and
        \begin{bussproof}[buf-in-dep]
            \bussAssume{
                \quad\smash[t]{P \ltrans{ p \recv q \Parens{\ell} } P'}
            }
            \bussUn{
                \bufImpl{ p }{ P }{ \vect m , q \send p \Parens{\ell} }
                \ltrans{ \tau }
                \bufImpl{ p }{ P' }{ \vect m }
            }
        \end{bussproof}
        \and
        \begin{bussproof}[buf-tau]
            \bussAssume{
                \quad\smash[t]{P \ltrans{ \tau } P'}
            }
            \bussUn{
                \bufImpl{ p }{ P }{ \vect m }
                \ltrans{ \tau }
                \bufImpl{ p }{ P' }{ \vect m }
            }
        \end{bussproof}
        \and
        \inferLabel{buf-end}~
        \begin{bussproof}
            \bussAssume{
                P
                \ltrans{ \pend }
                P'
            }
            \bussUn{
                \bufImpl{ p }{ P }{ \vect m }
                \mathrel{\raisebox{-2pt}[0ex][0ex]{$\ltrans{ \pend }$}}
                \bufImpl{ p }{ P' }{ \vect m }
            }
        \end{bussproof}
    \end{mathpar}
    \hrule
    \begin{mathpar}
        \inferLabel{mon-out}~
        \begin{bussproof}
            \bussAssume{
                \bufImpl{ p }{ P }{ \vect m }
                \ltrans{ p \send q (\msg j<T_j>) }
                \bufImpl{ p }{ P' }{ \vect m }
            }
            \bussAssume{
                \mc{M} = p \send q \mBraces{ \msg i<T_i> \mc. \mc{M_i} }_{i \in I}
            }
            \bussAssume{
                j \in I
            }
            \bussTern{
                \monImpl{ \bufImpl{ p }{ P }{ \vect m } }{ \mc{M} }{ \vect n }
                \mathrel{\raisebox{-2.5pt}[0ex][0ex]{$\ltrans{ p \send q (\msg j<T_j>) }$}}
                \monImpl{ \bufImpl{ p }{ P' }{ \vect m } }{ \mc{M_j} }{ \vect n }
            }
        \end{bussproof}
        \and
        \inferLabel{mon-in}~
        \begin{bussproof}
            \bussAssume{
                \mc{M} = p \recv q \mBraces{ y \mc. \mc{M_i} }_{y \in Y}
            }
            \bussAssume{
                x \in Y
            }
            \bussAssume{
                n' \in \{ q \send p (x) , q \send p \Parens{x} \}
            }
            \bussTern{
                \monImpl{ \bufImpl{ p }{ P }{ \vect m } }{ \mc{M} }{ \vect n , n' }
                \ltrans{ \tau }
                \monImpl{ \bufImpl{ p }{ P }{ n' , \vect m } }{ \mc{M_j} }{ \vect n }
            }
        \end{bussproof}
        \and
        \inferLabel{mon-tau}~
        \begin{bussproof}
            \bussAssume{
                \bufImpl{ p }{ P }{ \vect m }
                \ltrans{ \tau }
                \bufImpl{ p }{ P' }{ \vect m' }
            }
            \bussAssume{
                \mc{M} \neq \checkmark
            }
            \bussBin{
                \monImpl{ \bufImpl{ p }{ P }{ \vect m } }{ \mc{M} }{ \vect n }
                \ltrans{ \tau }
                \monImpl{ \bufImpl{ p }{ P' }{ \vect m' } }{ \mc{M} }{ \vect n }
            }
        \end{bussproof}
        \\
        \raisebox{2mm}{$\inferLabel{mon-out-dep}~$}
        \begin{bussproof}
            \bussAssume{ }
            \bussUn{
                \monImpl{ \bufImpl{ p }{ P }{ \vect m } }{ p \send (D \cup \braces{q}) (\ell)  \mc. \mc{M'} }{ \vect n }
                \mathrel{\raisebox{-2.5pt}[0ex][0ex]{$\ltrans{ p \send q \Parens{\ell} }$}}
                \monImpl{ \bufImpl{ p }{ P }{ \vect m } }{ p \send D (\ell) \mc. \mc{M'} }{ \vect n }
            }
        \end{bussproof}
        \and
        \raisebox{2mm}{$\inferLabel{mon-out-dep-empty}~$}
        \begin{bussproof}
            \bussAssume{ }
            \bussUn{
                \monImpl{ \bufImpl{ p }{ P }{ \vect m } }{ p \send \emptyset (\ell) \mc. \mc{M'} }{ \vect n }
                \ltrans{ \tau }
                \monImpl{ \bufImpl{ p }{ P }{ \vect m } }{ \mc{M'} }{ \vect n }
            }
        \end{bussproof}
        \and
        \inferLabel{mon-rec}~
        \begin{bussproof}
            \bussAssume{
                \monImpl{ \bufImpl{ p }{ P }{ \vect m } }{ \mc{M} \braces{ \rec X \mc. \mc{M} / X } }{ \vect n }
                \ltrans{ \alpha }
                \monImpl{ \bufImpl{ p }{ P' }{ \vect m' } }{ \mc{M'} }{ \vect n' }
            }
            \bussUn{
                \monImpl{ \bufImpl{ p }{ P }{ \vect m } }{ \rec X \mc. \mc{M} }{ \vect n }
                \ltrans{ \alpha }
                \monImpl{ \bufImpl{ p }{ P' }{ \vect m' } }{ \mc{M'} }{ \vect n' }
            }
        \end{bussproof}
        \and
        \inferLabel{mon-end}~
        \begin{bussproof}
            \bussAssume{
                \bufImpl{ p }{ P }{ \vect m }
                \ltrans{ \pend }
                \bufImpl{ p }{ P' }{ \vect m }
            }
            \bussUn{
                \monImpl{ \bufImpl{ p }{ P }{ \vect m } }{ \pend }{ \epsi }
                \mathrel{\raisebox{-2pt}[0ex][0ex]{$\ltrans{ \pend }$}}
                \monImpl{ \bufImpl{ p }{ P' }{ \vect m } }{ \checkmark }{ \epsi }
            }
        \end{bussproof}
    \end{mathpar}
    \hrule
    \begin{mathpar}
        \begin{bussproof}[error-out]
            \bussAssume{
                \bufImpl{ p }{ P }{ \vect m }
                \ltrans{ p \send q (\msg j<T_j>) }
                \bufImpl{ p }{ P' }{ \vect m }
            }
            \bussAssume{
                \smash[t]{\begin{array}[b]{@{}l@{}}
                        \mc{M} = p \send r \mBraces{ \msg i<T_i> \mc. \mc{M_i} }_{i \in I}
                    \implies
                    (r \neq q \vee j \notin I)
                    \\
                    \mc{M} \notin \braces{ \rec X \mc. \mc{M'} , p \send D (\ell) \mc. \mc{M'} }
                \end{array}}
            }
            \bussBin{
                \monImpl{ \bufImpl{ p }{ P }{ \vect m } }{ \mc{M} }{ \vect n }
                \ltrans{ \tau }
                \perror_{\braces{ p }}
            }
        \end{bussproof}
        \and
        \inferLabel{error-end}~
        \begin{bussproof}
            \bussAssume{
                \bufImpl{ p }{ P }{ \vect m } \ltrans{ \pend } \bufImpl{ p }{ P' }{ \vect m }
            }
            \bussAssume{
                \vect n = \epsi \implies \mc{M} \notin \braces{ \rec X \mc. \mc{M'} , p \send D (\ell) \mc. \mc{M'} , \pend }
            }
            \bussBin{
                \monImpl{ \bufImpl{ p }{ P }{ \vect m } }{ \mc{M} }{ \vect n } \ltrans{ \tau } \perror_{\braces{p}}
            }
        \end{bussproof}
        \and
        \inferLabel{error-in}~
        \begin{bussproof}
            \bussAssume{
                \mc{M} = p \recv q \mBraces{ y \mc. \mc{M_y} }_{y \in Y}
            }
            \bussAssume{
                x \notin Y
            }
            \bussAssume{
                n' \in \{ q \send p (x) , q \send p \Parens{x} \}
            }
            \bussTern{
                \monImpl{ \bufImpl{ p }{ P }{ \vect m } }{ \mc{M} }{ \vect n , n' }
                \ltrans{ \tau }
                \perror_{\braces{ p }}
            }
        \end{bussproof}
        \and
        \begin{bussproof}[error-mon]
            \bussAssume{ }
            \bussUn{
                \monImpl{ \bufImpl{ p }{ P }{ \vect m } }{ \perror }{ \vect n }
                \ltrans{ \tau }
                \perror_{\braces{ p }}
            }
        \end{bussproof}
        \and
        \begin{bussproof}[par-error]
            \bussAssume{ }
            \bussUn{
                \perror_D \| \monImpl{ \bufImpl{ p }{ P }{ \vect m } }{ \mc{M} }{ \vect n }
                \ltrans{ \tau }
                \perror_{D \cup \braces{ p }}
            }
        \end{bussproof}
    \end{mathpar}
    \hrule
    \begin{mathpar}
        \inferLabel{out-mon-buf}~
        \begin{bussproof}
            \bussAssume{
                \net P
                \ltrans{ n' }
                \net P'
            }
            \bussAssume{
                n' \in \{ q \send p (x) , q \send p \Parens{x} \}
            }
            \bussBin{
                \net P \| \monImpl{ \bufImpl{ p }{ P }{ \vect m } }{ \mc{M} }{ \vect n }
                \ltrans{ \tau }
                \net P' \| \monImpl{ \bufImpl{ p }{ P }{ \vect m } }{ \mc{M} }{ n' , \vect n }
            }
        \end{bussproof}
        \and
        \inferLabel{par}~
        \begin{bussproof}
            \bussAssume{
                \net P
                \ltrans{ \alpha }
                \net P'
            }
            \bussAssume{
                \recip(\alpha) \notin \subjs(\net Q)
            }
            \bussBin{
                \net P \| \net Q
                \ltrans{ \alpha }
                \net P' \| \net Q
            }
        \end{bussproof}
        \and
        \inferLabel{cong}~
        \begin{bussproof}
            \bussAssume{
                \net P \equiv \net P'
            }
            \bussAssume{
                \net P'
                \ltrans{ \alpha }
                \net Q'
            }
            \bussAssume{
                \net Q' \equiv \net Q
            }
            \bussTern{
                \net P
                \mathrel{\raisebox{-1pt}[0ex][0ex]{$\ltrans{ \alpha }$}}
                \net Q
            }
        \end{bussproof}
    \end{mathpar}
    \caption{\nameref{d:ltsNetworks} (\cref{d:ltsNetworks}).}\label{f:ltsNetworks}
\end{figure}

\noindent
\cref{f:ltsNetworks} gives four groups of rules, which we briefly discuss.
The Transition group \ruleLabel*{buf} defines the behavior of a buffered blackbox in terms of the behavior of the blackbox it contains; note that input transitions are hidden as $\tau$-transitions.
The Transition group~\ruleLabel*{mon} defines the behavior of a monitored blackbox when the behavior of the enclosed buffered blackbox concurs with the monitor; again, input transitions are hidden as $\tau$-transitions.

When the behavior of the buffered blackbox does not concur with the monitor, the Transition group~\ruleLabel*{error} replaces the monitored blackbox with an error signal.
Transition~\ruleLabel{par-error} propagates error signals to parallel monitored blackboxes.
If a network parallel to the monitored blackbox of $p$ has an outgoing message with recipient $p$, Transition~\ruleLabel{out-mon-buf} places this message in the buffer of the monitored blackbox as a $\tau$-transition.
Transition~\ruleLabel{par} closes transitions under parallel composition, as long as the recipient in the action of the transition ($\recip(\alpha)$) is not a subject of the composed network ($\subjs(\net Q)$, the participants for which monitored blackboxes and error signals appear in $\net Q$).
Transition~\ruleLabel{cong} closes transitions under $\equiv$, which denotes a congruence that defines parallel composition as commutative and associative.

\Cref{f:trans} shows transitions of correct/incorrect communications in networks.

\setlength{\textfloatsep}{\textfloatsepsave}
\begin{figure}[t]
    \begin{align*}
        & \begin{array}{@{}rl@{}}
            & \monImpl{ \bufImpl{ c }{ P_c }{ \epsi } }{ c \recv s \mBraces{ \msg \mathsf{quit}<> \mc. \pend } }{ \epsi }
            \\
            \|
            & 
                \monImpl{ \bufImpl{ s }{ P_s }{ \epsi } }{ s \send c \mBraces{ \msg \mathsf{quit}<> \mc. \pend } }{ \epsi }
        \end{array}
        \ltrans{ \tau }
        \begin{array}{@{}rl@{}}
            & 
                \monImpl{ \bufImpl{ c }{ P_c }{ \epsi } }{ c \recv s \mBraces{ \msg \mathsf{quit}<> \mc. \pend } }{ s \send c (\msg \mathsf{quit}<>) }
            \\
            \|
            & \monImpl{ \bufImpl{ s }{ P_s^{\mathsf{q}} }{ \epsi } }{ \pend }{ \epsi }
        \end{array}
        \\
        & \hspace{25em}
        \downarrow \mkern-4mu {\scriptstyle \tau }
        \\
        & 
        \begin{array}{@{}rl@{}}
            & \monImpl{ \bufImpl{ c }{ P_c^{\mathsf{e}} }{ \epsi } }{ \checkmark }{ \epsi }
            \\
            \|
            & \monImpl{ \bufImpl{ s }{ P_s^{\mathsf{e}} }{ \epsi } }{ \checkmark }{ \epsi }
        \end{array}
        \xleftarrow{\pend}\xleftarrow{\pend}
        \begin{array}{@{}rl@{}}
            & 
                \monImpl{ \bufImpl{ c }{ P_c^{\mathsf{q}} }{ \epsi } }{ \pend }{ \epsi }
            \\
            \|
            & 
                \monImpl{ \bufImpl{ s }{ P_s^{\mathsf{q}} }{ \epsi } }{ \pend }{ \epsi }
        \end{array}
        \xleftarrow{\tau}
        \begin{array}{@{}rl@{}}
            & 
                \monImpl{ \bufImpl{ c }{ P_c }{ s \send c (\msg \mathsf{quit}<>) } }{ \pend }{ \epsi }
            \\
            \|
            & \monImpl{ \bufImpl{ s }{ P_s^{\mathsf{q}} }{ \epsi } }{ \pend }{ \epsi }
        \end{array}
        \\[-0.7ex] \cline{1-2}
        & \begin{array}{@{}rl@{}}
            & 
            \monImpl{ \bufImpl{ c }{ P_c }{ \epsi } }{ c \recv s \mBraces{ \msg \mathsf{login}<> \mc. \mc{M_c^{\mathsf{l}}} } }{ s \send c (\msg \mathsf{quit}<>) }
            \\
            \|
            & \monImpl{ \bufImpl{ s }{ P_s^{\mathsf{q}} }{ \epsi } }{ \pend }{ \epsi }
        \end{array}
        \ltrans{\tau}
        \begin{array}{@{}rl@{}}
            & 
                \perror_{\braces{c}}
            \\
            \|
            & \monImpl{ \bufImpl{ s }{ P_s^{\mathsf{q}} }{ \epsi } }{ \pend }{ \epsi }
        \end{array}
        \ltrans{\tau}
        \perror_{\braces{c,s}}
    \end{align*}
    \caption{
        The \nameref{d:ltsNetworks} at work:
        transitions of correctly/incorrectly communicating monitored blackboxes of participants of $\gtc{G_{\mathsf{a}}}$~\eqref{eq:auth}.
        Top: $s$ sends to $c$ label $\mathsf{quit}$, monitor of $c$ reads message, blackbox of $c$ reads message, both components end.
        Bottom: monitor of $c$ expects $\mathsf{login}$ message but finds $\mathsf{quit}$ message so signals error, error propagates to $s$.
    }
    \label{f:trans}
\end{figure}

\section{Monitors for Blackboxes Synthesized from Global Types}
\label{s:globalTypesAsMonitors}
In theories of multiparty session types~\cite{conf/popl/HondaYC08,journal/acm/HondaYC16},  \emph{global types} conveniently describe message-passing protocols between sets of participants from a vantage point.
Here we use them as specifications for monitors in networks (\cref{alg:gtToMon});
for a local view of such global protocols we use \emph{relative types}~\cite{journal/scico/vdHeuvelP22}, which describe the interactions and dependencies between \emph{pairs} of participants.

\begin{definition}[Global and Relative Types]\label{d:globalRelativeTypes}
    \begin{align*}
        & \text{Global types}
        & \gtc{G},\gtc{G'} &
        \begin{array}[t]{@{}l@{}lr@{}lr@{}}
            {} ::= {}&
            p \send q \gtBraces{ \msg i<T_i> \gtc. \gtc{G} }_{i \in I} & \text{(exchange)}
            \sepr* &
            \pend & \text{(end)}
            \\ \sepr* &
            \rec X \gtc. \gtc{G} \sepr X & \text{(recursion)}
            \hphantom{\sepr*}
        \end{array}
        \\
        & \text{Relative types}\!
        & \rtc{R},\rtc{R'} &
        \begin{array}[t]{@{}l@{}lr@{}lr@{}}
            {} ::= {}&
            p \send q \rtBraces{ \msg i<T_i> \rtc. \rtc{R} }_{i \in I} & \text{(exchange)}
            \!\!\sepr* &
            \pend & \text{(end)}
            \\ \sepr* &
            (p \send r) \send q \rtBraces{ i \rtc. \rtc{R} }_{i \in I} & \text{(output dep.)}
            \!\!\sepr* &
            \rec X \rtc. \rtc{R} \sepr X & \text{(recursion)}
            \\ \sepr* &
            (p \recv r) \send q \rtBraces{ i \rtc. \rtc{R} }_{i \in I} & \text{(input dep.)}
            \hphantom{\!\!\sepr*}
        \end{array}
    \end{align*}

    \noindent
    We write $\prt(\gtc{G})$ to denote the set of participants involved in exchanges in $\gtc{G}$.
\end{definition}

\noindent
The global type $p \send q \gtBraces{ \msg i<T_i> \gtc. \gtc{G_i} }_{i \in I}$ specifies that $p$ sends to $q$ some $j \in I$ with $T_j$, continuing as $\gtc{G_j}$.
A relative type specifies a protocol between a pair of participants, say $p$ and $q$.
The type $p \send q \rtBraces{ \msg i<T_i> \rtc. \rtc{R_i} }_{i \in I}$ specifies that $p$ sends to $q$ some $j \in I$ with $T_j$, continuing as $\rtc{R_j}$.
If the protocol between $p$ and $q$ depends on a prior choice involving $p$ or $q$, their relative type includes a \emph{dependency}: $(p \send r) \send q \rtBraces{ i \rtc. \rtc{R_i} }_{i \in I}$ (resp.\ $(p \recv r) \send q \rtBraces{ i \rtc. \rtc{R_i} }_{i \in I}$) specifies that $p$ forwards to $q$ the $j \in I$ sent to (resp.\ received from) $r$ by $p$.
For both global and relative types, tail-recursion is defined with recursive definitions $\rec X$ and recursive calls $X$, and $\pend$ specifies the end of the protocol.

Relative types are obtained from global types by means of projection:

\begin{definition}[Relative Projection]\label{d:relativeProjection}
    The \emph{relative projection} of a global type onto a pair of participants, denoted $\gtc{G} \wrt (p,q)$, is defined by \cref{alg:relativeProjection}.
\end{definition}

\setlength{\textfloatsep}{\algfloatsep}
\begin{algorithm}[t!]
    \DontPrintSemicolon
    \SetAlgoNoEnd
    \SetInd{.2em}{.7em}
    \SetSideCommentRight
    \Def{$\gtc{G} \wrt (p,q)$}{

        \Switch{$\gtc{G}$}{

            \uCase{$s \send r \gtBraces{ \msg i<T_i> \gtc. \gtc{G_i} }_{i \in I}$}{
                $\forall i \in I.~ \rtc{R_i} := \gtc{G_i} \wrt (p,q)$ \;

                \lIf{$(p = s \wedge q = r)$}{
                    \KwRet
                    $p \send q \rtBraces{ \msg i<T_i> \rtc. \rtc{R_i} }_{i \in I}$
                }\label{li:rpEx1}

                \lElseIf{$(q = s \wedge p = r)$}{
                    \KwRet
                    $q \send p \rtBraces{ \msg i<T_i> \rtc. \rtc{R_i} }_{i \in I}$
                }\label{li:rpEx2}

                \lElseIf{$\forall i,j \in I.~ \rtc{R_i} = \rtc{R_j}$}{
                    \KwRet
                    $\bigcup_{i \in I} \rtc{R_i}$
                }\label{li:rpIndep}

                \lElseIf{$s \in \braces{p,q} \wedge t \in \braces{p,q} \setminus \braces{s}$}{
                    \KwRet
                    $(s \send r) \send t \rtBraces{ i \rtc. \rtc{R_i} }_{i \in I}$
                }\label{li:rpDep1}

                \lElseIf{$r \in \braces{p,q} \wedge t \in \braces{p,q} \setminus \braces{r}$}{
                    \KwRet
                    $(r \recv s) \send t \rtBraces{ i \rtc. \rtc{R_i} }_{i \in I}$
                }\label{li:rpDep2}
            }

            \uCase{$\rec X \gtc. \gtc{G'}$}{
                $\rtc{R'} := \gtc{G'} \wrt (p,q)$ \;

                \lIf{$(\text{$\rtc{R'}$ contains an exchange or a recursive call on any $Y \neq X$})$}{
                    \KwRet
                    $\rec X \rtc. \rtc{R'}$
                }\label{li:rpRec}

                \lElse{
                    \KwRet
                    $\pend$
                }\label{li:rpRecEnd}
            }

            \lCase{$X$}{
                \KwRet
                $X$
            }\label{li:rpCall}

            \lCase{$\pend$}{
                \KwRet
                $\pend$
            }\label{li:rpEnd}
        }

    }

    \caption{\nameref{d:relativeProjection} of $\gtc{G}$ onto $p$ and $q$ (Def.~\labelcref{d:relativeProjection}).}
    \label{alg:relativeProjection}
\end{algorithm}

\noindent
The projection of an exchange onto $(p,q)$ is an exchange if $p$ and $q$ are sender and recipient (\cref{li:rpEx1,li:rpEx2}).
Otherwise, if the protocol between $p$ and $q$ does not depend on the exchange (the projections of all branches are equal), the projection is the union of the projected branches (\cref{li:rpIndep}).
The union of relative types, denoted $\rtc{R} \cup \rtc{R'}$, is defined only on identical relative types (e.g., $p \send q \rtBraces{ \msg i<T_i> \rtc. \rtc{R_i} }_{i \in I} \cup p \send q \rtBraces{ \msg i<T_i> \rtc. \rtc{R_i} }_{i \in I} = p \send q \rtBraces{ \msg i<T_i> \rtc. \rtc{R_i} }_{i \in I}$; see \ifappendix \cref{a:relativeTypesWithLocs} \else \cite{report/vdHeuvelPD23} \fi for a formal definition).
If there is a dependency and $p$ or $q$ is sender/recipient, the projection is a dependency (\cref{li:rpDep1,li:rpDep2}).
Projection is undefined if there is a dependency but $p$ nor $q$ is involved.

The projection of $\rec X \gtc. \gtc{G'}$ is a relative type starting with a recursive definition, provided that the projection of $\gtc{G'}$ onto $(p,q)$ contains an exchange or nested recursive call (\cref{li:rpRec}) to avoid recursion with only dependencies;
otherwise, the projection returns $\pend$ (\cref{li:rpRecEnd}).
The projections of recursive calls and $\pend$ are homomorphic (\cref{li:rpCall,li:rpEnd}).

\begin{example}\label{ex:cRelProj}
    The relative projections of $\gtc{G_{\mathsf{a}}}$~\eqref{eq:auth} are:
    \begin{align*}
        \rtc{R_{c,s}} &:= \gtc{G_{\mathsf{a}}} \wrt (c,s) = \rec X \rtc. s \send c \rtBraces{ \msg \mathsf{login}<> \rtc. X \rtc, \msg \mathsf{quit}<> \rtc. \pend }
        \\
        \rtc{R_{c,a}} &:= \gtc{G_{\mathsf{a}}} \wrt (c,a) = \rec X \rtc. (c \recv s) \send a \rtBraces{ \mathsf{login} \rtc. c \send a \rtBraces{ \msg \mathsf{pwd}<\mathsf{str}> \rtc. X } \rtc, \mathsf{quit} \rtc. \pend }
        \\
        \rtc{R_{s,a}} &:= \gtc{G_{\mathsf{a}}} \wrt (s,a) = \rec X \rtc. (s \send c) \send a \rtBraces{ \mathsf{login} \rtc. a \send s \rtBraces{ \msg \mathsf{succ}<\mathsf{bool}> \rtc. X } \rtc, \mathsf{quit} \rtc. \pend }
    \end{align*}
    Hence, the exchange from $s$ to $c$ is a dependency for the protocols of $a$.
\end{example}

Not all global types are sensible.
A valid global type may, e.g., require a participant $p$ to have different behaviors, depending on a choice that $p$ is unaware of (see, e.g.,~\cite{journal/lmcs/CastagnaDP12}).
In the following, we work only with \emph{well-formed} global types:

\begin{definition}[Well-formedness]\label{d:wf}
    We say a global type $\gtc{G}$ is \emph{well-formed} if and only if,
    for all pairs of participants $p \neq q \in \prt(\gtc{G})$, the projection $\gtc{G} \wrt (p,q)$ is defined,
    and
    all recursion in $\gtc{G}$ is non-contractive (e.g., $\gtc{G} \neq \rec X \gtc. X$) and bound.
\end{definition}

\noindent
Our running example $\gtc{G_{\mathsf{a}}}$~\eqref{eq:auth} is well-formed in the above sense; also, as explained in~\cite{conf/popl/ScalasY19}, $\gtc{G_{\mathsf{a}}}$ is \emph{not} well-formed in most theories of multiparty sessions (based on projection onto local types).
As such, $\gtc{G_{\mathsf{a}}}$ goes beyond the scope of such theories.

\paragraph{Synthesizing monitors.}
Next, we define a procedure to synthesize monitors for the participants of global types. 
This procedure detects dependencies as follows:

\begin{definition}[Dependence]\label{d:depsOn}
    Given a global type $\gtc{G}$, we say $p$ \emph{depends on} $q$ in $\gtc{G}$, denoted $\depsOn p q \gtc{G}$, if and only if
    \\
    $
            \gtc{G} = s \send r \gtBraces{ \msg i<T_i> \gtc. \gtc{G_i} }_{i \in I} \wedge p \notin \braces{s,r} \wedge q \in \braces{s,r}
            \wedge \exists i,j \in I.~ \gtc{G_i} \wrt (p,q) \neq \gtc{G_j} \wrt (p,q)
    $.
\end{definition}

\noindent
Thus, $\depsOn p q \gtc{G}$ holds if and only if $\gtc{G}$ is an exchange involving $q$ but not $p$, and the relative projections of at least two branches of the exchange are different.

\begin{definition}[Synthesis of Monitors from Global Types]\label{d:gtToMon}
 \cref{alg:gtToMon} synthesizes the monitor for $p$ in $\gtc{G}$ with participants $D$, denoted $\gtToMon( \gtc{G} , p , D )$.
\end{definition}

\setlength{\textfloatsep}{\algfloatsep}
\begin{algorithm}[t]
    \DontPrintSemicolon
    \SetAlgoNoEnd
    \SetInd{.2em}{.6em}
    \SetSideCommentRight

    \Def{$\gtToMon(\gtc{G},p,D)$}{
        \Switch{$\gtc{G}$}{

            \uCase{$s \send r \gtBraces{ \msg i<T_i> \gtc. \gtc{G_i} }_{i \in I}$}{
                $\depsVar := \braces{ q \in D \mid \depsOn q p \gtc{G} }$

                \lIf{$p = s$}{
                    \KwRet
                    $p \send r \mBraces{ \msg i<T_i> \mc. p \send \depsVar (i) \mc. \gtToMon(\gtc{G_i},p,D) }_{i \in I}$
                }\label{li:monEx1}

                \lElseIf{$p = r$}{
                    \KwRet
                    $p \recv s \mBraces{ \msg i<T_i> \mc. p \send \depsVar (i) \mc. \gtToMon(\gtc{G_i},p,D) }_{i \in I}$
                }\label{li:monEx2}

                \uElseIf{$p \notin \braces{r,s}$}{
                    $\depOnVar_s := (s \in D \wedge \depsOn p s \gtc{G})$ \;
                    $\depOnVar_r := (r \in D \wedge \depsOn p r \gtc{G})$

                    \lIf{$(\depOnVar_s \wedge \neg \depOnVar_r)$}{
                        \KwRet
                        $p \recv s \mBraces{ i \mc. \gtToMon(\gtc{G_i},p,D) }_{i \in I}$
                    }\label{li:monDep1}

                    \lElseIf{$(\depOnVar_r \wedge \neg \depOnVar_s)$}{
                        \KwRet
                        $p \recv r \mBraces{ i \mc. \gtToMon(\gtc{G_i},p,D) }_{i \in I}$
                    }\label{li:monDep2}

                    \uElseIf{$(\depOnVar_s \wedge \depOnVar_r)$}{
                        \KwRet
                        $p \recv s \mBraces[\big]{ i \mc. p \recv r \mBraces{ i \mc. \gtToMon(\gtc{G_i},p,D) } \cup \mBraces{ j \mc. \perror }_{j \in I \setminus \braces{i}} }_{i \in I}$
                        \label{li:monDep3}
                    }

                    \lElse({(arbitrary $k \in I$)}){
                        \KwRet
                        $\gtToMon(\gtc{G_k},p,D)$
                    }\label{li:monNoDep}
                }
            }

            \uCase{$\rec X \gtc. \gtc{G'}$}{
                $D' := \braces{ q \in D \mid \gtc{G} \wrt (p,q) \neq \pend }$

                \lIf{$D' \neq \emptyset$}{
                    \KwRet
                    $\rec X \mc. \gtToMon(\gtc{G'},p,D')$
                }\label{li:monRec}

                \lElse{
                    \KwRet
                    $\pend$
                }\label{li:monRecEnd}
            }

            \lCase{$X$}{
                \KwRet
                $X$
            }\label{li:monCall}

            \lCase{$\pend$}{
                \KwRet
                $\pend$
            }\label{li:monEnd}

        }
    }

    \caption{\nameref{d:gtToMon} (\cref{d:gtToMon}).}
    \label{alg:gtToMon}
\end{algorithm}

\noindent
Initially, $D = \prt(\gtc{G}) \setminus \braces{p}$.
The monitor for $p$ of an exchange where $p$ is sender (resp.\ recipient) is an output (resp.\ input) followed in each branch by a dependency output, using \nameref{d:depsOn} to compute the participants with dependencies (\cref{li:monEx1,li:monEx2}).
If $p$ is not involved, we detect a dependency for $p$ with \nameref{d:depsOn}.
In case $p$ depends on sender/recipient but not both, the monitor is a dependency input (\cref{li:monDep1,li:monDep2}).
If $p$ depends on  sender \emph{and} recipient, the monitor contains two consecutive dependency inputs (\cref{li:monDep3}); when the two received labels differ, the monitor enters an $\perror$-state.
When there is no dependency for $p$, the monitor uses an arbitrary branch (\cref{li:monNoDep}).
To synthesize a monitor for $\rec X \gtc. \gtc{G'}$, the algorithm uses projection to compute $D'$ with participants having exchanges with $p$ in $\gtc{G'}$ (cf.\ \cref{alg:relativeProjection} \cref{li:rpRec}).
If $D'$ is non-empty, the monitor starts with a recursive definition (\cref{li:monRec}) and the algorithm continues with $D'$; otherwise, the monitor is $\pend$ (\cref{li:monRecEnd}).
The monitors of $X$ and $\pend$ are homomorphic (\cref{li:monCall,li:monEnd}).

\begin{example}
    Let us use
    $\gtc{G} = p \send q \gtBraces{ \msg \ell<T> \gtc. \rec X \gtc. p \send r \gtBraces{ \msg \ell'<T'> \gtc. X \gtc, \msg \ell''<T''> \gtc. \pend } }$
    to illustrate \cref{alg:gtToMon}. 
    We have $\gtc{G} \wrt (p,q) = p \send q \rtBraces{ \msg \ell<T> \rtc. \pend }$: the projection of the recursive body in $\gtc{G}$ is $(p \send r) \send q \rtBraces{ \msg \ell'<T'> \rtc. X \rtc, \msg \ell''<T''> \rtc. \pend }$, but there are no exchanges between $p$ and $q$, so the projection of the recursive definition is $\pend$.
    Were the monitor for $p$ synthesized with $q \in D$, \nameref{d:depsOn} would detect a dependency: the recursive definition's monitor would be $p \send r \mBraces{ \msg \ell'<T'> \mc. p \send \braces{q} (\ell') \mc. X \mc, \msg \ell''<T''> \mc. p \send \braces{q} (\ell'') \mc. \pend }$.
    However, per $\gtc{G} \wrt (p,q)$, $p$ nor $q$ expects a dependency at this point of the protocol.
    Hence, the algorithm removes $q$ from $D$ when entering the recursive body in $\gtc{G}$.
\end{example}

\begin{example}\label{ex:cMon}
    The monitors of $c$, $s$, $a$ in $\gtc{G_{\mathsf{a}}}$~\eqref{eq:auth} are:
    \begin{align*}
        \mc{M_c} &:= \gtToMon(\gtc{G_{\mathsf{a}}},c,\braces{s,a})
        \\
        &=
        \rec X \mc. c \recv s \mBraces*{
            \begin{array}{@{}l@{}}
                \msg \mathsf{login}<> \mc. c \send \braces{a} (\mathsf{login}) \mc. c \send a \mBraces{ \msg \mathsf{pwd}<\mathsf{str}> \mc. c \send \emptyset (\mathsf{pwd}) \mc. X } \mc,
                \\
                \msg \mathsf{quit}<> \mc. c \send \braces{a} (\mathsf{quit}) \mc. \pend
            \end{array}
        }
        \\
        \mc{M_s} &:= \gtToMon(\gtc{G_{\mathsf{a}}},s,\braces{c,a})
        \\
        &=
        \rec X \mc. s \send c \mBraces*{
            \begin{array}{@{}l@{}}
                \msg \mathsf{login}<> \mc. s \send \braces{a} (\mathsf{login}) \mc. s \recv a \mBraces{ \msg \mathsf{succ}<\mathsf{bool}> \mc. s \send \emptyset (\mathsf{succ}) \mc. X } \mc,
                \\
                \msg \mathsf{quit}<> \mc. s \send \braces{a} (\mathsf{quit}) \mc. \pend
            \end{array}
        }
        \\
        \mc{M_a} &:= \gtToMon(\gtc{G_{\mathsf{a}}},a,\braces{c,s})
        \\
        &=
        \rec X \mc. a \recv s \mBraces*{
            \begin{array}{@{}l@{}}
                \mathsf{login} \mc. a \recv c \mBraces*{
                    \begin{array}{@{}l@{}}
                        \mathsf{login} \mc. \begin{array}[t]{@{}l@{}}
                            a \recv c
                            \color{mClr}\lbrace\mskip-5mu\lbrace\normalcolor
                            \msg \mathsf{pwd}<\mathsf{str}> \mc. a \send \emptyset (\mathsf{pwd}) \mc.
                            \\
                            a \send s \mBraces{ \msg \mathsf{succ}<\mathsf{bool}> \mc. a \send \emptyset (\mathsf{succ}) \mc. X }
                            \color{mClr}\rbrace\mskip-5mu\rbrace\normalcolor
                            \mc,
                        \end{array}
                        \\
                        \mathsf{quit} \mc. \perror
                    \end{array}
                } \mc,
                \\
                \mathsf{quit} \mc. a \recv c \mBraces{ \mathsf{quit} \mc. \pend \mc, \mathsf{login} \mc. \perror }
            \end{array}
        }
    \end{align*}
\end{example}

\section{Properties of Correct Monitored Blackboxes}
\label{s:correctMonitoredBlackboxes}

Given a global type $\gtc{G}$, we establish the precise conditions under which a network of monitored blackboxes correctly implements $\gtc{G}$.
That is, we define how the monitored blackbox $\net P$ of a participant $p$ of $\gtc{G}$ should behave, i.e., when $\net P$ \emph{satisfies} the role of $p$ in $\gtc{G}$ (\nameref{d:satisfaction}, \cref{d:satisfaction}).
We then prove two important properties of networks of monitored blackboxes that satisfy a given global type:
\begin{itemize}
    \item[\textbf{Soundness:}]
        The network behaves correctly according to the global type (\cref{t:soundnessMain});

    \item[\textbf{Transparency:}]
        The monitors interfere minimally with buffered blackboxes (\cref{t:transparencyMain}).
\end{itemize}

\noindent
As we will see in \cref{s:soundness}, satisfaction is exactly the condition under which a network $\net P$ is sound with respect to a global type $\gtc{G}$.

\subsection{Satisfaction}
\label{s:satisfaction}

Our aim is to attest that $\net P$ satisfies the role of $p$ in $\gtc{G}$ if it meets certain conditions on the behavior of monitored blackboxes with respect to the protocol.
As we have seen, the role of $p$ in $\gtc{G}$ is determined by  projection.
Satisfaction is then a relation $\mathcal{R}$ between (i)~monitored blackboxes and (ii)~maps from participants $q \in \prt(\gtc{G}) \setminus \braces{p}$ to relative types between $p$ and $q$, denoted $\RTs$; $\mathcal{R}$ must contain $(\net P,\RTs)$ with relative projections of $\gtc{G}$.
Given any $(\net P',\RTs')$ in $\mathcal{R}$, the general idea of satisfaction is (i)~that an output to $q$ by $\net P'$ means that $\RTs'(q)$ is a corresponding exchange from $p$ to $q$, and (ii)~that if there is a $q$ such that $\RTs'(q)$ is an exchange from $q$ to $p$ then $\net P'$ behaves correctly afterwards.

In satisfaction, dependencies in relative types require care.
For example, if $\RTs'(q)$ is an exchange from $p$ to $q$ and $\RTs'(r)$ is a dependency on this exchange, then $\net P'$ must first send a label to $q$ and then send \emph{the same label} to $r$.
Hence, we need to track the labels chosen by the monitored blackbox for later reference.
To this end, we uniquely identify each exchange in a global type by its \emph{location} $\vect \ell$: a sequence of labels denoting the choices leading to the exchange.
Projection then uses these locations to annotate each exchange and recursive definition/call in the relative type it produces.
Because projection skips independent exchanges (\cref{alg:relativeProjection}, \cref{li:rpIndep}), some exchanges and recursive definitions/calls may be associated with multiple locations; hence, they are annotated with \emph{sets} of locations, denoted $\mbb L$.
Satisfaction then tracks choices using a map from sets of locations to labels, denoted $\Lbls$.
Projection with location annotations is formally defined in~\ifappendix\cref{a:relativeTypesWithLocs}\else\cite{report/vdHeuvelPD23}\fi, along with a corresponding definition for unfolding recursion.

Before defining satisfaction, we  set  up some useful notation for type signatures, locations, relative types, and maps.

\begin{notation}
    Let $\bm{P}$ denote the set of all participants, $\bm{R}$ the set of all relative types, $\bm{N}$ the set of all networks, and $\bm{L}$ the set of all labels.

    Notation $\Pow(S)$ denotes the powerset of $S$.
    Given a set $S$, we write $\vect S$ to denote the set of all sequences of elements from $S$.
    We write $\mbb L \overlap \mbb L'$ to stand for $\mbb L \cap \mbb L' \neq \emptyset$.
    We write $\mbb L \leq \mbb L'$ if every $\vect \ell' \in \mbb L'$ is prefixed by some $\vect \ell \in \mbb L$.

    In relative types, we write $\lozenge$ to denote either $\send$ or $\recv$.
    We write $\unfold(\rtc{R})$ for the inductive unfolding of $\rtc{R}$ if $\rtc{R}$ starts with recursive definitions, and for $\rtc{R}$ itself otherwise.
    We write $\rtc{R} \unfoldeq \rtc{R'}$ whenever $\unfold(\rtc{R}) = \unfold(\rtc{R'})$.

    We shall use monospaced fonts to denote maps (such as $\RTs$ and $\Lbls$).
    We often define maps using the notation of injective relations.
    Given a map $\mathtt{M}$, we write $(x,y) \in \mathtt{M}$ to denote that $x \in \dom(\mathtt{M})$ and $\mathtt{M}(x) = y$.
    We write $\mathtt{M} \update{x \mapsto y'}$ to denote the map obtained by adding to $\mathtt{M}$ an entry for $x$ pointing to $y'$, or updating an already existing entry for $x$.
    Maps are partial unless stated otherwise.
\end{notation}

\begin{definition}[Satisfaction]\label{d:satisfaction}
    A relation $\mathcal{R}$ is \emph{sat-signed} if its signature is $\bm{N} \times (\bm{P} \rightarrow \bm{R}) \times (\Pow(\vect {\bm{L}}) \rightarrow \bm{L})$.
    We define the following properties of relations:
    \begin{itemize}
        \item
            A sat-signed relation $\mathcal{R}$ \emph{holds at $p$} if it satisfies the conditions in \cref{f:satisfaction}.

        \item
            A sat-signed relation $\mathcal{R}$  \emph{progresses at $p$} if for every $( \net P , \RTs , \Lbls ) \in \mathcal{R}$, we have $\net P \ltrans{\alpha} \net P'$ for some $\alpha$ and $\net P'$, given that one of the following holds:
            \begin{itemize}
                \item
                    $\RTs \neq \emptyset$ and, for every $(q,\rtc{R}) \in \RTs$, $\rtc{R} \unfoldeq \pend$;

                \item
                    There is $(q,\rtc{R}) \in \RTs$ such that (i)~$\rtc{R} \unfoldeq p \send q^{\mbb L} \rtBraces{ \msg i<T_i> \rtc. \rtc{R_i} }_{i \in I}$ or $\rtc{R} \unfoldeq (p \lozenge r) \send q^{\mbb L} \rtBraces{ i \rtc. \rtc{R_i} }_{i \in I}$, and (ii)~for every $(q',\rtc{R'}) \in \RTs \setminus \braces{(q,\rtc{R})}$, either
                    $\rtc{R'} \unfoldeq \pend$
                    or $\unfold(\rtc{R'})$ has locations $\mbb L'$ with $\mbb L \leq \mbb L'$.
            \end{itemize}

        \item
            A sat-signed relation $\mathcal{R}$ is a \emph{satisfaction at $p$} if it holds and progresses at $p$.
    \end{itemize}
    We write $\mathcal{R} \satisfies{ \net P }[ \Lbls ]{ \RTs }[ p ]$ if $\mathcal{R}$ is a satisfaction at $p$ with $(\net P , \RTs , \Lbls) \in \mathcal{R}$, and
    $ \mathcal{R} \satisfies{ \net P }{ \RTs }[ p ]$ when $\Lbls$ is empty.
    We omit $\mathcal{R}$ to indicate such $\mathcal{R}$ exists.
\end{definition}

\setlength{\textfloatsep}{\textfloatsepsave}
\begin{figure}[t]
    Given $( \net P , \RTs , \Lbls ) \in \mathcal{R}$, all the following conditions hold:
    \begin{enumerate}[itemsep=2.0pt]
        \item\label{i:sTau}
            \textbf{(Tau)}
            If $\net P \ltrans{ \tau } \net P'$,
            then $( \net P' , \RTs , \Lbls ) \in \mathcal{R}$.

        \item\label{i:sEnd}
            \textbf{(End)}
            If $\net P \ltrans{ \pend } \net P'$,
            then, for every $(q,\rtc{R}) \in \RTs$,
            $\rtc{R} \unfoldeq \pend$,
            $\net P' \ntrans$, and $( \net P' , \emptyset , \emptyset ) \in \mathcal{R}$.

        \item\label{i:sOutput}
            \textbf{(Output)}
            If $\net P \ltrans{ p \send q (\msg j<T_j>) } \net P'$,
            then
            $\RTs(q)\unfoldeq p \send q^{\mbb L} \rtBraces{ \msg i<T_i> \rtc. \rtc{R_i} }_{i \in I}$
            with $j \in I$, and $( \net P' , \RTs \update{q \mapsto \rtc{R_j}} , \Lbls \update{\mbb L \mapsto j} ) \in \mathcal{R}$.

        \item\label{i:sInput}
            \textbf{(Input)}
            If there is $(q,\rtc{R}) \in \RTs$ such that
            $\rtc{R} \unfoldeq q \send p^{\mbb L} \rtBraces{ \msg i<T_i> \rtc. \rtc{R_i} }_{i \in I}$,
            then
            \\
            $\net P = \monImpl{ \bufImpl{ p }{ P }{ \vect m } }{ \mc{M} }{ \vect n }$,
            and, for every $j \in I$,
            \\
            $( \monImpl{ \bufImpl{ p }{ P }{ \vect m } }{ \mc{M} }{ q \send p (\msg j<T_j>) , \vect n } , \RTs \update{q \mapsto \rtc{R_j}} , \Lbls \update{\mbb L \mapsto j} ) \in \mathcal{R}$.

        \item\label{i:sDependencyOutput}
            \textbf{(Dependency output\kern-1pt)}
            If $\net P \mkern-3mu \ltrans{ p \send q \Parens{j} } \mkern-3mu \net P'$,
            then
            $\RTs(q) \unfoldeq (p \lozenge r) \send q^{\mbb L} \rtBraces{ i \rtc. \rtc{R_i} }_{i \in I}$
            with $j \in I$,
            there is $(\mbb L',j) \in \Lbls$ such that $\mbb L' \overlap \mbb L$,
            and $( \net P' , \RTs \update{q \mapsto \rtc{R_j}} , \Lbls ) \in \mathcal{R}$.

        \item\label{i:sDependencyInput}
            \textbf{(Dependency input)}
            If there is $(q,\rtc{R}) \in \RTs$ s.t.\
            $\rtc{R} \unfoldeq (q \lozenge r) \send p^{\mbb L} \rtBraces{ i \rtc. \rtc{R_i} }_{i \in I}$,
            then $\net P = \monImpl{ \bufImpl{ p }{ P }{ \vect m } }{ \mc{M} }{ \vect n }$, and either of the following holds:
            \begin{itemize}
                \item
                    (Fresh label)
                    there is no $\mbb L' \in \dom(\Lbls)$ such that $\mbb L' \overlap \mbb L$, and, for every $j \in I$, $( \monImpl{ \bufImpl{ p }{ P }{ \vect m } }{ \mc{M} }{ q \send p \Parens{j} , \vect n } , \RTs \update{q \mapsto \rtc{R_j}} , \Lbls \update{\mbb L \mapsto j} ) \in \mathcal{R}$;

                \item
                    (Known label)
                    there is $(\mbb L',j) \in \Lbls$ such that $\mbb L' \overlap \mbb L$ and $j \in I$, and
                    \\
                    $( \monImpl{ \bufImpl{ p }{ P }{ \vect m } }{ \mc{M} }{ q \send p \Parens{j} , \vect n } , \RTs \update{q \mapsto \rtc{R_j}} , \Lbls ) \in \mathcal{R}$.
            \end{itemize}
    \end{enumerate}
    \caption{\nameref{d:satisfaction}: conditions under which $\mathcal{R}$ holds at $p$ (\cref{d:satisfaction}).}
    \label{f:satisfaction}
\end{figure}

\noindent
Satisfaction requires $\mathcal{R}$ to hold at $p$: each $(\net P,\RTs,\Lbls) \in \mathcal{R}$ enjoys the conditions in \cref{f:satisfaction}, discussed next, which ensure that $\net P$ respects the protocols in $\RTs$.

\satclause{i:sTau}{Tau} allows $\tau$-transitions without affecting $\RTs$ and $\Lbls$.
\satclause{i:sEnd}{End} allows an $\pend$-transition, given that all relative types in $\RTs$ are $\pend$.
The resulting state should not transition, enforced by empty $\RTs$ and $\Lbls$.

\satclause{i:sOutput}{Output} allows an output-transition with a message to $q$, given that $\RTs(q)$ is a corresponding output by $p$.
Then, $\RTs(q)$ updates to the continuation of the appropriate branch and $\Lbls$ records the choice under the locations of $\RTs(q)$.

\satclause{i:sInput}{Input} triggers when there is $(q,\rtc{R}) \in \RTs$ such that $\rtc{R}$ is a message from $q$ to~$p$.
Satisfaction targets the behavior of $\net P$ on its own, so we simulate a message sent by $q$.
The resulting behavior is then analyzed after buffering any such message;  $\RTs(q)$ is updated to the continuation of the corresponding branch.
As for outputs, $\Lbls$ records the choice at the locations of $\RTs(q)$.

\satclause{i:sDependencyOutput}{Dependency Output} allows an output-transition with a dependency message to $q$, given that $\RTs(q)$ is a corresponding dependency output by $p$ with locations~$\mbb L$.
The message's label should be recorded in $\Lbls$ at some $\mbb L'$ that shares a location with $\mbb L$: here $\mbb L'$ relates to a past exchange between $p$ and some $r$ in $\gtc{G}$ from which the dependency output in $\RTs(q)$ originates.
This ensures that the dependency output is preceded by a corresponding exchange, and that the dependency output carries the same label as originally chosen for the preceding exchange.
Afterwards, $\RTs(q)$ is updated to the continuation of the appropriate branch.

\satclause{i:sDependencyInput}{Dependency Input} triggers when there is $(q,\rtc{R}) \in \RTs$ such that $\rtc{R}$ is a dependency exchange from $q$ to $p$, forwarding a label exchanged between $q$ and~$r$.
As in the input case, a message from $q$ is simulated by buffering it in $\net P$.
In this case, $\RTs(r)$ could be a dependency exchange from $r$ to $p$, originating from the same exchange between $q$ and $r$ in $\gtc{G}$.
To ensure that the buffered messages contain the same label, we distinguish ``fresh'' and ``known'' cases.
In the fresh case, we consider the first of the possibly two dependency exchanges: there is no $\mbb L' \in \dom(\Lbls)$ that shares a location with the locations $\mbb L$ of $\RTs(q)$.
Hence, we analyze each possible dependency message, updating $\RTs(q)$ appropriately and recording the choice in $\Lbls$.
The known case then considers the second dependency exchange: there is a label in $\Lbls$ at $\mbb L'$ that shares a location with $\mbb L$.
Hence, we buffer a message with the same label, and update $\RTs(q)$ accordingly.

Satisfaction also requires $\mathcal{R}$ to progress at $p$, for each $(\net P,\RTs,\Lbls) \in \mathcal{R}$ making sure that $\net P$ does not idle whenever we are expecting a  transition from~$\net P$.
There are two cases.
(1)~If all relative types in $\RTs$ are $\pend$, we expect an $\pend$-transition.
(2)~If there is a relative type in $\RTs$ that is a (dependency) output,
we expect an output transition.
However, $\net P$ may idle if it is waiting for a message: there is $(q,\rtc{R}) \in \RTs$ such that $\rtc{R}$ is a (dependency) input originating from an exchange in $\gtc{G}$ that precedes the exchange related to the output.

\begin{definition}[Satisfaction for Networks]\label{d:monSat}
    Let us write $\RTsOf(\gtc{G},p)$ to denote the set $\braces{ ( q , \gtc{G} \wrt (p,q) ) \mid q \in \prt(\gtc{G}) \setminus \braces{p} }$.
    Moreover, we write
    \begin{itemize}

        \item
            $\mathcal{R} \satisfies{ \monImpl{ \bufImpl{ p }{ P }{ \epsi } }{ \mc{M} }{ \epsi } }{ \gtc{G} }[ p ]$ if and only if $\mc{M} = \gtToMon( \gtc{G} , p , \prt(\gtc{G}) \setminus \braces{p} )$ and $\mathcal{R} \satisfies{ \monImpl{ \bufImpl{ p }{ P }{ \epsi } }{ \mc{M} }{ \epsi } }{ \RTsOf(\gtc{G},p) }[ p ]$.
            We omit $\mathcal{R}$ to say such $\mathcal{R}$ exists.

        \item
            $\satisfies{ \net P }{ \gtc{G} }$ if and only if $\net P \equiv \prod_{p \in \prt(\gtc{G})} \monImpl{ \bufImpl{ p }{ P_p }{ \epsi } }{ \mc{M_p} }{ \epsi }$ and, for every $p \in \prt(\gtc{G})$, $\satisfies{ \monImpl{ \bufImpl{ p }{ P_p }{ \epsi } }{ \mc{M_p} }{ \epsi } }{ \gtc{G} }[ p ]$.

    \end{itemize}
\end{definition}

\begin{example}
    The following satisfaction assertions hold with implementations, relative types, and monitors from \Cref{ex:cImpl,ex:cRelProj,ex:cMon}, respectively:
    \begin{align*}
        & \satisfies{ \monImpl{ \bufImpl{ c }{ P_c }{ \epsi } }{ \mc{M_c} }{ \epsi } }{ \braces{ ( s , \rtc{R_{c,s}} ) , ( a , \rtc{R_{c,a}} ) } }[ c ]
        \\
        & \satisfies{ \monImpl{ \bufImpl{ s }{ P_s }{ \epsi } }{ \mc{M_s} }{ \epsi } }{ \braces{ ( c , \rtc{R_{c,s}} ) , ( a , \rtc{R_{s,a}} ) } }[ s ]
        \\
        & \satisfies{ \monImpl{ \bufImpl{ a }{ P_a }{ \epsi } }{ \mc{M_a} }{ \epsi } }{ \braces{ ( c , \rtc{R_{c,a}} ) , ( s , \rtc{R_{s,a}} ) } }[ a ]
        \\
        & \satisfies{ \monImpl{ \bufImpl{ c }{ P_c }{ \epsi } }{ \mc{M_c} }{ \epsi } \| \monImpl{ \bufImpl{ s }{ P_s }{ \epsi } }{ \mc{M_s} }{ \epsi } \| \monImpl{ \bufImpl{ a }{ P_a }{ \epsi } }{ \mc{M_a} }{ \epsi } }{ \gtc{G_{\mathsf{a}}} }
    \end{align*}
    We also have:
    $\satisfies[\nvDash]{ \monImpl{ \bufImpl{ c }{ P_c }{ \epsi } }{ \rec X \mc. c \recv s \mBraces{ \msg \mathsf{quit}<> \mc. \pend } }{ \epsi } }{ \gtc{G_{\mathsf{a}}} }[ c ]$. This is because
    $\rtc{R_{c,s}}$ specifies that $s$ may send $\mathsf{login}$ to $c$, which this monitor would not accept.
\end{example}

\subsection{Soundness}
\label{s:soundness}

Our \emph{first result} is that satisfaction is sound with respect to global types: when a network of monitored blackboxes satisfies a global type $\gtc{G}$ (\cref{d:monSat}), any path of transitions eventually reaches a state that satisfies another global type reachable from $\gtc{G}$.
Hence, the satisfaction of the individual components that a network comprises is enough to ensure that the network behaves as specified by the global type.
Reachability between global types is defined as an LTS:

\begin{definition}[LTS for Global Types]\label{d:ltsGlobalTypes}
    We define an LTS for global types, denoted $\gtc{G} \ltrans{ \ell } \gtc{G'}$, by the following rules:

    \begin{mathparpagebreakable}
        \bussAssume{
            j \in I
        }
        \bussUn{
            p \send q \gtBraces{ \msg i<T_i> \gtc. \gtc{G_i} }_{i \in I}
            \mathrel{\raisebox{-1.5pt}[0ex][0ex]{$\ltrans{ j }$}}
            \gtc{G_j}
        }
        \bussDisplay
        \and
        \bussAssume{
            \gtc{G} \braces{ \rec X \gtc. \gtc{G} / X }
            \ltrans{ \ell }
            \gtc{G'}
        }
        \bussUn{
            \rec X \gtc. \gtc{G}
            \mathrel{\raisebox{-1.5pt}[0ex][0ex]{$\ltrans{ \ell }$}}
            \gtc{G'}
        }
        \bussDisplay
    \end{mathparpagebreakable}

    \noindent
    Given $\vect \ell = \ell_1 , \ldots , \ell_n$, we write $\gtc{G} \ltrans{ \vect \ell } \gtc{G'}$ to denote $\gtc{G} \ltrans{ \ell_1 } \ldots \ltrans{ \ell_n } \gtc{G'}$. 
\end{definition}

\begin{theorem}[Soundness]\label{t:soundnessMain}
    If $\satisfies{\net P}{\gtc{G}}$ (\Cref{d:monSat}) and $\net P \trans* \net P_0$  then there exist $\gtc{G'},\vect \ell,\net P'$ such that $\gtc{G} \ltrans{\vect \ell} \gtc{G'}$, $\net P_0 \trans* \net P'$, and $\satisfies{\net P'}{\gtc{G'}}$.
\end{theorem}

We sketch the proof of \cref{t:soundnessMain} (see \ifappendix \Cref{s:soundnessProof} \else \cite{report/vdHeuvelPD23}\fi~for details).
We prove a stronger statement that starts from a network $\net P$ that satisfies an intermediate $\gtc{G_0}$ reachable from  $\gtc{G}$.
This way, we apply induction on the number of transitions between $\net P$ and $\net P_0$, relating the transitions of the network to the transitions of $\gtc{G_0}$ one step at a time by relying on \nameref{d:satisfaction}.
Hence, we inductively ``consume'' the transitions between $\net P$ and $\net P_0$ until we have passed through $\net P_0$ and end up in a network satisfying $\gtc{G'}$ reachable from $\gtc{G_0}$.
We use an auxiliary lemma to account for global types with \emph{independent exchanges}, such as  $\gtc{G'} = p \send q \gtBraces{ \msg \ell<T> \gtc. r \send s \gtBraces{ \msg \ell'<T'> \gtc. \pend } }$.
In $\gtc{G'}$, the exchange involving $(p,q)$ is unrelated to that involving $(r,s)$, so they occur concurrently in a network implementing $\gtc{G'}$.
Hence, the transitions from $\net P$ to $\net P_0$ might not follow the order specified in $\gtc{G_0}$.
The lemma ensures that concurrent  (i.e., unrelated) transitions always end up in the same state, no matter the order.
This way we show transitions from $\net P$ in the order specified in $\gtc{G_0}$, which we restore to the observed order using the lemma when we are done.

\cref{t:soundnessMain} implies that any  $\net P$ that satisfies some global type is \emph{error free}, i.e., $\net P$ never reduces to a network containing $\perror_D$ (\ifappendix\cref{t:errorFreedom}, \cref{s:soundnessProof}\else{}stated and proved formally in the full version of this paper~\cite{report/vdHeuvelPD23}\fi).

\subsection{Transparency}
\label{s:transparency}

The task of monitors is to observe and verify behavior \emph{with minimal interference}: monitors should be \emph{transparent}.
Transparency is usually expressed as a bisimulation between a monitored and unmonitored component~\cite{journal/ijis/LigattiBW05,journal/ijsttt/FalconeFM12,journal/tcs/BocchiCDHY17,conf/concur/AcetoCFI18}.

Our \emph{second result} is thus a transparency result.
For it to be informative, we assume that we observe the (un)monitored blackbox as if it were running in a network of monitored blackboxes that adhere to a given global protocol.
This way, we can assume that received messages are correct, such that the monitor does not transition to an error signal.
To this end, we enhance the \nameref{d:ltsNetworks}:
\begin{enumerate}
    \item
        As in \nameref{d:satisfaction}, we consider (un)monitored blackboxes on their own.
        Hence, we need a way to simulate messages sent by other participants.
        Otherwise, a blackbox would get stuck waiting for a message and the bisimulation would hold trivially.
        We thus add a transition that buffers messages.
        Similar to \satref{i:sInput}{Input} and \satclause{i:sDependencyInput}{Dependency Input}, these messages cannot be arbitrary; we parameterize the enhanced LTS by an \emph{oracle} that determines which messages are allowed as stipulated by a given global type.

    \item
        Besides observing and verifying transitions, our monitors additionally send dependency messages.
        This leads to an asymmetry in the behavior of monitored blackboxes and unmonitered blackboxes, as the latter do not send dependency messages.
        Hence, we rename dependency output actions to $\tau$.
\end{enumerate}

\noindent
We now define the enhanced LTS for networks, after setting up some notation.

\begin{notation}
    Let $\bm{A}$ denote the set of all actions.
    Given $\Omega : \Pow(\vect {\bm{A}})$, we write $\alpha + \Omega$ to denote the set containing every sequence in $\Omega$ prepended with $\alpha$.
    We write $\Omega(\alpha) = \Omega'$ iff $\alpha + \Omega' \subseteq \Omega$ and there is no $\Omega''$ such that $\alpha + \Omega' \subset \alpha + \Omega'' \subseteq \Omega$.
\end{notation}

\begin{definition}[Enhanced LTS for Networks]\label{d:enhancedLTS}
    We define an \emph{enhanced LTS for Networks}, denoted $\net P \bLtrans{\Omega}{\alpha}{\Omega'} \net P'$ where $\Omega , \Omega' : \Pow(\vect {\bm{A}})$, by the rules in \cref{f:enhancedLTS}.
    We write $\net P \bLtrans*{\Omega}{~~}{\Omega} \net P'$
    whenever $\net P$ transitions to $\net P'$ in zero or more $\tau$-transitions, i.e., $\net P \bLtrans{\Omega}{ \tau }{\Omega} \cdots \bLtrans{\Omega}{ \tau }{\Omega} \net P'$.
    We write $\net P \bLtrans*{\Omega}{ \alpha }{\Omega'} \net P'$ when $\net P \bLtrans*{\Omega}{~~}{\Omega} \net P_1 \bLtrans{\Omega}{\alpha}{\Omega'} \net P_2 \bLtrans*{\Omega'}{~~}{\Omega'} \net P'$, omitting the $\alpha$-transition when $\alpha = \tau$.
    Given $\vect \alpha = \alpha_1 , \ldots , \alpha_n$, we write $\net P \bLtrans*{\Omega_0}{ \vect \alpha }{\Omega_n} \net P'$ when $\net P \bLtrans*{\Omega_0}{\alpha_1}{\Omega_1} \net P_1 \cdots \net P_{n-1} \bLtrans*{\Omega_{n-1}}{\alpha_n}{\Omega_n} \net P'$.
\end{definition}

\begin{figure}[t]
    \begin{mathpar}
        \def\defaultHypSeparation{\hskip .8ex}
        \inferLabel{buf-mon}~
        \begin{bussproof}
            \bussAssume{
                \alpha = p \recv q (x), n' = q \send p (x)
                \mathbin{\text{or}}
                \alpha = p \recv q \Parens{x}, n' = q \send p \Parens{x}
            }
            \bussAssume{
                \Omega(\alpha) = \Omega'
            }
            \bussBin{
                \raisebox{-8.0pt}[0ex][0ex]{$
                    \monImpl{ \bufImpl{ p }{ P }{ \vect m } }{ {\mc{M}} }{ \vect n }
                    \mathrel{\raisebox{-3.0pt}[0ex][0ex]{$\bLtrans{\Omega}{ \alpha }{\Omega'}$}}
                    \monImpl{ \bufImpl{ p }{ P }{ \vect m } }{ \mc{M} }{ n' , \vect n }
                $}
            }
        \end{bussproof}
        \and
        \inferLabel{buf-unmon}~
        \begin{bussproof}
            \bussAssume{
                \vphantom{\raisebox{7pt}{$\net P$}}
                \alpha = p \recv q (x), m' = q \send p (x)
                \mathbin{\text{or}}
                \alpha = p \recv q \Parens{x}, m' = q \send p \Parens{x}
            }
            \bussAssume{
                \Omega(\alpha) = \Omega'
            }
            \bussBin{
                \raisebox{-8.0pt}[0ex][0ex]{$
                    \bufImpl{ p }{ P }{ \vect m }
                    \mathrel{\raisebox{-3.0pt}[0ex][0ex]{$\bLtrans{\Omega}{ \alpha }{\Omega'}$}}
                    \bufImpl{ p }{ P }{ m' , \vect m }
                $}
            }
        \end{bussproof}
        \\
        \inferLabel{dep}~
        \begin{bussproof}
            \bussAssume{
                \vphantom{\raisebox{7pt}{$\net P$}}
                \net P
                \ltrans{ p \send q \Parens{\ell} }
                \net P'
            }
            \bussUn{
                \net P
                \bLtrans{\Omega}{ \tau }{\Omega}
                \net P'
            }
        \end{bussproof}
        \and
        \inferLabel{no-dep}~
        \begin{bussproof}
            \bussAssume{
                \vphantom{\raisebox{7pt}{$\net P$}}
                \net P
                \ltrans{ \alpha }
                \net P'
            }
            \bussAssume{
                \alpha \notin \braces{ p \recv q (x), p \recv q \Parens{x} , p \send q \Parens{\ell} }
            }
            \bussAssume{
                \Omega(\alpha) = \Omega'
            }
            \bussTern{
                \net P
                \bLtrans{\Omega}{ \alpha }{\Omega'}
                \net P'
            }
        \end{bussproof}
    \end{mathpar}
    \caption{\nameref{d:enhancedLTS} (\cref{d:enhancedLTS}).}\label{f:enhancedLTS}
\end{figure}

\noindent
Thus, Transitions~\ruleLabel*{buf} simulate messages from other participants, consulting $\Omega$ and transforming it into $\Omega'$.
Transition~\ruleLabel{dep} renames dependency outputs to $\tau$.
Transition~\ruleLabel{no-dep} passes any other transitions, updating  $\Omega$ to $\Omega'$ accordingly.

We now define a weak bisimilarity on networks, governed by oracles.

\begin{definition}[Bisimilarity]\label{d:weakBisim}
    A relation $\mathcal{B} : \bm{N} \times \Pow(\vect {\bm{A}}) \times \bm{N}$ is a \emph{(weak) bisimulation} if, for every $( \net P , \Omega , \net Q ) \in \mathcal{B}$:
    (1)~For every $\net P' , \alpha , \Omega_1$ such that $\net P \bLtrans{\Omega}{ \alpha }{\Omega_1} \net P'$, there exist $\vect b , \Omega_2 , \net Q' , \net P''$ such that $\net Q \bLtrans*{\Omega}{ \vect b , \alpha }{\Omega_2} \net Q'$, $\net P' \bLtrans*{\Omega_1}{ \vect b }{\Omega_2} \net P''$, and $( \net P'' , \Omega_2 , \net Q' ) \in \mathcal{B}$; and (2)~The symmetric analog.

    We say $\net P$ and $\net Q$ are \emph{bisimilar} with respect to $\Omega$, denoted $\net P \weakBisim{}{\Omega} \net Q$, if there exists a bisimulation $\mathcal{B}$ such that $( \net P , \Omega , \net Q ) \in \mathcal{B}$.
\end{definition}

\noindent
Clause~1 says that $\net Q$ can mimic a transition from $\net P$ to $\net P'$, possibly after   $\tau$- and \ruleLabel{buf}-transitions.
We then allow $\net P'$ to ``catch up'' on those additional transitions, after which the results are bisimilar (under a new oracle); Clause~2 is symmetric.
Additional \ruleLabel{buf}-transitions are necessary: an unmonitored blackbox can read messages from its buffer directly, whereas a monitor may need to move messages between buffers first.
If the monitor first needs to move messages that are not in its buffer yet, we need to add those messages with \ruleLabel{buf}-transitions.
The unmonitored blackbox then needs to catch up on those additional messages.

Similar to \nameref{t:soundnessMain}, \nameref{d:satisfaction} defines the conditions under which we prove  transparency of monitors.
Moreover, we need to define the precise oracle under which  bisimilarity holds.
This oracle is defined  similarly to \nameref{d:satisfaction}: it depends on actions observed, relative types (in $\RTs$), and prior choices (in $\Lbls$).

\begin{definition}[Label Oracle]\label{d:labelOracle}
    The \emph{label oracle} of participant $p$ under $\RTs : \bm{P} \rightarrow \bm{R}$ and $\Lbls : \Pow(\vect {\bm{L}}) \rightarrow \bm{L}$, denoted $\LO(p,\RTs,\Lbls)$ is defined in \cref{f:labelOracle}.
\end{definition}

\begin{figure}[t]
    To improve readability, below we write `$\bigcup [x \in S \ldots]$' instead of `$\bigcup_{x \in S \ldots}$'.
    \begin{align*}
        \hphantom{\cup}
        ~&
        \bigcup [(q,\rtc{R}) \in \RTs .~ \rtc{R} \unfoldeq p \send q^{\mbb L} \rtBraces{ \msg i<T_i> \rtc. \rtc{R_i} }_{i \in I}]
        ~
        \bigcup [j \in I]
        \\ & \quad
        p \send q (\msg j<T_j>) + \LO( p , \RTs \update{q \mapsto \rtc{R_j}} , \Lbls \update{\mbb L \mapsto j} )
        \tag*{(Output)}
        \displaybreak[0] \\[.3mm] \cup
        ~&
        \bigcup [(q,\rtc{R}) \in \RTs .~ \rtc{R} \unfoldeq q \send p^{\mbb L} \rtBraces{ \msg i<T_i> \rtc. \rtc{R_i} }_{i \in I}]
        ~
        \bigcup [j \in I]
        \\ & \quad
        p \recv q (\msg j<T_j>) + \LO( p , \RTs \update{q \mapsto \rtc{R_j}} , \Lbls \update{\mbb L \mapsto j} )
        \tag*{(Input)}
        \displaybreak[0] \\[.3mm] \cup
        ~&
        \bigcup [(q,\rtc{R}) \in \RTs .~ \rtc{R} \unfoldeq (q \lozenge r) \send p^{\mbb L} \rtBraces{ i \rtc. \rtc{R_i} }_{i \in I} \wedge {\not\exists} \mbb L' \in \dom(\Lbls) .~ \mbb L' \overlap \mbb L]
        ~
        \bigcup [j \in I]
        \\ & \quad
        p \recv q (j) + \LO( p , \RTs \update{q \mapsto \rtc{R_j}} , \Lbls \update{\mbb L \mapsto j} )
        \tag*{(Fresh Dependency Input)}
        \displaybreak[0] \\[.3mm] \cup
        ~&
        \bigcup [(q,\rtc{R}) \in \RTs .~ \rtc{R} \unfoldeq (q \lozenge r) \send p^{\mbb L} \rtBraces{ i \rtc. \rtc{R_i} }_{i \in I} \wedge \exists (\mbb L',j) \in \Lbls .~ \mbb L' \overlap \mbb L \wedge j \in I]
        \\ & \quad
        p \recv q (j) + \LO( p , \RTs \update{q \mapsto \rtc{R_j}} , \Lbls )
        \tag*{(Known Dependency Input)}
        \displaybreak[0] \\[.3mm] \cup
        ~&
        \bigcup [(q,\rtc{R}) \in \RTs .~ \rtc{R} \unfoldeq (p \lozenge r) \send q^{\mbb L} \rtBraces{ i \rtc. \rtc{R_i} }_{i \in I} \wedge \exists (\mbb L',j) \in \Lbls .~ \mbb L' \overlap \mbb L \wedge j \in I]
        \\ & \quad
        \LO( p , \RTs \update{q \mapsto \rtc{R_j}} , \Lbls )
        \tag*{(Dependency Output)}
        \displaybreak[0] \\[.3mm] \cup
        ~&
        \pend + \LO( p , \emptyset , \emptyset ) \quad [\text{only if } \forall ( q , \rtc{R} ) \in \RTs .~  \rtc{R} \unfoldeq \pend]
        \tag*{(End)}
        \displaybreak[0] \\[.3mm] \cup
        ~&
        \tau + \LO( p , \RTs , \Lbls )
        \tag*{(Tau)}
    \end{align*}
    \caption{Definition of the \nameref{d:labelOracle} (\cref{d:labelOracle}), $\LO(p,\RTs,\Lbls)$.}\label{f:labelOracle}
\end{figure}

\noindent
The label oracle $\LO(p,\RTs,\Lbls)$ thus consists of several subsets, each resembling a condition of \nameref{d:satisfaction} in \cref{f:satisfaction}.
Dependency outputs are exempt: the \nameref{d:enhancedLTS} renames them to $\tau$, so the label oracle simply looks past them without requiring a dependency output action.

We now state our transparency result, after defining a final requirement: minimality of satisfaction.
This allows us to step backward through satisfaction relations, such that we can reason about buffered messages.

\begin{definition}[Minimal Satisfaction]\label{d:minimalSat}
    We write $\satisfies[\vdash]{ \net P }{ \gtc{G} }[ p ]$ whenever there exists $\mathcal{R}$ such that $\mathcal{R} \satisfies{ \net P }{ \RTsOf(\gtc{G},p) }[ p ]$ (Def.~\labelcref{d:satisfaction}) and $\mathcal{R}$ is \emph{minimal}, i.e., there is no $\mathcal{R}' \subset \mathcal{R}$ such that $\mathcal{R}' \satisfies{ \net P }{ \RTsOf(\gtc{G},p) }[ p ]$.
\end{definition}

\begin{theorem}[Transparency]\label{t:transparencyMain}
    Suppose $\satisfies[\vdash]{ \monImpl{ \bufImpl{ p }{ P }{ \epsi } }{ \mc{M} }{ \epsi } }{ \gtc{G} }[ p ]$ (Def.~\labelcref{d:minimalSat}).
    Let $\Omega := \LO(p,\RTsOf(\gtc{G},p),\emptyset)$.
    Then $\monImpl{ \bufImpl{ p }{ P }{ \epsi } }{ \mc{M} }{ \epsi } \weakBisim{}{\Omega} \bufImpl{ p }{ P }{ \epsi }$.
\end{theorem}

We sketch the proof of \cref{t:transparencyMain} (see \ifappendix \cref{s:transparencyProof}\else \cite{report/vdHeuvelPD23}\fi).
The minimal satisfaction of the monitored blackbox contains all states that the monitored blackbox can reach through transitions.
We create a relation $\mathcal{B}$ by pairing each such state $\monImpl{ \bufImpl{ p }{ P' }{ \vect m } }{ \mc{M'} }{ \vect n }$ with $\bufImpl{ p }{ P' }{ \vect n , \vect m }$---notice how the buffers are combined.
We do so while keeping an informative relation between relative types, monitors, buffers, and oracles.
This information gives us the appropriate oracles to include in $\mathcal{B}$.
We then show that $\mathcal{B}$ is a weak bisimulation by proving that the initial monitored and unmonitored blackbox are in $\mathcal{B}$, and that the conditions of \cref{d:weakBisim} hold.
While Clause~1 is straightforward, 
Clause~2 requires care: by using the relation between relative types, monitors, and buffers, we infer the shape of the monitor from a transition of the unmonitored blackbox.
This allows us to show that the monitored blackbox can mimic the transition, possibly after outputting dependencies and/or receiving additional messages (as discussed above).

\smallskip

We close by comparing our \cref{t:soundnessMain,t:transparencyMain} with Bocchi \etal's safety and transparency results~\cite{journal/tcs/BocchiCDHY17}, respectively.
First, their safety result~\cite[Thm.~5.2]{journal/tcs/BocchiCDHY17} guarantees satisfaction instead of assuming it; their framework suppresses unexpected messages, which prevents the  informative guarantee given by our \cref{t:soundnessMain}.
Second, \cref{t:transparencyMain} and their transparency result~\cite[Thm.~6.1]{journal/tcs/BocchiCDHY17} differ, among other things, in the presence of an oracle, which is not needed in their setting: they can inspect the inputs of monitored processes, whereas we cannot verify the inputs of a blackbox without actually sending messages to it.

\section{Conclusion}
\label{s:conclusion}

We have proposed a new framework for dynamically analyzing networks of communicating components (blackboxes), governed by global types, with minimal assumptions about observable behavior.
We use global types and relative projection~\cite{journal/scico/vdHeuvelP22} to synthesize monitors, and define when a monitored component satisfies the governing protocol.
We prove that networks of correct monitored components are sound with respect to a global type, and that monitors are transparent.

We have implemented a practical toolkit, called \texttt{RelaMon}, based on the framework presented here.
\texttt{RelaMon} allows users to deploy JavaScript programs that monitor web-applications in any programming language and with third-party/closed-source components according to a global type.
The toolkit is publicly available~\cite{web/DobreHP23} and
includes implementations of our running example (the global type $\gtc{G_{\mathsf{a}}}$), as well as an example that incorporates a closed-source weather API.
\ifappendix
App.~\labelcref{a:toolkit} includes more details.
\fi

As future work, we plan to extend our framework to uniformly analyze systems combining monitored blackboxes and statically checked components (following~\cite{journal/scico/vdHeuvelP22}).
We also plan to study under which restrictions our approach coincides with Bocchi \etal's~\cite{journal/tcs/BocchiCDHY17}.

\paragraph{Acknowledgments}
We are grateful to
the anonymous reviewers for useful remarks.

\newpage
\addcontentsline{toc}{section}{References}
\bibliography{refs}

\appendices

\input{appendix.tex}

\end{document}

%% file: appendix.tex
\newpage
\section{The Running Example from~\cite{journal/tcs/BocchiCDHY17}}
\label{a:exampleBocchi}

Bocchi \etal~\cite{journal/tcs/BocchiCDHY17} develop a running example that is similar to ours.
Their example concerns an ATM protocol between a client ($c$) and a payment server ($s$), preceded by a client authorization through a separate authenticator ($a$).

Bocchi \etal's example includes assertions.
Assertions are orthogonal to the method of projecting global types onto local procotols and extracting monitors from global types; adding assertions does not modify the spirit of our approach.
Next we present $\gtc{G_{\mathsf{ATM}}}$, a version of Bocchi \etal's running example without assertions; it allows us to illustrate how our approach covers protocols considered in Bocchi \etal's approach.
\begin{align*}
    \gtc{G_{\mathsf{ATM}}} &:= c \send a \gtBraces{ \msg \mathsf{login}<\mathsf{str}> \gtc. a \send s \gtBraces{ \msg \mathsf{ok}<> \gtc. a \send c \gtBraces{ \msg \mathsf{ok}<> \gtc. \gtc{G_{\mathsf{loop}}} } \gtc, \msg \mathsf{fail}<> \gtc. a \send c \gtBraces{ \msg \mathsf{fail}<> \gtc. \pend } } }
    \\
    \gtc{G_{\mathsf{loop}}} &:= \rec X \gtc. s \send c \gtBraces{ \msg \mathsf{account}<\mathsf{int}> \gtc. c \send s \gtBraces{ \msg \mathsf{withdraw}<\mathsf{int}> \gtc. X \gtc, \msg \mathsf{deposit}<\mathsf{int}> \gtc. X \gtc, \msg \mathsf{quit}<> \gtc. \pend } }
\end{align*}
Notice how, for this example to work under traditional forms of projection, $a$ needs to explicitly forward the success of the login attempt to $c$.

Our framework supports $\gtc{G_{\mathsf{ATM}}}$ as is, because it is well-formed according to \cref{d:wf}.
The relative projections attesting to this are as follows (cf.~\cref{alg:relativeProjection}):
\begin{align*}
    \gtc{G_{\mathsf{loop}}} \wrt (c,s) &= \rec X \rtc. s \send c \rtBraces{ \msg \mathsf{account}<\mathsf{int}> \rtc. c \send s \rtBraces{ \msg \mathsf{withdraw}<\mathsf{int}> \rtc. X \rtc, \msg \mathsf{deposit}<\mathsf{int}> \rtc. X \rtc, \msg \mathsf{quit}<> \rtc. \pend } }
    \\
    \gtc{G_{\mathsf{ATM}}} \wrt (c,s) &= (s \recv a) \send c \rtBraces{ \mathsf{ok} \rtc. \big( \gtc{G_{\mathsf{loop}}} \wrt (c,s) \big) \rtc, \mathsf{quit} \rtc. \pend }
    \\
    \gtc{G_{\mathsf{loop}}} \wrt (c,a) &= \pend
    \\
    \gtc{G_{\mathsf{ATM}}} \wrt (c,a) &= c \send a \rtBraces{ \msg \mathsf{login}<\mathsf{str}> \rtc. (a \send s) \send c \rtBraces{ \mathsf{ok} \rtc. a \send c \rtBraces{ \msg \mathsf{ok}<> \rtc. \pend } \rtc, \mathsf{quit} \rtc. a \send c \rtBraces{ \msg \mathsf{quit}<> \rtc. \pend } } }
    \\
    \gtc{G_{\mathsf{loop}}} \wrt (s,a) &= \pend
    \\
    \gtc{G_{\mathsf{ATM}}} \wrt (s,a) &= a \send s \rtBraces{ \msg \mathsf{ok}<> \rtc. \pend \rtc, \msg \mathsf{fail}<> \rtc. \pend }
\end{align*}

As mentioned before, $\gtc{G_{\mathsf{ATM}}}$ contains an explicit dependency.
We can modify the global type to make this dependency implicit, without altering $\gtc{G_{\mathsf{loop}}}$:
\[
    \gtc{G'_{\mathsf{ATM}}} := c \send a \gtBraces{ \msg \mathsf{login}<\mathsf{str}> \gtc. a \send s \gtBraces{ \msg \mathsf{ok}<> \gtc. \gtc{G_{\mathsf{loop}}} \gtc, \msg \mathsf{fail}<> \gtc. \pend } }
\]
The resulting relative projections are then as follows.
Notice how the change has simplified the projection onto $(c,a)$.
\begin{align*}
    \rtc{R'_{c,s}} := \gtc{G'_{\mathsf{ATM}}} \wrt (c,s) &= (s \recv a) \send c \rtBraces{ \mathsf{ok} \rtc. \big(\gtc{G_{\mathsf{loop}}} \wrt (c,s)\big) \rtc, \mathsf{quit} \rtc. \pend }
    \\
    \rtc{R'_{c,a}} := \gtc{G'_{\mathsf{ATM}}} \wrt (c,a) &= c \send a \rtBraces{ \msg \mathsf{login}<\mathsf{str}> \gtc. \pend }
    \\
    \rtc{R'_{s,a}} := \gtc{G'_{\mathsf{ATM}}} \wrt (s,a) &= a \send s \rtBraces{ \msg \mathsf{ok}<> \rtc. \pend \rtc, \msg \mathsf{quit}<> \rtc. \pend }
\end{align*}

Using \cref{alg:gtToMon}, we extract monitors from $\gtc{G'_{\mathsf{ATM}}}$:
\begin{align*}
    \gtToMon(\gtc{G_{\mathsf{loop}}},c,\braces{s,a}) &= \rec X \mc. c \recv s \mBraces{
        \msg \mathsf{account}<\mathsf{int}> \mc. c \send \emptyset (\mathsf{account}) \mc. c \send s \mBraces*{
            \begin{array}{@{}l@{}}
                \msg \mathsf{withdraw}<\mathsf{int}> \mc. c \send \emptyset (\mathsf{withdraw}) \mc. X \mc,
                \\
                \msg \mathsf{deposit}<\mathsf{int}> \mc. c \send \emptyset (\mathsf{deposit}) \mc. X \mc,
                \\
                \msg \mathsf{quit}<> \mc. c \send \emptyset (\mathsf{quit}) \mc. \pend
            \end{array}
        }
    }
    \\
    \mc{M'_c} := \gtToMon(\gtc{G'_{\mathsf{ATM}}},c,\braces{s,a}) &= c \send a \mBraces{ \msg \mathsf{login}<\mathsf{str}> \mc. c \send \emptyset (\mathsf{login}) \mc. c \recv s \mBraces{ \mathsf{ok} \mc. \gtToMon(\gtc{G_{\mathsf{loop}}},c,\braces{s,a}) \mc, \mathsf{fail} \mc. \pend } }
    \\
    \gtToMon(\gtc{G_{\mathsf{loop}}},s,\braces{c,a}) &= \rec X \mc. s \send c \mBraces{
        \msg \mathsf{account}<\mathsf{int}> \mc. s \send \emptyset (\mathsf{account}) \mc. s \recv c \mBraces{
            \begin{array}{@{}l@{}}
                \msg \mathsf{withdraw}<\mathsf{int}> \mc. s \send \emptyset (\mathsf{withdraw}) \mc. X \mc,
                \\
                \msg \mathsf{deposit}<\mathsf{int}> \mc. s \send \emptyset (\mathsf{deposit}) \mc. X \mc,
                \\
                \msg \mathsf{quit}<> \mc. s \send \emptyset (\mathsf{quit}) \mc. \pend
            \end{array}
        }
    }
    \\
    \mc{M'_s} := \gtToMon(\gtc{G'_{\mathsf{ATM}}},s,\braces{c,a}) &= s \recv a \mBraces{ \msg \mathsf{ok}<> \mc. s \send \braces{c} (\mathsf{ok}) \mc. \gtToMon(\gtc{G_{\mathsf{loop}}},s,\braces{c,a}) \mc, \msg \mathsf{fail}<> \mc. s \send \braces{c} (\mathsf{fail}) \mc. \pend }
    \\
    \mc{M'_a} := \gtToMon(\gtc{G'_{\mathsf{ATM}}},a,\braces{c,s}) &= a \recv c \mBraces{ \msg \mathsf{login}<\mathsf{str}> \mc. a \send \emptyset (\mathsf{login}) \mc. a \send s \mBraces{ \msg \mathsf{ok}<> \mc. a \send \emptyset (\mathsf{ok}) \mc. \pend \mc, \msg \mathsf{fail}<> \mc. a \send \emptyset (\mathsf{fail}) \mc. \pend } }
\end{align*}

The following are example blackboxes for the participants of $\gtc{G'_{\mathsf{ATM}}}$:
\begin{center}
    \begin{tikzpicture}[node distance=24mm,every path/.style={->},every edge quotes/.style={font=\scriptsize,auto=left},tight/.style={inner sep=1pt},loose/.style={inner sep=5pt}]
        \node (ci) {$Q_c$};
        \node[right=of ci] (cl) {$Q_c^{\,\mathsf{l}}$};
        \node[right=of cl] (cf) {$Q_c^{\mathsf{f}}$};
        \node[right=of cf] (ce) {$Q_c^{\mathsf{e}}$};
        \node[below=of cl] (co) {$Q_c^{\mathsf{o}}$};
        \node[right=of co] (ca) {$Q_c^{\mathsf{a}}$};
        \node[right=of ca] (cq) {$Q_c^{\mathsf{q}}$};
        \draw (ci) edge["$c \send a (\msg \mathsf{login}<\mathsf{str}>)$" tight] (cl)
              (cl) edge["$c \recv s \Parens{\mathsf{fail}}$" tight] (cf)
              (cf) edge["$\pend$"] (ce)
              (cl) edge["$c \recv s \Parens{\mathsf{ok}}$" loose] (co)
              (co) edge[bend left=35,"$c \recv s (\msg \mathsf{account}<\mathsf{int}>)$" tight] (ca)
              (ca) edge[bend left=35,"$\begin{array}{@{}r@{}}c \send s (\msg \mathsf{withdraw}<\mathsf{int}>)\\{}c \send s (\msg \mathsf{deposit}<\mathsf{int}>)\end{array}$" loose] (co)
              (ca) edge["$c \send s (\msg \mathsf{quit}<>)$" tight] (cq)
              (cq) edge["$\pend$"] (ce);

          \node[below=60mm of ci] (si) {$Q_s$};
          \node[right=of si] (so) {$Q_s^{\mathsf{o}}$};
          \node[right=of so] (sa) {$Q_s^{\mathsf{a}}$};
          \node[below=of sa] (sq) {$Q_s^{\mathsf{q}}$};
          \node[right=of sq] (se) {$Q_s^{\mathsf{e}}$};

          \draw (si) edge["$s \recv a (\msg \mathsf{ok}<>)$" tight] (so)
                (so) edge[bend left=35,"$s \send c (\msg \mathsf{account}<\mathsf{int}>)$" loose] (sa)
                (sa) edge[bend left=35,"$\begin{array}{@{}c@{}}s \recv c (\msg \mathsf{withdraw}<\mathsf{int}>)\\{}s \recv c (\msg \mathsf{deposit}<\mathsf{int}>)\end{array}$" loose] (so)
                (sa) edge["$s \recv c (\msg \mathsf{quit}<>)$"] (sq)
                (si) edge[bend right=25,"$s \recv a (\msg \mathsf{fail}<>)$"] (sq)
                (sq) edge["$\pend$"] (se);

            \node[below=45mm of si] (ai) {$Q_a$};
            \node[right=of ai] (al) {$Q_a^{\,\mathsf{l}}$};
            \node[right=of al] (aq) {$Q_a^{\mathsf{q}}$};
            \node[right=of aq] (ae) {$Q_a^{\mathsf{e}}$};
            \draw (ai) edge["$a \recv c (\msg \mathsf{login}<\mathsf{str}>)$" tight] (al)
                  (al) edge[bend left=35,"$a \send c (\msg \mathsf{ok}<>)$"] (aq)
                  (al) edge[bend right=35,"$a \send c (\msg \mathsf{fail}<>)$"] (aq)
                  (aq) edge["$\pend$"] (ae);
    \end{tikzpicture}
\end{center}

It is not difficult to confirm that the following satisfactions hold (cf.\ \cref{d:satisfaction,d:monSat}):
\begin{align*}
    & \satisfies{ \monImpl{ \bufImpl{ c }{ Q_c }{ \epsi } }{ \mc{M'_c} }{ \epsi } }{ \braces{ ( s , \rtc{R'_{c,s}} ) , ( a , \rtc{R'_{c,a}} ) } }[ c ]
    \\
    & \satisfies{ \monImpl{ \bufImpl{ s }{ Q_s }{ \epsi } }{ \mc{M'_s} }{ \epsi } }{ \braces{ ( c , \rtc{R'_{c,s}} ) , ( a , \rtc{R'_{s,a}} ) } }[ s ]
    \\
    & \satisfies{ \monImpl{ \bufImpl{ a }{ Q_a }{ \epsi } }{ \mc{M'_a} }{ \epsi } }{ \braces{ ( c , \rtc{R'_{c,a}} ) , ( s , \rtc{R'_{s,a}} ) } }[ a ]
    \\
    & \satisfies{ \monImpl{ \bufImpl{ c }{ Q_c }{ \epsi } }{ \mc{M'_c} }{ \epsi } \| \monImpl{ \bufImpl{ s }{ Q_s }{ \epsi } }{ \mc{M'_s} }{ \epsi } \| \monImpl{ \bufImpl{ a }{ Q_a }{ \epsi } }{ \mc{M'_a} }{ \epsi } }{ \gtc{G_{\mathsf{a}}} }
\end{align*}

\newpage
\section{A Toolkit for Monitoring Networks of Blackboxes in Practice}
\label{a:toolkit}

To demonstrate the practical potential of our approach, we have developed a toolkit based on our framework---see \url{https://github.com/basvdheuvel/RelaMon}~\cite{web/DobreHP23}.

The toolkit enhances message-passing web-applications with monitors.
This way, it is possible to add a layer of security when communicating with, e.g., untrusted third-party APIs by monitoring their behavior according to an assumed governing protocol.
The toolkit includes:
\begin{enumerate}
    \item
        A tool, written in Rascal~\cite{conf/scam/KlintSV09,conf/scam/KlintSV19}, that transpiles protocols specified as well-formed global types to JSON.

    \item
        A monitor microservice, written in JavaScript, initialized with a protocol specification, a participant ID, the IP-addresses of the unmonitored component and the other components.
        The microservice uses relative projection on the supplied JSON protocol specification to construct a finite state machine using \cref{alg:gtToMon}, which acts as the monitor for the specified participant.
\end{enumerate}
When all components and their respective monitors have been deployed, the monitors perform a handshake such that all components are ready to start executing the protocol.
The monitors forward all correct messages between their respective components and the other monitors in the network, and if needed they send dependency messages.
When a monitor detects an incorrect message, it signals an error to its component and the other monitors.
This way, eventually the entire network becomes aware of the protocol violation, and the execution stops.
It is then up to the components to gracefully deal with the protocol violation, e.g., by reverting to a prior state or restarting the protocol from the start.

The toolkit comes with two test suites:
\begin{itemize}
    \item
        The authorization protocol in $G_{\mathsf{a}}$~\eqref{eq:auth}, our running example.

    \item
        A weather protocol $G_{\mathsf{w}}$ between a client ($c$), a city database ($d$), and a weather API ($w$); we omit curly braces for exchanges with a single branch:
        \[
            c \send w \, \msg \mathsf{key}<\mathsf{str}> . \rec X . c \send d \, \msg \mathsf{city}<\mathsf{str}> . d \send c \braces*{
                \begin{array}{@{}l@{}}
                    \msg \mathsf{coord}<\mathsf{str}> . c \send w \, \msg \mathsf{coord}<\mathsf{str}> . w \send c \, \msg \mathsf{temp}<\mathsf{real}> . X
                    , \\
                    \msg \mathsf{unknown}<> . X
                \end{array}
            }
        \]
        This is an interesting test suite, because the weather API (which requires an API $\mathsf{key}$) is not set up to deal with dependencies.
        The suite compensates by including a program that acts as a ``translator'' for the weather API.
        The system is then still protected from protocol violations by the weather API.
\end{itemize}

\newpage
\section{Definitions and Proofs}
\label{a:defsProofs}

\subsection{Relative Types with Locations (\Cref{s:globalTypesAsMonitors,s:correctMonitoredBlackboxes})}
\label{a:relativeTypesWithLocs}

Here, we formally define relative types with locations, and define how they are used in related definitions.

We refer to an ordered sequence of labels $\vect \ell$ as a \emph{location}.
We write $\mbb L$ to denote a set of locations.

\begin{definition}[Relative Types with Locations]\label{d:relativeTypesWithLocs}
    \emph{Relative types with locations} are defined by the following syntax:
    \begin{align*}
        \rtc{R},\rtc{R'} &::=
        p \send q^{\mbb L} \rtBraces{ \msg i<T_i> \rtc. \rtc{R} }_{i \in I}
        & \sepr* &
        (p \send r) \send q^{\mbb L} \rtBraces{ i \rtc. \rtc{R} }_{i \in I}
        & \sepr* &
        (p \recv r) \send q^{\mbb L} \rtBraces{ i \rtc. \rtc{R} }_{i \in I}
        & \sepr* &
        \rec X \rtc. \rtc{R}
        & \sepr* &
        X
        & \sepr* &
        \pend
    \end{align*}
\end{definition}

\begin{definition}[Relative Projection with Locations]\label{d:relativeProjectionWithLocs}
    \emph{Relative projection with Locations}, denoted $\gtc{G} \wrt (p,q)^{\mbb L}$, is defined by \cref{alg:relativeProjectionWithLocs}.
    This algorithm relies on three auxiliary definitions:
    \begin{itemize}
    	\item
            The \emph{erasure} of a relative type, denoted $\erase(\rtc{R})$, is defined by replacing each set of locations in $\rtc{R}$ by $\emptyset$ (e.g., $\erase(p \send q^{\mbb L} \rtBraces{ \msg i<T_i> \rtc. \rtc{R_i} }_{i \in I}) := p \send q^\emptyset \rtBraces{ \msg i<T_i> \rtc. \erase(\rtc{R_i}) }_{i \in I}$).

        \item
            Given $\rtc{R}$ and $\rtc{R'}$ such that $\erase(\rtc{R}) = \erase(\rtc{R'})$, we define the \emph{union} of $\rtc{R}$ and $\rtc{R'}$, denoted $\rtc{R} \cup \rtc{R'}$, by combining each set of locations for each corresponding message in $\rtc{R}$ and $\rtc{R'}$ (e.g., $p \send q^{\mbb L} \rtBraces{ \msg i<T_i> \rtc. \rtc{R_i} }_{i \in I} \cup p \send q^{\mbb L'} \rtBraces{ \msg i<T_i> \rtc. \rtc{R'_i} }_{i \in I} := p \send q^{\mbb L \cup \mbb L'} \rtBraces{ \msg i<T_i> \rtc. (\rtc{R_i} \cup \rtc{R'_i}) }_{i \in I}$).
            Given $(\rtc{R_a})_{a \in A}$ for finite $A$ such that, for each $a,b \in A$, $\erase(\rtc{R_a}) = \erase(\rtc{R_b})$, we inductively define $\bigcup_{a \in A} \rtc{R_a}$ as expected.

        \item
            We define the \emph{appendance} of a label to a set of locations, denoted $\mbb L + \ell$, as follows:
            \begin{align*}
                (\mbb L \cup \braces{ \vect \ell }) + \ell' &:= (\mbb L + \ell') \cup \braces{ \vect \ell , \ell' }
                &
                \emptyset + \ell' &:= \emptyset
            \end{align*}

            We extend this definition to the appendance of a location to a set of locations as follows:
            \begin{align*}
                \mbb L + (\ell' , \vect \ell) &:= (\mbb L + \ell') + \vect \ell
                &
                \mbb L + \epsi &:= \mbb L
            \end{align*}

            We extend this definition to the appendance of two sets of locations as follows:
            \begin{align*}
                \mbb L + \mbb L' &:= \bigcup_{\vect \ell' \in \mbb L'} \mbb L + \vect \ell'
            \end{align*}
    \end{itemize}
\end{definition}

\begin{algorithm}[t]
    \DontPrintSemicolon
    \SetAlgoNoEnd
    \SetInd{.2em}{.7em}
    \SetSideCommentRight
    \Def{$\gtc{G} \wrt (p,q)^{\mbb L}$}{

        \Switch{$\gtc{G}$}{

            \uCase{$s \send r \gtBraces{ \msg i<T_i> \gtc. \gtc{G_i} }_{i \in I}$}{
                $\forall i \in I.~ \rtc{R_i} := \gtc{G_i} \wrt (p,q)^{\mbb L + i}$ \;

                \lIf{$(p = s \wedge q = r)$}{
                    \KwRet
                    $p \send q^{\mbb L} \rtBraces{ \msg i<T_i> \rtc. \rtc{R_i} }_{i \in I}$
                }

                \lElseIf{$(q = s \wedge p = r)$}{
                    \KwRet
                    $q \send p^{\mbb L} \rtBraces{ \msg i<T_i> \rtc. \rtc{R_i} }_{i \in I}$
                }

                \lElseIf{$\forall i,j \in I.~ \erase(\rtc{R_i}) = \erase(\rtc{R_j})$}{
                    \KwRet
                    $\bigcup_{i \in I} \rtc{R_i}$
                }\label{li:rpIndepLab}

                \lElseIf{$s \in \braces{p,q} \wedge t \in \braces{p,q} \setminus \braces{s}$}{
                    \KwRet
                    $(s \send r) \send t^{\mbb L} \rtBraces{ i \rtc. \rtc{R_i} }_{i \in I}$
                }

                \lElseIf{$r \in \braces{p,q} \wedge t \in \braces{p,q} \setminus \braces{r}$}{
                    \KwRet
                    $(r \recv s) \send t^{\mbb L} \rtBraces{ i \rtc. \rtc{R_i} }_{i \in I}$
                }
            }

            \uCase{$\rec X \gtc. \gtc{G'}$}{
                $\rtc{R'} := \gtc{G'} \wrt (p,q)^{\mbb L}$ \;

                \lIf{$(\text{$\rtc{R'}$ contains an exchange or a recursive call on any $Y \neq X$})$}{
                    \KwRet
                    $\rec X^{\mbb L} \rtc. \rtc{R'}$
                }

                \lElse{
                    \KwRet
                    $\pend$
                }
            }

            \lCase{$X$}{
                \KwRet
                $X^{\mbb L}$
            }

            \lCase{$\pend$}{
                \KwRet
                $\pend$
            }
        }

    }

    \caption{\nameref{d:relativeProjectionWithLocs}.}
    \label{alg:relativeProjectionWithLocs}
\end{algorithm}

\begin{definition}[Dependence with Locations]\label{d:depsOnWithLocs}
    Given a well-formed global type $G$, we say $p$'s role in $\gtc{G}$ \emph{depends on} $q$'s role in the initial exchange in $\gtc{G}$, denoted $\depsOn p q \gtc{G}$, if an only if
    \[
        \gtc{G} = s \send r \gtBraces{ \msg i<T_i> \gtc. \gtc{G_i} }_{i \in I}
        \wedge p \notin \braces{s,r}
        \wedge q \in \braces{s,r}
        \wedge \exists i,j \in I.~ \erase(\gtc{G_i} \wrt (p,q)) \neq \erase(\gtc{G_j} \wrt (p,q)).
    \]
\end{definition}

\subsection{Unfolding Recursive Types and Monitors}

When unfolding recursive relative types, the location annotations require care.
Consider, for example
\[
    \rec X^{\braces{1,2}} \rtc. p \send q^{\braces{1,2}} \rtBraces{ \msg \ell_1<T_1> \rtc. q \send p^{\braces{1,2,\ell_1}} \rtBraces{ \msg \ell_2<T_2> \rtc. X^{\braces{1,2,\ell_1,\ell_2}} } \rtc, \msg \ell'<T'> \rtc. X^{\braces{1,2,\ell'}} }.
\]
To unfold this type, we should replace each recursive call on $X$ with a copy of the whole recursive definition.
However, this is insufficient: the locations of the original recursive definition (starting at $1,2$) do not concur with the locations of the recursive calls ($1,2,\ell_1,\ell_2$ and $1,2,\ell'$).
Thus, we need to update the locations of the copied recursive calls by inserting the new path behind the location of the original recursive definition.
This way, the recursive call at $1,2,\ell_1,\ell_2$ would get replaced by (inserted locations are underlined)
\[
    \rec X^{\braces{1,2,\underline{\ell_1,\ell_2}}} \rtc. p \send q^{\braces{1,2,\underline{\ell_1,\ell_2}}} \rtBraces{ \msg \ell_1<T_1> \rtc. q \send p^{\braces{1,2,\underline{\ell_1,\ell_2},\ell_1}} \rtBraces{ \msg \ell_2<T_2> \rtc. X^{\braces{1,2,\underline{\ell_1,\ell_2},\ell_1,\ell_2}} } \rtc, \msg \ell'<T'> \rtc. X^{\braces{1,2,\underline{\ell_1,\ell_2},\ell'}} }.
\]

We formally define the unfolding of a relative type $\rec X^{\mbb L} \rtc. \rtc{R}$ by removing the prefix $\mbb L$ from the locations in $\rtc{R}$ (using $\remPref$), replacing each recursive call $X^{\mbb L'}$ (where $\mbb L'$ is the location of the replaced recursive call) with the recursive definition beginning at location $\mbb L'$ (using $\prepend$), and then replacing the original location $\mbb L$ (using $\prepend$).

\begin{definition}[Manipulation and Comparison of Locations]\label{d:manipCompLocs}
\begin{itemize}

    \item
        Given locations $\vect \ell$ and $\vect \ell'$, we say $\vect \ell'$ is a \emph{prefix} of $\vect \ell$, denoted $\vect \ell' \prefixeq \vect \ell$, if there exists a \emph{suffix} $\vect \ell''$ such that $\vect \ell = \vect \ell' , \vect \ell''$.

    \item
        We say $\vect \ell'$ is a \emph{strict prefix} of $\vect \ell$, denoted $\vect \ell' \prefix \vect \ell$, if $\vect \ell' \prefixeq \vect \ell$ with suffix $\vect \ell'' \neq \epsi$.
        We extend the prefix relation to sets of locations as follows: $\mbb L' \prefixeq \mbb L$ iff $\forall \vect \ell \in \mbb L.~ \exists \vect \ell' \in \mbb L'.~ \vect \ell' \prefixeq \vect \ell$, i.e., each location in $\mbb L$ is prefixed by a location in $\mbb L'$.

    \item
        Given $\vect \ell \neq \epsi$, we write $\fst(\vect \ell)$ to denote the first element of $\vect \ell$; formally, there exists $\vect \ell'$ such that $\fst(\vect \ell) , \vect \ell' = \vect \ell$.

    \item
        Given a set of locations $\mbb L$ and a relative type $\rtc{R}$, we define the \emph{prependance} of $\mbb L$ to the locations in $\rtc{R}$, denoted $\prepend( \mbb L , \rtc{R} )$, by prepending $\mbb L$ to each location in $\rtc{R}$ inductively; e.g.,
        \[
            \prepend( \mbb L , p \send q^{\mbb L'} \rtBraces{ \msg i<T_i> \rtc. \rtc{R_i} }_{i \in I} ) := p \send q^{\mbb L + \mbb L'} \rtBraces{ \msg i<T_i> \rtc. \prepend( \mbb L , \rtc{R_i} ) }_{i \in I}.
        \]

    \item
        Given a relative type $\rtc{R}$, we define its \emph{first location}, denoted $\fstLoc(\rtc{R})$, as the location annotation on the first exchange in $\rtc{R}$; e.g., $\fstLoc(p \send q^{\mbb L} \rtBraces{ \msg i<T_i> \rtc. \rtc{R_i} }_{i \in I}) := \mbb L$.

    \item
        Given a set of locations $\mbb L$ and a relative type $\rtc{R}$, we define the \emph{removal of prefix} $\mbb L$ from the locations in $\rtc{R}$, denoted $\remPref( \mbb L , \rtc{R} )$.
        Formally, it checks each location in $\mbb L$ as a possible prefix of each location in the set of locations of each exchange of $\rtc{R}$, and leaves only the suffix; e.g., $\remPref( \mbb L , p \send q^{\mbb L'} \rtBraces{ \msg i<T_i> \rtc. \rtc{R_i} }_{i \in I} ) := p \send q^{\mbb L''} \rtBraces{ \msg i<T_i> \rtc. \remPref( \mbb L , \rtc{R_i} ) }_{i \in I}$, where $\mbb L'' := \braces{ \vect \ell'' \mid \exists \vect \ell \in \mbb L , \vect \ell' \in \mbb L'.~ \vect \ell \prefixeq \vect \ell' \text{ with suffix } \vect \ell'' }$.

    \end{itemize}
\end{definition}

\begin{definition}[Unfold Relative Type]\label{d:unfoldRT}
    Given a relative type $\rec X^{\mbb L} \rtc. \rtc{R}$, we define its \emph{one-level unfolding}, denoted $\unfold_1( \rec X^{\mbb L} \rtc. \rtc{R} )$, as follows:
    \[
        \unfold_1( \rec X^{\mbb L} \rtc. \rtc{R} ) := \prepend( \mbb L , \remPref( \mbb L , \rtc{R} ) \braces{ \prepend( \mbb L' , \remPref( \mbb L , \rec X^{\mbb L} \rtc. \rtc{R} ) ) / X^{\mbb L'} } )
    \]

    Given a relative type $\rtc{R}$, we define its \emph{(full) unfolding}, denoted $\unfold( \rtc{R} )$, as follows:
    \[
        \unfold( \rtc{R} )
        := \begin{cases}
            \unfold_1( \rec X^{\mbb L} \rtc. \unfold( \rtc{R'} ) )
            & \text{if $\rtc{R} = \rec X^{\mbb L} \rtc. \rtc{R'}$}
            \\
            \rtc{R}
            & \text{otherwise}
        \end{cases}
    \]
\end{definition}

\begin{lemma}\label{l:prefProjection}
    For any well-formed global type $\gtc{G}$, participants $p,q \in \prt(\gtc{G})$, and set of locations $\mbb L$,
    \begin{align*}
        \prepend( \mbb L , \gtc{G} \wrt (p,q)^{\braces{\epsi}} ) &= \gtc{G} \wrt (p,q)^{\mbb L},
        \\
        \remPref( \mbb L , \gtc{G} \wrt (p,q)^{\mbb L} ) &= \gtc{G} \wrt (p,q)^{\braces{\epsi}}.
    \end{align*}
\end{lemma}

\begin{proof}
    By definition.
\end{proof}

\begin{definition}[Unfold Global Type]\label{d:unfoldGT}
    Given a well-formed global type $\rec X \gtc. \gtc{G}$, we define its \emph{one-level unfolding}: $\unfold_1( \rec X \gtc. \gtc{G} ) := \gtc{G} \braces{ \rec X \rtc. \gtc{G} / X }$.
    Given a well-formed global type $\gtc{G}$, we define its \emph{(full) unfolding}, denoted $\unfold(\gtc{G})$, as follows:
    \[
        \unfold(\gtc{G}) := \begin{cases}
            \unfold_1( \rec X \gtc. \unfold( \gtc{G'} ) )
            & \text{if $\gtc{G} = \rec X \gtc. \gtc{G'}$}
            \\
            \gtc{G}
            & \text{otherwise}
        \end{cases}
    \]
\end{definition}

\begin{lemma}\label{l:unfoldRecursiveGTRTnoLabel}
    For any well-formed global type $\rec X \gtc. \gtc{G}$ and participants $p,q \in \prt(\gtc{G})$,
    \[
        (\gtc{G} \wrt (p,q)^{\braces{\epsi}}) \braces{ ((\rec X \gtc. \gtc{G}) \wrt (p,q)^{\mbb L'} ) / X^{\mbb L'} }
        =
        (\gtc{G} \braces{ \rec X \gtc. \gtc{G} / X }) \wrt (p,q)^{\braces{\epsi}}.
    \]
\end{lemma}

\begin{proof}
    By definition.
    The starting location $\mbb L'$ for replacing each $X^{\mbb L}$ is correct on the right-hand-side, because $\mbb L'$ is the location of the recursive call.
    Hence, the projection of the unfolded global type will at those spots start with $\mbb L'$.
\end{proof}

\begin{lemma}\label{l:unfoldRecursiveGTRT}
    For any well-formed global type $\rec X \gtc. \gtc{G}$, participants $p,q \in \prt(\gtc{G})$, and set of locations $\mbb L$,
    \[
        \unfold_1( \rec X \gtc. \gtc{G} \wrt (p,q)^{\mbb L} )
        =
        \unfold_1( \rec X \gtc. \gtc{G} ) \wrt (p,q)^{\mbb L}.
    \]
\end{lemma}

\begin{proof}
    By \cref{d:unfoldRT,l:prefProjection,l:unfoldRecursiveGTRTnoLabel}:
    \begin{align*}
        & \unfold_1( \rec X \gtc. \gtc{G} \wrt (p,q)^{\mbb L} )
        \\
        \overset{\text{\cref{d:relativeProjectionWithLocs}}}{=} & \unfold_1( \rec X^{\mbb L} \rtc. (\gtc{G} \wrt (p,q)^{\mbb L}) )
        \\
        \overset{\text{\cref{d:unfoldRT}}}{=} & \prepend( \mbb L , \remPref( \mbb L , \gtc{G} \wrt (p,q)^{\mbb L} ) \braces{ \prepend( \mbb L' , \remPref( \mbb L , \rec X^{\mbb L} \rtc. (\gtc{G} \wrt (p,q)^{\mbb L}) ) ) / X^{\mbb L'} } )
        \\
        \overset{\text{\cref{d:relativeProjectionWithLocs}}}{=} & \prepend( \mbb L , \remPref( \mbb L , \gtc{G} \wrt (p,q)^{\mbb L} ) \braces{ \prepend( \mbb L' , \remPref( \mbb L , \rec X \gtc. \gtc{G} \wrt (p,q)^{\mbb L} ) ) / X^{\mbb L'} } )
        \\
        \overset{\text{\cref{l:prefProjection}}}{=} & \prepend( \mbb L , (\gtc{G} \wrt (p,q)^{\braces{\epsi}}) \braces{ \prepend( \mbb L' , \rec X \gtc. \gtc{G} \wrt (p,q)^{\braces{\epsi}} ) / X^{\mbb L'} } )
        \\
        \overset{\text{\cref{l:prefProjection}}}{=} & \prepend( \mbb L , (\gtc{G} \wrt (p,q)^{\braces{\epsi}}) \braces{ (\rec X \gtc. \gtc{G} \wrt (p,q)^{\mbb L'}) / X^{\mbb L'} } )
        \\
        \overset{\text{\cref{l:unfoldRecursiveGTRTnoLabel}}}{=} & \prepend( \mbb L , \gtc{G} \braces{ \rec X \gtc. \gtc{G} / X } \wrt (p,q)^{\braces{\epsi}} )
        \\
        \overset{\text{\cref{l:prefProjection}}}{=} & \gtc{G} \braces{ \rec X \gtc. \gtc{G} / X } \wrt (p,q)^{\mbb L}
        \\
        \overset{\text{\cref{d:unfoldGT}}}{=} & \unfold_1( \rec X \gtc. \gtc{G} ) \wrt (p,q)^{\mbb L}
        \tag*{\qedhere}
    \end{align*}
\end{proof}

\begin{lemma}\label{l:unfoldRT}
    For any well-formed global type $\gtc{G}$, participants $p,q \in \prt(\gtc{G})$, and set of locations $\mbb L$,
    \[
        \unfold( \gtc{G} \wrt (p,q)^{\mbb L} ) = \unfold( \gtc{G} ) \wrt (p,q)^{\mbb L}.
    \]
\end{lemma}

\begin{proof}
    By induction on the number of recursive definitions that $\gtc{G}$ starts with (finite by \nameref{d:wf}).
    In the base case, the thesis follows trivially.
    In the inductive case, $\gtc{G} = \rec X \gtc. \gtc{G'}$:
    \begin{align*}
        & \unfold( \rec X \gtc. \gtc{G'} \wrt (p,q)^{\mbb L} )
        \\
        \overset{\text{\cref{d:relativeProjectionWithLocs}}}{=} & \unfold( \rec X^{\mbb L} \rtc. (\gtc{G'} \wrt (p,q)^{\mbb L}) )
        \\
        \overset{\text{\cref{d:unfoldRT}}}{=} & \unfold_1( \rec X^{\mbb L} \rtc. \unfold( \gtc{G'} \wrt (p,q)^{\mbb L} ) )
        \\
        \overset{\text{IH}}{=} & \unfold_1( \rec X^{\mbb L} \rtc. (\unfold( \gtc{G'} ) \wrt (p,q)^{\mbb L}) )
        \\
        \overset{\text{\cref{d:relativeProjectionWithLocs}}}{=} & \unfold_1( \rec X \gtc. \unfold( \gtc{G'} ) \wrt (p,q)^{\mbb L} )
        \\
        \overset{\text{\cref{l:unfoldRecursiveGTRT}}}{=} & \unfold_1( \rec X \gtc. \unfold( \gtc{G'} ) ) \wrt (p,q)^{\mbb L}
        \\
        \overset{\text{\cref{d:unfoldGT}}}{=} & \unfold( \rec X \gtc. \gtc{G'} ) \wrt (p,q)^{\mbb L}
        \tag*{\qedhere}
    \end{align*}
\end{proof}

\begin{definition}[Unfold Monitor]\label{d:unfoldMon}
    Given a monitor $\rec X \mc. \mc{M}$, we define its \emph{one-level unfolding}: $\unfold_1( \rec X \mc. \mc{M} ) := \mc{M} \braces{ \rec X \mc. \mc{M} / X }$.
    Given a monitor $\mc{M}$, we define its \emph{(full) unfolding}, denoted $\unfold( \mc{M} )$, as follows:
    \[
        \unfold( \mc{M} ) := \begin{cases}
            \unfold_1( \rec X \mc. \unfold( \mc{M} ) )
            & \text{if $\mc{M} = \rec X \mc. \mc{M'}$}
            \\
            \mc{M}
            & \text{otherwise}
        \end{cases}
    \]
\end{definition}

\begin{lemma}\label{l:monitorUnfoldOne}
    Suppose given a well-formed global type $\gtc{G} = \rec X \gtc. \gtc{G'}$, a participant $p$, and a set of locations $\mbb L$.
    Let $D := \braces{ q \in \prt(\gtc{G}) \setminus \braces{p} \mid \gtc{G} \wrt (p,q)^{\mbb L} \neq \pend }$.
    Then
    \[
        \unfold_1( \rec X \mc. \gtToMon( \gtc{G'} , p , D ) ) = \gtToMon( \unfold_1( \rec X \gtc. \gtc{G'} ) , p , D ).
    \]
\end{lemma}

\begin{proof}
    By definition.
    The recursive calls in $\mc{M'}$ concur with the recursive calls in $\gtc{G'}$, and are prefixed by all the exchanges in $\gtc{G'}$ in which $p$ is involved.
\end{proof}

\begin{lemma}\label{l:activeUnfold}
    Suppose given
    \begin{itemize}
        \item a well-formed global type $\rec X \gtc. \gtc{G}$,
        \item a set of participants $D$,
        \item a participant $p \notin D$, and
        \item a set of locations $\mbb L$.
    \end{itemize}
    Let $E_1 := \braces{ q \in D \mid (\rec X \gtc. \gtc{G}) \wrt (p,q)^{\mbb L} \neq \pend }$ and $E_2 := \braces{ q \in D \mid (\gtc{G} \braces{ \rec X \gtc. \gtc{G} / X }) \wrt (p,q)^{\mbb L} \neq \pend }$.

    Then $E_1 = E_2$.
\end{lemma}

\begin{proof}
    For any $q \in E_1$, $(\mu X \gtc. \gtc{G}) \wrt (p,q)^{\mbb L} \neq \pend$.
    Then, by definition, $G \wrt (p,q)^{\mbb L} \neq \pend$.
    Hence, the projection of the unfolding of $\gtc{G}$ is also not $\pend$, and thus $q \in E_2$.

    For any $q \in E_2$, the projection of the unfolding of $\gtc{G}$ is not $\pend$.
    Then, by definition, $\gtc{G} \wrt (p,q)^{\mbb L} \neq \pend$.
    Hence, by definition, $(\mu X \gtc. \gtc{G}) \wrt (p,q)^{\mbb L} \neq \pend$, and thus $q \in E_1$.
\end{proof}

\begin{lemma}\label{l:monitorDropInactive}
    Suppose given
    \begin{itemize}
        \item a well-formed global type $\gtc{G}$,
        \item a set of participants $D$,
        \item a participant $p \notin D$, and
        \item a set of locations $\mbb L$.
    \end{itemize}
    Let $D' := \braces{ q \in D \mid \gtc{G} \wrt (p,q)^{\mbb L} \neq \pend }$.

    Then $\gtToMon( \gtc{G} , p , D ) = \gtToMon( \gtc{G} , p , D' )$.
\end{lemma}

\begin{proof}
    By definition.
    Since $p$ does not interact with any $q \in D \setminus D'$, only the $q \in D'$ affect the creation of the monitor.
\end{proof}

\begin{lemma}\label{l:monitorUnfold}
    Suppose given
    \begin{itemize}
        \item a well-formed global type $\gtc{G}$,
        \item a participant $p$, and
        \item a set of locations $\mbb L$.
    \end{itemize}
    Let $D := \prt(\gtc{G}) \setminus \braces{p}$.

    Then $\unfold( \gtToMon( \gtc{G} , p , D ) ) = \gtToMon( \unfold( \gtc{G} ) , p , D )$.
\end{lemma}

\begin{proof}
    By induction on the number of recursive definitions that $\gtc{G}$ starts with (finite by \nameref{d:wf}).
    In the base case, the thesis follows trivially.
    In the inductive case, $\gtc{G} = \rec X \gtc. \gtc{G'}$.
    Let $D' := \braces{ q \in D \mid \gtc{G'} \wrt (p,q)^\epsi \neq \pend }$.
    \begin{align*}
        & \unfold( \gtToMon( \rec X \gtc. \gtc{G'} , p , D ) )
        \\
        \overset{\text{\cref{d:gtToMon}}}{=} & \unfold( \rec X \mc. \gtToMon( \gtc{G'} , p , D' ) )
        \\
        \overset{\text{\cref{d:unfoldMon}}}{=} & \unfold_1( \rec X \mc. \unfold( \gtToMon( \gtc{G'} , p , D' ) ) )
        \\
        \overset{\text{IH}}{=} & \unfold_1( \rec X \mc. \gtToMon( \unfold( \gtc{G'} ) , p , D' ) )
        \\
        \overset{\text{\cref{l:monitorUnfoldOne}}}{=} & \gtToMon( \unfold_1 ( \rec X \gtc. \unfold( \gtc{G'} ) ) , p , D' )
        \\
        \overset{\text{\cref{d:unfoldGT}}}{=} & \gtToMon( \unfold( \rec X \gtc. \gtc{G'} ) , p , D' )
        \\
        \overset{\text{\cref{l:activeUnfold,l:monitorDropInactive}}}{=} & \gtToMon( \unfold( \rec X \gtc. \gtc{G'} ) , p , D )
        \tag*{\qedhere}
    \end{align*}
\end{proof}

\begin{lemma}\label{l:monitorUnfoldEnd}
    Suppose given
    \begin{itemize}
        \item a well-formed global type $\gtc{G} = \rec X \gtc. \gtc{G'}$,
        \item a set of participants $D$, and
        \item a participant $p \notin D$.
    \end{itemize}
    If $\gtToMon( \gtc{G} , p , D ) = \pend$, then $\gtToMon( \gtc{G'}, p , D ) \braces{ \gtToMon( \gtc{G'} , p , D ) / X } = \pend$.
\end{lemma}

\begin{proof}
    Since $\gtToMon( \gtc{G} , p , D ) = \pend$, we have that $\gtc{G} \wrt (p,q)^{\braces{\epsi}} = \pend$ for every $q \in D$.
    It follows by definition that $\gtToMon( \gtc{G'} , p , D ) = \pend$.
    Hence, the unfolding is also $\pend$.
\end{proof}

\subsection{Proof of Soundness}
\label{s:soundnessProof}

Here we prove \cref{t:soundness}, which is a generalized version of \cref{t:soundnessMain} (Page \pageref{t:soundnessMain}).
We start with an overview of intermediate results used for the proof:
\begin{itemize}

    \item
        \Cref{l:independence} shows that transitions that do not affect parallel networks are independent, i.e., they can be executed in any order without changing the outcome.

    \item
        \Cref{l:satisfactionDropLabel} shows that we can empty the $\Lbls$ map of a satisfaction relation if the locations of all the relative types in $\RTs$ succeed all labels in $\dom(\Lbls)$, i.e., no exchanges in relative types in $\RTs$ relate to the choices recorded in $\Lbls$.

    \item
        \Cref{l:satisfactionDropUnion} shows that we can eliminate unions of relative types in the $\RTs$ map of a satisfaction relation for participants that do not depend on some exchange, allowing us to specifize those independent relative types to some chosen branch.

    \item
        \Cref{d:initSatisfied} defines a relation between global type, participants, relative types, monitor, blackbox, and buffer, such that together they are witness to a satisfaction relation.

    \item
        \Cref{l:soundnessGeneralized} shows an inductive variant of soundness, given an intermediate global type $\gtc{G_0}$ between the initial $\gtc{G}$ and the final $\gtc{G'}$.

    \item
        \Cref{t:soundness} shows soundness.

\end{itemize}

\begin{lemma}[Independence]\label{l:independence}
    Suppose $\net P_1 \| \net Q \| \net R \ltrans{ \tau } \net P'_1 \| \net Q' \| \net R$ and $\net P_2 \| \net Q \| \net R \ltrans{ \tau } \net P'_2 \| \net Q \| \net R'$, where $\net P_1 \not\equiv \net P'_1$, $\net P_2 \not\equiv \net P'_2$, $\net Q \not\equiv \net Q'$, and $\net R \not\equiv \net R'$.
    Then
    \[
        \begin{array}{@{}c@{}c@{}c@{}}
            \net P_1 \| \net P_2 \| \net Q \| \net R
            &
            {} \ltrans{ \tau } {}
            &
            \net P'_1 \| \net P_2 \| \net Q' \| \net R
            \\
            \downarrow \mkern-4mu {\scriptstyle \tau}
            &
            &
            \downarrow \mkern-4mu {\scriptstyle \tau}
            \\
            \net P_1 \| \net P'_2 \| \net Q \| \net R'
            &
            {} \ltrans{ \tau } {}
            &
            \net P'_1 \| \net P'_2 \| \net Q' \| \net R'.
        \end{array}
    \]
\end{lemma}

\begin{proof}
    By definition of the \nameref{d:ltsNetworks}, the $\tau$-transitions of both premises are derived from applications of Transition~\ruleLabel{par} and an application of Transition~\ruleLabel{out-buf} or Transition~\ruleLabel{out-mon-buf}.
    For example, in the first premise, $\net P_1$ does an output which ends up in the buffer of a (monitored) blackbox in $\net Q$, leaving $\net R$ unchanged.
    In the second premise, $\net P_2$ does an output which ends up in a buffer in $\net R$, leaving $\net Q$ unchanged.
    Hence, the outputs by $\net P_1$ and $\net P_2$ have completely different senders and recipients.
    As a result, in a network with all of $\net P_1$, $\net P_2$, $\net Q$, and $\net R$ these exchanges do not influence each other.
    The conclusion is that the order of these exchanges does not matter.
\end{proof}

\begin{lemma}\label{l:satisfactionDropLabel}
    Suppose
    \begin{itemize}
        \item $\satisfies{ \net P }[ \Lbls ]{ \RTs }[ p ]$ and
        \item that, for every $( q , \rtc{R} ) \in \RTs$, $\rtc{R} \neq \pend$ implies $\big( \bigcup_{( \mbb L , \ell ) \in \Lbls} \mbb L \big) \prefix \fstLoc(\rtc{R})$.
    \end{itemize}
    Then $\satisfies{ \net P }{ \RTs }[ p ]$.
\end{lemma}

\begin{proof}
    By definition: clearly, if all the relative types are at a location past the recorded choices in $\Lbls$, none of the choices in $\Lbls$ will ever be used for satisfaction anymore.
\end{proof}

As satisfaction iterates through a collection of relative types obtained from a global type, there will be instances where some relative types are independent of the global type's initial exchange.
By \nameref{d:relativeProjectionWithLocs}, these relative types are the union of the relative projections of the branches of the exchange.
The following lemma assures that we can drop this union and simply continue with the relative projection of the branch followed by the participants involved in the exchange (sender, recipient, and/or depending participants).

\begin{lemma}\label{l:satisfactionDropUnion}
    Suppose given
    \begin{itemize}
        \item a well-formed global type $\gtc{G} = s \send r \gtBraces{ \msg i<T_i> \gtc. \gtc{G_i} }_{i \in I}$, and
        \item $D \supseteq \prt(\gtc{G}) \setminus \braces{p}$.
    \end{itemize}
    Let $D' := \braces{q \in D \mid \braces{p,q} \subseteq \braces{s,r}} \cup \braces{q \in D \mid \braces{p,q} \overlap \braces{s,r} \wedge (\depsOn q p \gtc{G} \vee \depsOn p q \gtc{G})}$.

    Suppose $\satisfies{ \net P }{ \RTs }[ p ]$, where $\RTs = \braces{ ( q , \gtc{G_j} \wrt (p,q)^{\braces{\vect \ell,j}} ) \mid q \in D' } \cup \braces{ (q , \bigcup_{i \in I} (\gtc{G_i} \wrt (p,q)^{\braces{\vect \ell,i}}) ) \mid q \in D \setminus D' }$ for some $j \in I$.
    Let $\RTs' := \RTs \update{q \mapsto \gtc{G_j} \wrt (p,q)^{\braces{\vect \ell,j}}}_{q \in D \setminus D'}$.

    Then $\satisfies{ \net P }{ \RTs' }[ p ]$.
\end{lemma}

\begin{proof}
    For any $q \in D \setminus D'$, for every $i,k \in I$, $\erase(\gtc{G_i} \wrt (p,q)^{\braces{\vect \ell,k}}) = \erase(\gtc{G_k} \wrt (p,q)^{\braces{\vect \ell,k}})$.
    Since the satisfaction holds for an empty label function (i.e., $\satisfies{ \net P }[ \emptyset ]{ \RTs }[ p ]$), the locations formed in each $\gtc{G_i} \wrt (p,q)^{\braces{\vect \ell,i}}$ are insignificant.
    Hence, it suffices to simply use $\gtc{G_j} \wrt (p,q)^{\braces{\vect \ell,j}}$ for every $q \in D \setminus D'$.
\end{proof}

\begin{definition}[Initial satisfaction]\label{d:initSatisfied}
    We define \emph{initial satisfaction}, denoted $( \gtc{G} , p , \mbb L , D ) \initSatisfied ( \RTs , D' , \mc{M} , P , \vect m )$, to hold if and only if
    \begin{align*}
        D &\supseteq \prt(\gtc{G}) \setminus \braces{ p },
        \\
        \RTs &= \braces{ ( q , \gtc{G} \wrt (p,q)^{\mbb L} ) \mid q \in D },
        \\
        D' &= \braces{ q \in D \mid \unfold\big(\RTs(q)\big) \neq \pend },
        \\
        \gtToMon( \gtc{G} , p , D' ) \neq \pend
        &\implies \mc{M} = \gtToMon( \gtc{G} , p , D' )
        \\
        \gtToMon( \gtc{G}, p , D' ) = \pend
        &\implies \mc{M} \in \braces{ \pend , \checkmark }
        \\
        &\satisfies{ \monImpl{ \bufImpl{ p }{ P }{ \vect m } }{ \mc{M} }{ \epsi } }{ \RTs }[ p ]
    \end{align*}
\end{definition}

\begin{lemma}[Soundness --- Generalized]\label{l:soundnessGeneralized}
    \begin{itemize}
        \item Suppose given well-formed global types $\gtc{G}$ and $\gtc{G_0}$ such that $\prt(\gtc{G}) \geq 2$ and $\gtc{G} \ltrans{ \vect \ell_0 } \gtc{G_0}$.
        \item For every $p \in \prt(\gtc{G})$,
            \begin{itemize}
                \item suppose given a process $P_p^0$ and buffer $\vect m_p^0$, and
                \item take $\RTs_p^0 , D_p^0 , \mc{M_p^0}$ such that
                    \[
                        ( \gtc{G_0} , p , \braces{ \vect \ell_0 } , \prt(\gtc{G}) \setminus \braces{p} ) \initSatisfied ( \RTs_p^0 , D_p^0 , \mc{M_p^0} , P_p^0 , \vect m_p^0 ).
                    \]
            \end{itemize}
        \item Suppose $\net P_0 := \prod_{p \in \prt(\gtc{G})} \monImpl{ \bufImpl{ p }{ P_p^0 }{ \vect m_p^0 } }{ \mc{M_p^0} }{ \epsi } \trans* \net P'$.
    \end{itemize}
    Then there exist $\vect \ell',\gtc{G'}$ such that
    \begin{itemize}
        \item $\gtc{G_0} \ltrans{ \vect \ell' } \gtc{G'}$,
        \item for every $p \in \prt(\gtc{G})$ there exist $P'_p,\vect m'_p,\RTs'_p,D'_p,\mc{M'_p}$ such that
            \[
                ( \gtc{G'} , p , \braces{ \vect \ell_0 , \vect \ell' } , \prt(\gtc{G}) \setminus \braces{p} ) \initSatisfied ( \RTs'_p , D'_p , \mc{M'_p} , P'_p , \vect m'_p ),
            \]
            and
        \item $\net P' \trans* \prod_{p \in \prt(\gtc{G})} \monImpl{ \bufImpl{ p }{ P'_p }{ \vect m'_p } }{ \mc{M'_p} }{ \epsi }$.
    \end{itemize}
\end{lemma}

\begin{proof}
    By induction on the number of transitions $n$ from $\net P_0$ to $\net P'$ (\ih{1}).

    In the base case, where $n = 0$, the thesis follows immediately, because $\net P' = \net P_0$.
    To be precise, the assumption satisfies the conclusion by letting: $\gtc{G'} = \gtc{G_0}$; $\vect \ell = \epsi$; for every $p \in \prt(\gtc{G})$, $P'_p = P_p^0$, $\vect m'_p = \vect m_p^0$, $D'_p = D_p^0$, $\mc{M'_p} = \mc{M_p^0}$, $\RTs'_p = \RTs_p^0$.

    In the inductive case, where $n \geq 1$, we use the shape of $\gtc{G_0}$ and the assumed satisfactions of the monitored processes to determine the possible transitions from $\net P_0$.
    We then follow these transitions, and show that we can reach a network $\net P_1$ where all the monitored processes satisfy some $\gtc{G_1}$ with $\gtc{G_0} \ltrans{ \ell_1 } \gtc{G_1}$.
    If at this point we already passed through $\net P'$, the thesis is proven.
    Otherwise, the thesis follows from \ih{1}, because the number of transitions from $\net P_1$ is less than $n$.

    There is a subtlety that we should not overlook: $\gtc{G_0}$ may contain several consecutive, independent exchanges.
    For example, suppose $\gtc{G_0} = p \send q \gtBraces{ \msg i<T_i> \gtc. s \send r \gtBraces{ \msg j <T_j> \gtc. \gtc{G_{i,j}} }_{j \in J} }_{i \in I}$ where $\braces{p,q} \disjoint \braces{s,r}$.
    Monitors nor satisfaction can prevent the exchange between $s$ and $r$ from happening before the exchange between $p$ and $q$ has been completed.
    Hence, the transitions from $\net P_0$ to $\net P'$ may not entirely follow the order specified by $\gtc{G_0}$.

    We deal with this issue by applying induction on the number of out-of-order exchanges observed in the transitions to $\net P'$.
    We then follow the transitions determined by monitors and satisfaction, in the order specified by $\gtc{G_0}$, effectively ``postponing'' the out-of-order exchanges until it is their turn.
    We keep doing this, until we have eventually passed through all the postponed out-of-order exchanges.
    At this point, we have found an alternative path from $\net P_0$ to the final network.
    To reconcile this alternative path with the path from $\net P_0$ to $\net P'$, we apply \nameref{l:independence}.
    This lemma essentially states that independent exchanges may be performed in any order, as they do not influence each other.
    Hence, we use \nameref{l:independence} to move the postponed transitions back to their original position in the path from $\net P_0$ to $\net P'$, proving the thesis.
    Hereafter, we assume independent exchanges dealt with.

    As a first step in our analysis, we consider the fact that $\gtc{G_0}$ may start with recursive definitions.
    Let $\gtc{G_1} := \unfold(\gtc{G_0})$; by definition, $\gtc{G_1}$ does not start with recursive definitions.
    For every $p \in \prt(\gtc{G})$, let $\mc{M_p^1} := \unfold(\mc{M_p^0})$; by \cref{l:monitorUnfold}, $\mc{M_p^1} = \gtToMon(\gtc{G_1},p,D_p^0)$, and, by definition of the \nameref{d:ltsNetworks}, the LTS of $\monImpl{ \bufImpl{ p }{ P_p^0 }{ \vect m_p^0 } }{ \mc{M_p^1} }{ \epsi }$ is equivalent to that of $\monImpl{ \bufImpl{ p }{ P_p^0 }{ \vect m_p^0 } }{ \mc{M_p^0} }{ \epsi }$.
    For every $p \in \prt(\gtc{G})$, let $\RTs_p^1 := \braces{ ( q , \unfold(\RTs_p^0(q)) ) \mid q \in D_p^0 }$; by \cref{l:unfoldRT}, for every $q \in D_p^0$, $\RTs_p^1(q) = \gtc{G_1} \wrt (p,q)$.
    The conditions for Satisfaction in \cref{f:satisfaction} unfold any relative types.
    Hence, we can reuse the satisfaction given by the original initial satisfaction, to show that $(\gtc{G_1},p,\braces{\vect l_0},\prt(\gtc{G}) \setminus \braces{p}) \initSatisfied (\RTs_p^1,D_p^0,\mc{M_p^1},P_p^0,\vect m_p^0)$.
    We then continue our analysis from this new unfolded initial satisfaction, for which all results transfer back to the original initial satisfaction.

    The rest of our analysis depends on the shape of $\gtc{G_1}$ (exchange or end).
    \begin{itemize}
        \item
            Exchange: $\gtc{G_1} = s \send r \gtBraces{ \msg i<T_i> \gtc. \gtc{H_i} }_{i \in I}$.

            Let us make an inventory of all relative types and monitors at this point.
            We use this information to determine the possible behavior of the monitored blackboxes, and their interactions.
            This behavior, in combination with satisfaction, allows us to determine exactly how the network evolves and reaches a state required to apply \ih1.

            \begin{itemize}
                \item
                    For $s$ we have:
                    \begin{align*}
                        \RTs_s^1(r)
                        &= s \send r^{\braces{\vect \ell_0}} \rtBraces{ \msg i<T_i> \rtc. (\gtc{H_i} \wrt (s,r)^{\braces{\vect \ell_0,i}}) }_{i \in I}
                        \\
                        \depsVar_s
                        &= \braces{ q \in \prt(\gtc{G}) \mid \depsOn q s \gtc{G_1} }
                        \\
                        \RTs_s^1(q)
                        &= (s \send r) \send q^{\braces{\vect \ell_0}} \rtBraces{ i \rtc. (\gtc{H_i} \wrt (s,q)^{\braces{\vect \ell_0,i}}) }_{i \in I}
                        \quad [q \in \depsVar_s]
                        \\
                        \RTs_s^1(q)
                        &= \bigcup_{i \in I} (\gtc{H_i} \wrt (s,q)^{\braces{\vect \ell_0,i}})
                        \quad [q \in \prt(\gtc{G}) \setminus \braces{s,r} \setminus \depsVar_s]
                        \\
                        \mc{M_s^1}
                        &= s \send r \mBraces{ \msg i<T_i> \mc. s \send \depsVar_s (i) \mc. \gtToMon( \gtc{H_i} , s , D_s^0 ) }_{i \in I}
                    \end{align*}

                    \satref{i:sOutput}{Output} allows the monitored blackbox of $s$ to send $\msg j<T_j>$ for any $j \in I$ to $r$.
                    Then \satref{i:sDependencyOutput}{Dependency Output} allows the monitored blackbox to send $j$ to all $q \in \depsVar_s$ (concurrently).
                    However, there may be other relative types in $\RTs_s^1$ that allow/require the monitored blackbox of to perform other tasks.

                    Since the outputs above precede any other communications in $\RTs_s^1$ (they originate from the first exchange in $\gtc{G_1}$), by the progress property of \nameref{d:satisfaction}, the monitored blackbox will keep transitioning.
                    The monitor $\mc{M_s^1}$ requires the blackbox to first perform the output above, and then the dependency outputs above.
                    The buffered blackbox does not take any transitions other than $\tau$ and the outputs above: otherwise, the monitor would transition to an error signal, which cannot transition, contradicting the satisfaction's progress property.

                    It then follows that the monitored blackbox sends $\msg j<T_j>$ to $r$ for some $j \in I$, after which it sends $j$ to each $q \in \depsVar_s$ (in any order), possibly interleaved with $\tau$-transitions from the blackbox (finitely many, by \bref{i:bTau}{Finite $\tau$}) or from the buffered blackbox reading messages; let $P_s^2$ and $\vect m_s^2$ be the resulting blackbox and buffer, respectively.
                    By \nameref{d:satisfaction}, we have the following:
                    \begin{align*}
                        \RTs_s^2
                        &:= \RTs_s^1
                        \update{q \mapsto \gtc{H_j} \wrt (s,q)^{\braces{\vect \ell_0,j}}}_{q \in \braces{r} \cup \depsVar_s}
                        \\
                        \mc{M_s^2}
                        &:= \gtToMon( \gtc{H_j} , s , D_s^0 )
                        \\
                        &\satisfies{ \monImpl{ \bufImpl{ s }{ P_s^2 }{ \vect m_s^2 } }{ \mc{M_s^2} }{ \epsi } }[ \braces{ (\braces{\vect \ell_0} , j) } ]{ \RTs_s^2 }[ s ]
                    \end{align*}
                    Clearly, for each $q \in \dom(\RTs_s^2)$, $\vect \ell_0 \leq \fstLoc( \RTs_s^2(q) )$.
                    Hence, by \cref{l:satisfactionDropLabel}, $\satisfies{ \monImpl{ \bufImpl{ s }{ P_s^2 }{ \vect m_s^2 } }{ \mc{M_s^2} }{ \epsi } }{ \RTs_s^2 }[ s ]$.
                    Let $\RTs_s^3 := \RTs_s^2 \update{q \mapsto \gtc{H_j} \wrt (r,q)^{\braces{\vect \ell_0,j}}}_{q \in D \setminus \depsVar_s \setminus \braces{s,r}}$.
                    Then, by \cref{l:satisfactionDropUnion}, $\satisfies{ \monImpl{ \bufImpl{ s }{ P_s^2 }{ \vect m_s^2 } }{ \mc{M_s^2} }{ \epsi } }{ \RTs_s^3 }[ s ]$.
                    Let $D_s^3 := \braces{ q \in \prt(\gtc{G}) \setminus \braces{s} \mid \unfold\big(\RTs_s^3(q)\big) \neq \pend }$ and $M_s^3 := \gtToMon( \gtc{H_j} , p , D_s^1 )$.
                    By \cref{l:monitorDropInactive}, $\mc{M_s^2} = \mc{M_s^3}$.
                    We have $\gtc{G_1} \ltrans{j} \gtc{H_j}$.
                    In conclusion,
                    \[
                        ( \gtc{H_j} , s , \braces{\vect \ell_0,j} , \prt(\gtc{G}) \setminus \braces{s} ) \initSatisfied ( \RTs_s^3 , D_s^3 , \mc{M_s^3} , P_s^2 , \vect m_s^2 ),
                    \]
                    such that the premise of \ih1 is satisfied for $s$.

                \item
                    For $r$ we have:
                    \begin{align*}
                        \RTs_r^1(s)
                        &= s \send r^{\braces{\vect \ell_0}} \rtBraces{ \msg i<T_i> \rtc. (\gtc{H_i} \wrt (r,s)^{\braces{\vect \ell_0,i}}) }_{i \in I}
                        \\
                        \depsVar_r
                        &= \braces{ q \in \prt(\gtc{G}) \mid \depsOn q r \gtc{G_1} }
                        \\
                        \RTs_r^1(q)
                        &= (r \recv s) \send q^{\braces{\vect \ell_0}} \rtBraces{ i \rtc. (\gtc{H_i} \wrt (r,q)^{\braces{\vect \ell_0,i}}) }_{i \in I}
                        \quad [q \in \depsVar_r]
                        \\
                        \RTs_r^1(q)
                        &= \bigcup_{i \in I} (\gtc{H_i} \wrt (r,q)^{\braces{\vect \ell_0,i}})
                        \quad [q \in \prt(\gtc{G}) \setminus \braces{s,r} \setminus \depsVar_r]
                        \\
                        \mc{M_r^1}
                        &= r \recv s \mBraces{ \msg i<T_i> \mc. r \send \depsVar_r (i) \mc. \gtToMon( \gtc{H_i} , r , D_r^0 ) }_{i \in I}
                    \end{align*}

                    Since $s$ sends $\msg j<T_j>$ to $r$, by Transition~\ruleLabel{out-mon-buf}, this message will end up in the buffer of the monitored blackbox of $r$.
                    The monitor $\mc{M_r^1}$ moves this message to the blackbox's buffer, and proceeds to send dependency messages $j$ to all $q \in \depsVar_r$ (concurrently).
                    Following the same reasoning as above, the monitored blackbox will keep outputting the dependencies above, possibly interleaved with $\tau$-transitions (from the blackbox or from the buffered blackbox reading messages), before doing anything else.

                    Let $P_r^2$ and $\vect m_r^2$ be the resulting blackbox and buffer, respectively.
                    By \nameref{d:satisfaction}, we have the following (applying \cref{l:satisfactionDropLabel,l:satisfactionDropUnion,l:monitorDropInactive} immediately):
                    \begin{align*}
                        \RTs_r^2
                        &:= \braces{ ( q , \gtc{H_j} \wrt (r,q)^{\braces{\vect \ell_0,j}} ) \mid q \in \prt(\gtc{G}) \setminus \braces{r} }
                        \\
                        D_r^2
                        &:= \braces{ q \in \prt(\gtc{G}) \setminus \braces{s} \mid \unfold\big(\RTs_r^2(q)\big) \neq \pend }
                        \\
                        \mc{M_r^2}
                        &:= \gtToMon( \gtc{H_j} , r , D_r^2 )
                        \\
                        &\satisfies{ \monImpl{ \bufImpl{ r }{ P_r^2 }{ \vect m_r^2 } }{ \mc{M_r^2} }{ \epsi } }{ \RTs_r^2 }[ r ]
                    \end{align*}
                    In conclusion,
                    \[
                        ( \gtc{H_j} , r , \braces{\vect \ell_0,j} , \prt(\gtc{G}) \setminus \braces{r} ) \initSatisfied ( \RTs_r^2 , D_r^2 , \mc{M_r^2} , P_r^2 , \vect m_r^2 ),
                    \]
                    such that the premise of \ih1 is satisfied for $r$.

                \item
                    For every $q \in \depsVar_s \setminus \depsVar_r$ we have:
                    \begin{align*}
                        \RTs_q^1(s)
                        &= (s \send r) \send q^{\braces{\vect \ell_0}} \rtBraces{ i \rtc. (\gtc{H_i} \wrt (q,s)^{\braces{\vect \ell_0}+i}) }_{i \in I}
                        \\
                        \RTs_q^1(q')
                        &= \bigcup_{i \in I} (\gtc{H_i} \wrt (q,q')^{\braces{\vect \ell_0}+i})
                        \quad [q' \in \prt(\gtc{G}) \setminus \braces{q,s}]
                        \\
                        \mc{M_q^1}
                        &= q \recv s \mBraces{ i \mc. \gtToMon( \gtc{H_i} , q , D_q^0 ) }_{i \in I}
                    \end{align*}

                    Since $s$ sends $j$ to $q$, by Transition~\ruleLabel{out-mon-buf}, this message will end up in the buffer of the monitored blackbox of $q$.
                    The monitor $\mc{M_q^1}$ moves this message to the blackbox's buffer.

                    Let $P_q^2$ and $\vect m_q^2$ be the resulting blackbox and buffer, respectively.
                    By \nameref{d:satisfaction}, we have the following (applying \cref{l:satisfactionDropLabel,l:satisfactionDropUnion,l:monitorDropInactive} immediately):
                    \begin{align*}
                        \RTs_q^2
                        &:= \braces{ ( q' , \gtc{H_j} \wrt (q,q')^{\braces{\vect \ell_0,j}} ) \mid q' \in \prt(\gtc{G}) \setminus \braces{q} }
                        \\
                        D_q^2
                        &:= \braces{ q' \in \prt(\gtc{G}) \setminus \braces{q} \mid \unfold\big(\RTs_q^2(q')\big) \neq \pend }
                        \\
                        \mc{M_q^2}
                        &:= \gtToMon( \gtc{H_j} , q , D_q^2 )
                        \\
                        &\satisfies{ \monImpl{ \bufImpl{ q }{ P_q^2 }{ \vect m_q^2 } }{ \mc{M_q^2} }{ \epsi } }{ \RTs_q^2 }[ q ]
                    \end{align*}
                    In conclusion,
                    \[
                        ( \gtc{H_j} , q , \braces{\vect \ell_0,j} , \prt(\gtc{G}) \setminus \braces{q} ) \initSatisfied ( \RTs_q^2 , D_q^2 , \mc{M_q^2} , P_q^2 , \vect m_q^2 ),
                    \]
                    such that the premise of \ih1 is satisfied for $q$.

                \item
                    For every $q \in \depsVar_r \setminus \depsVar_s$, the procedure is similar to above.

                \item
                    For every $q \in \depsVar_s \cap \depsVar_r$ we have:
                    \begin{align*}
                        \RTs_q^1(s)
                        &= (s \send r) \send q^{\braces{\vect \ell_0}} \rtBraces{ i \rtc. (\gtc{H_i} \wrt (q,s)^{\braces{\vect \ell_0}+i}) }_{i \in I}
                        \\
                        \RTs_q^1(r)
                        &= (r \recv s) \send q^{\braces{\vect \ell_0}} \rtBraces{ i \rtc. (\gtc{H_i} \wrt (q,r)^{\braces{\vect \ell_0}+i}) }_{i \in I}
                        \\
                        \RTs_q^1(q')
                        &= \bigcup_{i \in I} (\gtc{H_i} \wrt (q,q')^{\braces{\vect \ell_0}+i})
                        \quad [q' \in \prt(\gtc{G}) \setminus \braces{q,s,r}]
                        \\
                        \mc{M_q^1}
                    &= q \recv s \mBraces[\big]{ i \mc. q \recv r \mBraces{ i \mc. \gtToMon( \gtc{H_i} , q , D_q^0 ) } \cup \mBraces{ j \mc. \perror }_{j \in I \setminus \braces{i}} }_{i \in I}
                    \end{align*}

                    By Transition~\ruleLabel{out-mon-buf}, from $s$ and $r$ the message $j$ will end up in the buffer of the monitored blackbox of $q$.
                    The monitor $\mc{M_q^1}$ first moves the message from $s$ to the blackbox's buffer, and then the message from $r$.
                    Even if the message from $r$ is the first to end up in the monitor's buffer, this order of reception is enforced because buffers allow exchange of message from different senders.
                    Because $s$ and $r$ send the same $j$, the monitor will not reach the error state.

                    Let $P_q^2$ and $\vect m_q^2$ be the resulting blackbox and buffer, respectively.
                    By \nameref{d:satisfaction}, we have the following (applying \cref{l:satisfactionDropLabel,l:satisfactionDropUnion,l:monitorDropInactive} immediately):
                    \begin{align*}
                        \RTs_q^2
                        &:= \braces{ ( q' , \gtc{H_j} \wrt (q,q')^{\braces{\vect \ell_0,j}} ) \mid q' \in \prt(\gtc{G}) \setminus \braces{q} }
                        \\
                        D_q^2
                        &:= \braces{ q' \in \prt(\gtc{G}) \setminus \braces{q} \mid \unfold\big(\RTs_q^2(q')\big) \neq \pend }
                        \\
                        \mc{M_q^2}
                        &:= \gtToMon( \gtc{H_j} , q , D_q^2 )
                        \\
                        &\satisfies{ \monImpl{ \bufImpl{ q }{ P_q^2 }{ \vect m_q^2 } }{ \mc{M_q^2} }{ \epsi } }{ \RTs_q^2 }[ q ]
                    \end{align*}
                    In conclusion,
                    \[
                        ( \gtc{H_j} , q , \braces{\vect \ell_0,j} , \prt(\gtc{G}) \setminus \braces{q} ) \initSatisfied ( \RTs_q^2 , D_q^2 , \mc{M_q^2} , P_q^2 , \vect m_q^2 ),
                    \]
                    such that the premise of \ih1 is satisfied for $q$.

                \item
                    For every $q \in \prt(\gtc{G}) \setminus \braces{s,r} \setminus \depsVar_s \setminus \depsVar_r$ we have:
                    \begin{align*}
                        \RTs_q^1(q')
                        &= \bigcup_{i \in I} (\gtc{H_i} \wrt (q,q')^{\braces{\vect \ell_0}+i})
                        \quad [q' \in \prt(\gtc{G}) \setminus \braces{q}]
                        \\
                        \mc{M_q^1}
                        &= \gtToMon( \gtc{H_k} , q , D_q^0 )
                        \quad [\text{arbitrary } k \in I]
                    \end{align*}
                    If we observe transitions from $q$, we are guaranteed that the behavior is independent of the current exchange between $s$ and $r$ in $\gtc{G_1}$: otherwise, $q \in \depsVar_s \cup \depsVar_r$.
                    Hence, as mentioned before, we can safely postpone these steps.

                    We have (applying \cref{l:satisfactionDropLabel,l:monitorDropInactive}):
                    \begin{align*}
                        \RTs_q^2
                        &:= \braces{ ( q' , \gtc{H_j} \wrt (q,q')^{\braces{\vect \ell_0,j}} ) \mid q' \in \prt(\gtc{G}) \setminus \braces{q} }
                        \\
                        D_q^2
                        &:= \braces{ q' \in \prt(\gtc{G}) \setminus \braces{q} \mid \unfold\big(\RTs_q^2(q')\big) \neq \pend }
                        \\
                        \mc{M_q^2}
                        &:= \gtToMon( \gtc{H_j} , q , D_q^2 )
                        \\
                        &\satisfies{ \monImpl{ \bufImpl{ p }{ P_q^0 }{ \vect m_q^0 } }{ \mc{M_q^2} }{ \epsi } }{ \RTs_q^2 }[ q ]
                    \end{align*}
                    In conclusion,
                    \[
                        ( \gtc{H_j} , q , \braces{\vect \ell_0,j} , \prt(\gtc{G}) \setminus \braces{q} ) \initSatisfied ( \RTs_q^2 , D_q^2 , \mc{M_q^2} , P_q^2 , \vect m_q^2 ),
                    \]
                    such that the premise of \ih1 is satisfied for $q$.
            \end{itemize}

            At this point, the entire premise of \ih1 is satisfied.
            Hence, the thesis follows by \ih1.

        \item
            End: $\gtc{G_0} = \pend$.

            For every $p \neq q \in \prt(\gtc{G})$, we have $\RTs_p^1(q) = \pend$.
            By the definition of \nameref{d:initSatisfied}, for each $p \in \prt(\gtc{G})$, $\mc{M_p^1} \in \braces{\pend,\checkmark}$.
            By the progress properties of \nameref{d:satisfaction}, each monitored blackbox will continue performing transitions; these can only be $\tau$-transitions, for any other transition leads to an error signal or a violation of Satisfaction.
            That is, each monitored blackbox will be reading messages from buffers or doing internal computations.
            However, there are only finitely many messages, and, by \bref{i:bTau}{Finite $\tau$}, each blackbox is assumed to only perform finitely many $\tau$-transitions in a row.
            Hence, at some point, we must see an $\pend$-transition from each monitored blackbox.

            Suppose the observed transition originates from $p \in \prt(\gtc{G})$.
            Suppose $\mc{M_p^1} = \pend$.
            Then the transition is either labeled $\tau$ or $\pend$.
            If the label is $\tau$, \ih1 applies with premise trivially satisfied.
            If the label is $\pend$, we follow \satref{i:sEnd}{End} to apply \ih1.

            It cannot be that $\mc{M_p^1} = \checkmark$: the \nameref{d:ltsNetworks} does not define any transitions, contradicting the observed transition.
            \qedhere
    \end{itemize}
\end{proof}

\begin{theorem}[Soundness]\label{t:soundness}
    \begin{itemize}
        \item Suppose given a well-formed global type $\gtc{G}$.
        \item For every $p \in \prt(\gtc{G})$, suppose given a process $P_p$, and take $\RTs_p,\mc{M_p}$ such that
            \[
                ( \gtc{G} , p , \braces{\epsi} , \prt(\gtc{G}) \setminus \braces{p} ) \initSatisfied ( \RTs_p , \prt(\gtc{G}) \setminus \braces{p} , \mc{M_p} , P_p , \epsi ).
            \]
        \item Suppose also $\prod_{p \in \prt(\gtc{G})} \monImpl{ \bufImpl{ p }{ P_p }{ \epsi } }{ \mc{M_p} }{ \epsi } \trans* \net P'$.
    \end{itemize}
    Then there exist $\vect \ell',\gtc{G'}$ such that
    \begin{itemize}
        \item $\gtc{G} \ltrans{ \vect \ell' } \gtc{G'}$,
        \item for every $p \in \prt(\gtc{G})$ there exist $P'_p,\vect m'_p,\RTs'_p,D'_p,\mc{M'_p}$ such that
            \begin{itemize}
                \item $( \gtc{G'} , p , \braces{ \ell' } , \prt(\gtc{G}) \setminus \braces{p} ) \initSatisfied ( \RTs'_p , D'_p , \mc{M'_p} , P'_p , \vect m'_p )$, and
                \item $\net P' \trans* \prod_{p \in \prt(\gtc{G})} \monImpl{ \bufImpl{ p }{ P'_p }{ \vect m'_p } }{ \mc{M'_p} }{ \epsi }$.
           \end{itemize}
   \end{itemize}
\end{theorem}

\begin{proof}
    Follows directly from \cref{l:soundnessGeneralized}, given $\gtc{G} \ltrans{ \epsi } \gtc{G}$.
\end{proof}

\begin{theorem}[Error Freedom]\label{t:errorFreedom}
    Suppose given a well-formed global type $\gtc{G}$.
    For every $p \in \prt(G)$, suppose given a process $P_p$, and take $\RTs_p,\mc{M_p}$ such that $( \gtc{G} , p , \braces{ \epsi } , \prt(\gtc{G}) \setminus \braces{p} ) \initSatisfied ( \RTs_p , \prt(\gtc{G}) \setminus \braces{p} , \mc{M_p} , P_p , \epsi )$.
    If $\net P := \prod_{p \in P} \monImpl{ \bufImpl{ p }{ P_p }{ \epsi } }{ \mc{M_p} }{ \epsi } \trans* \net P'$ , then there are no $\net Q,D$ such that $\net P' \trans* \net Q \| \perror_D$.
\end{theorem}

\begin{proof}
    Suppose, towards a contradiction, that $\net P' \trans* \net Q \| \perror_D$.
    Then, by multiple applications of Transition~\ruleLabel{error-par}, $\net Q \| \perror_D \trans* \perror_{D'}$ for some $D'$.
    Hence, $\net P \trans* \perror_{D'}$.
    Then, by \nameref{t:soundness}, $\perror_{D'}$ would further transition.
    However, the LTS of networks does not specify any transitions for error signals: a contradiction.
\end{proof}

\subsection{Proof of Transparency}
\label{s:transparencyProof}

Here we prove \cref{t:transparencyBisim}, which is the full version of \cref{t:transparencyMain} (\cpageref{t:transparencyMain}).
We start with an overview of intermediate results used for the proof:
\begin{itemize}

    \item
        \Cref{d:ltsRelativeTypes} defines an LTS for relative types, where transitions can only be (dependency) outputs.

    \item
        \Cref{d:coherentSetup} defines a relation between a global type, participant, monitor, $\RTs$, and buffer.
        The idea is that the relative types in $\RTs$ may be in an intermediate state between exchanges in the global type, reflected by messages in the buffer.

    \item
        \Cref{l:satisfactionNoError} shows that monitored blackboxes in a satisfaction relation do not transition to an error state through the \nameref{d:enhancedLTS} (\cref{d:enhancedLTS}).

    \item
        \Cref{l:satisfactionReverseBuffer} shows that messages in the buffer of a monitored blackbox of $p$ in a satisfaction relation are related to exchanges to $p$ in $\RTs$ of the satisfaction relation.

    \item
        \Cref{l:satisfactionOutput} shows that an output transitions of a monitored blackbox in a satisfaction relation implies that the monitor is a related output, possibly preceded by a dependency output.

    \item
        \Cref{l:satisfactionEnd} shows that an $\pend$-transition of a monitored blackbox in a satisfaction relation implies that the monitor is $\pend$, possibly preceded by a dependency output.

    \item
        \Cref{l:finiteReceive} shows that a message in the buffer of a monitored blackbox relates to an exchange in a global type, reachable in finitely many steps.

    \item
        \Cref{l:monitorNoOutput} shows that if the blackbox of a monitored blackbox in a satisfaction relation does an input transition, then the monitor is a sequence of inputs ending in an input related to the input transition.

    \item
        \Cref{l:monitorNotCheck} shows that if the blackbox of a monitored blackbox in a satisfaction relation does a $\tau$-transition, then the monitor is not $\mc{\checkmark}$.

    \item
        \Cref{l:labelOracleDepOut} shows that label oracles are equal for relative types prior and past dependency outputs.

    \item
        \Cref{l:LOAlpha} shows a precise relation between actions and label oracles.

    \item
        \Cref{t:transparencyBisim} shows transparency.

\end{itemize}

\begin{definition}[LTS for Relative Types]\label{d:ltsRelativeTypes}
    We define an LTS for relative types, denoted $\rtc{R} \ltrans{ m }_p \rtc{R'}$, where actions are output (dependency) messages $m$ (\cref{f:netGrammars} (top)) and $p$ is the receiving participant, by the following rules:
    \begin{mathparpagebreakable}
        \begin{bussproof}
            \bussAssume{
                j \in I
            }
            \bussUn{
                q \send p^{\mbb L} \rtBraces{ \msg i<T_i> \rtc. \rtc{R_i} }_{i \in I}
                \mathrel{\raisebox{-2.0pt}{$\ltrans{ q \send p (\msg j<T_j>) }_p$}}
                \rtc{R_j}
            }
        \end{bussproof}
        \and
        \begin{bussproof}
            \bussAssume{
                j \in I
            }
            \bussUn{
                (q \lozenge r) \send p^{\mbb L} \rtBraces{ i \rtc. \rtc{R_i} }_{i \in I}
                \mathrel{\raisebox{-2.0pt}{$\ltrans{ q \send p \Parens{j} }_p$}}
                \rtc{R_j}
            }
        \end{bussproof}
        \and
        \begin{bussproof}
            \bussAssume{
                \unfold( \rec X^{\mbb L} \rtc. \rtc{R} ) \ltrans{m}_p \rtc{R'}
            }
            \bussUn{
                \rec X^{\mbb L} \rtc. \rtc{R} \ltrans{m}_p \rtc{R'}
            }
        \end{bussproof}
    \end{mathparpagebreakable}
    Given $\vect m = m_1 , \ldots , m_k$, we write $\rtc{R} \ltrans{ \vect m }_p \rtc{R'}$ to denote $\rtc{R} \ltrans{ m_1 }_p \ldots \ltrans{m_k}_p \rtc{R'}$ (reflexive if $\vect m = \epsi$).
    We write $\vect m(q)$ to denote the subsequence of messages from $\vect m$ sent by $q$.
\end{definition}

\begin{definition}[Coherent Setup ($\bowtie$)]\label{d:coherentSetup}
    A global type $\gtc{G_0}$, a participant $p$, a monitor $\mc{M} \neq \rec X \mc. \mc{M'}$, a map $\RTs : \bm{P} \rightarrow \bm{R}$, and a sequence of messages $\vect n$ are \emph{in a coherent setup}, denoted $\coherent{\gtc{G_0}}{p}{\mc{M}}{\RTs}{\vect n}$, if and only if there exist $\gtc{G},\vect \ell$ such that $\gtc{G_0} \ltrans{ \vect \ell } \gtc{G}$, $\vect n$ exclusively contains messages with sender in $\dom(\RTs)$, and
    \begin{itemize}
        \item
            \textbf{(Initial state)}
            $\mc{M} \in \braces{ \gtToMon( \gtc{G} , p , \prt(\gtc{G_0}) \setminus \braces{p} ) , \checkmark }$ implies that
            \begin{itemize}
                \item for every $q \in \dom(\RTs)$ there exists $\mbb L_q$ such that $(\gtc{G} \wrt (p,q)^{\mbb L_q}) \ltrans{ \vect n(q) }_p \RTs(q)$, and
                \item if $\mc{M} = \checkmark$ then $\gtToMon( \gtc{G} , p , \prt(\gtc{G_0}) \setminus \braces{p} ) = \pend$;
            \end{itemize}

        \item
            \textbf{(Intermediate state)}
            Otherwise, there exist $\gtc{G'},j$ such that $\gtc{G} \ltrans{j} \gtc{G'}$, and either of the following holds:
            \begin{itemize}
                \item all of the following hold:
                    \begin{itemize}
                        \item $\mc{M} = p \send D (j) \mc. \gtToMon( \gtc{G'} , p , \prt(\gtc{G_0}) \setminus \braces{p} )$,
                        \item for every $q \in \dom(\RTs) \setminus D$ there exists $\mbb L_q$ such that $(\gtc{G'} \wrt (p,q)^{\mbb L_q}) \ltrans{ \vect n(q) }_p \RTs(q)$, and
                        \item for every $q \in D$ there exists $\mbb L_q$ such that $\RTs(q) = \gtc{G} \wrt (p,q)^{\mbb L_q}$;
                    \end{itemize}

                \item all of the following hold:
                    \begin{itemize}
                        \item $\mc{M} = p \recv r \mBraces{ j \mc. \gtToMon( \gtc{G'} , p , \prt(\gtc{G_0}) \setminus \braces{p} ) } \cup \mBraces{ i \mc. \perror }_{ i \in I \setminus \braces{j} }$,
                        \item for every $q \in \dom(\RTs) \setminus \braces{r}$ there exists $\mbb L_q$ such that $(\gtc{G'} \wrt (p,q)^{\mbb L_q}) \ltrans{ \vect n(q) }_p \RTs(q)$, and
                        \item there exists $\mbb L_r$ such that $(\gtc{G} \wrt (p,r)^{\mbb L_r}) \ltrans{ \vect n(r) }_p \RTs(r)$.
                    \end{itemize}
            \end{itemize}
    \end{itemize}
\end{definition}

\begin{lemma}\label{l:satisfactionNoError}
    Suppose
    \begin{itemize}
        \item $\mathcal{R} \satisfies{ \monImpl{ \bufImpl{ p }{ P_0 }{ \epsi } }{ \mc{M_0} }{ \epsi } }{ \RTs_0 }[ p ]$,
        \item where $\RTs_0 := \braces{ ( q , \gtc{G_0} \wrt (p,q)^{\braces{\epsi}} ) \mid q \in \prt(\gtc{G_0}) \setminus \braces{p} }$ for well-formed $\gtc{G_0}$.
    \end{itemize}
    Then
    \begin{itemize}
        \item for any $( \monImpl{ \bufImpl{ p }{ P }{ \vect m } }{ \mc{M} }{ \vect n } , \RTs , \Lbls ) \in \mathcal{R}$ such that $\coherent{ \gtc{G_0} }{ p }{ \mc{M} }{ \RTs }{ \vect n }$, and
        \item for any $\alpha,\net P',\Omega,\Omega'$ such that $\monImpl{ \bufImpl{ p }{ P }{ \vect m } }{ {\mc{M}} }{ \vect n } \bLtrans{\Omega}{\alpha}{\Omega'} \net P'$,
    \end{itemize}
    we have $\net P' \not\equiv \perror_{\braces{p}}$.
\end{lemma}

\begin{proof}
    Suppose, towards a contradiction, that $\net P' \equiv \perror_{\braces{p}}$.
    Then the transition is due to Transition~\ruleLabel{no-dep}, $\alpha = \tau$, $\Omega' = \Omega$, and the transition is derived from Transition~\ruleLabel{error-out}, \ruleLabel{error-in}, or~\ruleLabel{error-mon}.
    We show in each case separately that this leads to a contradiction.
    \begin{itemize}
        \item
            Transition~\ruleLabel{error-out}.
            Then $\mc{M}$ denotes an output, and so, by the definition of \nameref{d:coherentSetup}, there is a $q \in \dom(\RTs)$ such that $\unfold\big(\RTs(q)\big)$ is an output with a location that prefixes all other locations in $\RTs$.
            By \satref{i:sTau}{Tau}, $( \perror_{\braces{p}} , \RTs , \Lbls ) \in \mathcal{R}$.
            Then by the progress property of \nameref{d:satisfaction}, there exist $\alpha,\net P''$ such that $\perror_{\braces{p}} \ltrans{\alpha} \net P''$.
            However, the \nameref{d:ltsNetworks} does not define such a transition: a contradiction.

        \item
            Transition~\ruleLabel{error-in}.
            Then the monitor tries to read a message from $\vect n$, but the first relevant message is incorrect.
            By the definition of \nameref{d:coherentSetup}, all messages in $\vect n$ are in accordance with \satref{i:sInput}{Input} or~\satclause{i:sDependencyInput}{Dependency input}.
            But then the message cannot be incorrect: a contradiction.

        \item
            Transition~\ruleLabel{error-mon}.
            This means that $\mc{M} = \perror$.
            This can only be the case after two related dependency inputs of different labels.
            However, the messages read must be due to \satref{i:sDependencyInput}{Dependency input}, where the first one records the chosen label in $\Lbls$, and the second one uses that label.
            Hence, it is not possible that two different labels have been received: a contradiction.
    \end{itemize}
    Hence, $\net P' \not\equiv \perror_{\braces{p}}$.
\end{proof}

\begin{lemma}\label{l:satisfactionReverseBuffer}
    Suppose
    \begin{itemize}
        \item $\mathcal{R} \satisfies{ \monImpl{ \bufImpl{ p }{ P_0 }{ \epsi } }{ \mc{M_0} }{ \epsi } }{ \RTs_0 }[ p ]$,
        \item where $\RTs_0 := \braces{ ( q , \gtc{G_0} \wrt (p,q)^{\braces{\epsi}} ) \mid q \in \prt(\gtc{G_0}) \setminus \braces{p} }$
        \item for well-formed $\gtc{G_0}$ with $p \in \prt(\gtc{G_0})$.
    \end{itemize}
    Then
    \begin{itemize}
        \item for any $( \monImpl{ \bufImpl{ p }{ P }{ \vect m } }{ \mc{M} }{ \vect n } , \RTs , \Lbls ) \in \mathcal{R}$
        \item with non-empty $\vect n$
        \item such that $\coherent{ \gtc{G_0} }{ p }{ \mc{M} }{ \RTs }{ \vect n }$,
    \end{itemize}
    there exist
    \begin{itemize}
        \item $\RTs',\Lbls'$ such that $( \monImpl{ \bufImpl{ p }{ P }{ \vect m } }{ \mc{M} }{ \epsi } , \RTs' , \Lbls' ) \in \mathcal{R}$, and
        \item $(q,\rtc{R}) \in \dom(\RTs')$ such that
            \begin{itemize}
                \item $\unfold(\rtc{R}) = q \send p^{\mbb L} \rtBraces{ \msg i<T_i> \rtc. \rtc{R_i} }_{i \in I}$
                \item or $\unfold(\rtc{R}) = (q \lozenge r) \send p^{\mbb L} \rtBraces{ i \rtc. \rtc{R_i} }_{i \in I}$.
            \end{itemize}
    \end{itemize}
\end{lemma}

\begin{proof}
    Take any $( \monImpl{ \bufImpl{ p }{ P }{ \vect m } }{ \mc{M} }{ \vect n } , \RTs , \Lbls ) \in \mathcal{R}$ with non-empty $\vect n$ such that $\coherent{ \gtc{G_0} }{ p }{ \mc{M} }{ \RTs }{ \vect n }$.
    By the minimality of $\mathcal{R}$, each message in $\vect n$ traced back to applications of \satref{i:sInput}{Input} or~\satclause{i:sDependencyInput}{Dependency input}.
    By the definition of \nameref{d:coherentSetup}, each relative type in $\RTs$ relates $\gtc{G_0}$ with the messages in $\vect n$ through the \nameref{d:ltsRelativeTypes}.
    Hence, at $( \monImpl{ \bufImpl{ p }{ P }{ \vect m } }{ \mc{M} }{ \epsi } , \RTs' , \Lbls' ) \in \mathcal{R}$, since $\vect n$ is non-empty, the unfolding of at least one relative type in $\RTs'$ denotes an input by $p$.
\end{proof}

\begin{lemma}\label{l:satisfactionOutput}
    \leavevmode
    \begin{itemize}
        \item Suppose
            \begin{itemize}
                \item $\mathcal{R} \satisfies{ \monImpl{ \bufImpl{ p }{ P_0 }{ \epsi } }{ \mc{M_0} }{ \epsi } }{ \RTs_0 }[ p ]$,
                \item where $\RTs_0 := \braces{ ( q , \gtc{G_0} \wrt (p,q)^{\braces{\epsi}} ) \mid q \in \prt(\gtc{G_0}) \setminus \braces{p} }$
                \item for well-formed $\gtc{G_0}$ with $p \in \prt(\gtc{G_0})$.
            \end{itemize}
        \item Take any $( \monImpl{ \bufImpl{ p }{ P }{ \vect m } }{ \mc{M} }{ \vect n } , \RTs , \Lbls ) \in \mathcal{R}$ such that $\coherent{ \gtc{G_0} }{ p }{ \mc{M} }{ \RTs }{ \vect n }$.
        \item Suppose $\bufImpl{ p }{ P }{ \vect m } \ltrans{ p \send q (\msg j<T_j>) } \bufImpl{ p }{ P' }{ \vect m }$.
    \end{itemize}
    Then
    \begin{itemize}
        \item $\unfold(\mc{M}) = \mc{M'}$
        \item or $\unfold(\mc{M}) = p \send D (\ell) \mc. \mc{M'}$
    \end{itemize}
    where $\mc{M'} = p \send q \mBraces{ \msg i<T_i> \mc. \mc{M'_i} }_{i \in I}$ s.t.\ $j \in I$.
\end{lemma}

\begin{proof}
    Suppose, toward a contradiction, that $\unfold(\mc{M}) \neq \mc{M'}$ and $\unfold(\mc{M}) \neq p \send D (\ell) \mc. \mc{M'}$.
    Then, by Transition~\ruleLabel{error-out}, $\monImpl{ \bufImpl{ p }{ P }{ \vect m } }{ \mc{M} }{ \vect n } \ltrans{\tau} \perror_{\braces{p}}$.
    By \satref{i:sTau}{Tau}, $( \perror_{\braces{p}} , \RTs , \Lbls ) \in \mathcal{R}$.
    If $\mc{M}$ starts with a sequence of recursive definitions, we unfold them; by Transition~\ruleLabel{mon-rec}, the behavior is unchanged; let us assume, w.l.o.g., that $\mc{M}$ does not start with a recursive definition.
    From each possible shape of $\mc{M}$, we derive a contradiction.
    \begin{itemize}
        \item
            $\mc{M} = p \send r \mBraces{ \msg k<T_k> \mc. \mc{M_k} }_{k \in K}$ for $r \neq q$.
            By definition of \nameref{d:coherentSetup}, $\gtc{G_0} \ltrans{\vect \ell} \gtc{G}$, $\mc{M} = \gtToMon( \gtc{G} , p , \prt(\gtc{G_0}) \setminus \braces{p} )$, and there exists $\mbb L$ such that $(\gtc{G} \wrt (p,r)^{\mbb L}) \ltrans{ \vect n(r) }_p \RTs(r)$.
            Then $\unfold( \gtc{G} \wrt (p,r)^{\mbb L} ) = p \send r \rtBraces{ \msg k<T_k> \rtc. \rtc{R_k} }_{k \in K}$.
            Since the \nameref{d:ltsRelativeTypes} only allows input transitions, then $\RTs(r) = \gtc{G} \wrt (p,r)^{\mbb L}$.
            Clearly, $\mbb L$ prefixes all other locations in $\RTs$.
            Hence, by the progress property of \nameref{d:satisfaction}, there exists $\net P'$ such that $\perror_{\braces{p}} \ltrans{ \tau } \net P'$.
            However, the \nameref{d:ltsNetworks} does not define such a transition: a contradiction.

        \item
            $\mc{M} = p \send q \mBraces{ \msg i<T_i> \mc. \mc{M_i} }_{i \in I}$ for $j \notin I$.
            Analogous to the case above.

        \item
            $\mc{M} = p \recv r \mBraces{ \msg k<T_k> \mc. \mc{M_k} }_{k \in K}$.
            By definition of \nameref{d:coherentSetup}, $\gtc{G_0} \ltrans{\vect \ell} \gtc{G}$, $\mc{M} = \gtToMon( \gtc{G} , p , \prt(\gtc{G_0}) \setminus \braces{p} )$, and there exists $\mbb L$ such that $(\gtc{G} \wrt (p,r)^{\mbb L}) \ltrans{ \vect n(r) }_p \RTs(r)$.
            Whether $\vect n$ is empty or not, there exist $\RTs',\Lbls'$ such that $( \monImpl{ \bufImpl{ p }{ P }{ \vect m } }{ \mc{M} }{ \epsi } , \RTs' , \Lbls' ) \in \mathcal{R}$: if empty, this holds vacuously with $\RTs' := \RTs,\Lbls' := \Lbls$; if non-empty, it follows from \cref{l:satisfactionReverseBuffer}.
            Then $\unfold\big(\RTs'(r)\big) = \unfold( \gtc{G} \wrt (p,r)^{\mbb L} ) = r \send p \rtBraces{ \msg k<T_k> \rtc. \rtc{R_k} }_{k \in K}$.
            Again, by Transition~\ruleLabel{error-out}, $\monImpl{ \bufImpl{ p }{ P }{ \vect m } }{ \mc{M} }{ \epsi } \ltrans{\tau} \perror_{\braces{p}}$, and, by \satref{i:sTau}{Tau}, $( \perror_{\braces{p}} , \RTs' , \Lbls' ) \in \mathcal{R}$.
            Then, by \satref{i:sInput}{Input}, $\perror_{\braces{p}} = \monImpl{ \bufImpl{ p }{ P }{ \vect m } }{ \mc{M'} }{ \vect n }$: a contradiction.

        \item
            $\mc{M} = p \recv r \mBraces{ k \mc. \mc{M_k} }_{k \in K}$.
            Analogous to the case above.

        \item
            $\mc{M} = \pend$.
            By definition of \nameref{d:coherentSetup}, $\gtc{G_0} \ltrans{\vect \ell} \gtc{G}$, $\mc{M} = \gtToMon( \gtc{G} , p , \prt(\gtc{G_0}) \setminus \braces{p} )$, and for every $q \in \dom(\RTs)$ there exists $\mbb L_q$ such that $(\gtc{G} \wrt (p,q)^{\mbb L_q}) \ltrans{ \vect n(q) }_p \RTs(q)$.
            Since the \nameref{d:ltsRelativeTypes} only allows input transitions, for every $q \in \dom(\RTs)$, $\unfold\big(\RTs(q)\big) = \unfold( \gtc{G} \wrt (p,q)^{\mbb L_q} ) = \pend$.
            Then, by \satref{i:sEnd}{End}, there exists $\net P'$ such that $\perror_{\braces{p}} \trans* \ltrans{\pend} \net P'$.
            However, the \nameref{d:ltsNetworks} does not define such transitions: a contradiction.

        \item
            $\mc{M} = \perror$.
            By definition of \nameref{d:coherentSetup}, $\mc{M} \neq \perror$: a contradiction.

        \item
            $\mc{M} = \checkmark$.
            It cannot be the case that $\mc{M_0} = \checkmark$, and, by the \nameref{d:ltsNetworks}, $\mc{M} = \checkmark$ can only be reached through a transition labeled $\pend$.
            Hence, $\monImpl{ \bufImpl{ p }{ P }{ \vect m } }{ \checkmark }{ \vect n }$ must have been reached through \satref{i:sEnd}{End}.
            Then $\monImpl{ \bufImpl{ p }{ P }{ \vect m } }{ \checkmark }{ \vect n } \ntrans$: a contradiction.
    \end{itemize}
    Hence, the thesis must indeed hold.
\end{proof}

\begin{lemma}\label{l:satisfactionEnd}
    \leavevmode
    \begin{itemize}
        \item Suppose
            \begin{itemize}
                \item $\mathcal{R} \satisfies{ \monImpl{ \bufImpl{ p }{ P_0 }{ \epsi } }{ \mc{M_0} }{ \epsi } }{ \RTs_0 }[ p ]$,
                \item where $\RTs_0 := \braces{ ( q , \gtc{G_0} \wrt (p,q)^{\braces{\epsi}} ) \mid q \in \prt(\gtc{G_0}) \setminus \braces{p} }$
                \item for well-formed $\gtc{G_0}$ with $p \in \prt(\gtc{G_0})$.
            \end{itemize}
        \item Take any $( \monImpl{ \bufImpl{ p }{ P }{ \vect m } }{ \mc{M} }{ \vect n } , \RTs , \Lbls ) \in \mathcal{R}$ such that $\coherent{ \gtc{G_0} }{ p }{ \mc{M} }{ \RTs }{ \vect n }$.
        \item Suppose $\bufImpl{ p }{ P }{ \vect m } \ltrans{ \pend } \bufImpl{ p }{ P' }{ \vect m }$.
    \end{itemize}
    Then $\vect n = \epsi$, and
    \begin{itemize}
        \item $\unfold(\mc{M}) = \pend$
        \item or $\unfold(\mc{M}) = p \send D (\ell) \mc. \pend$.
    \end{itemize}
\end{lemma}

\begin{proof}
    Suppose toward a contradiction that $\vect n \neq \epsi$ or that $\unfold(\mc{M}) \notin \braces{ \pend , p \send D (\ell) \mc. \pend }$.
    Then, by Transition~\ruleLabel{error-end}, $\monImpl{ \bufImpl{ p }{ P }{ \vect m } }{ \mc{M} }{ \vect n } \ltrans{\tau} \perror_{\braces{p}}$.
    The contradiction follows similar to the reasoning in the proof of \cref{l:satisfactionOutput}.
    We only discuss the additional case where $\vect n \neq \epsi$ and $\mc{M} = p \send D (\ell) \mc. \mc{M'}$.
    By induction on the size of $D$, $\monImpl{ \bufImpl{ p }{ P }{ \vect m } }{ \mc{M} }{ \vect n } \trans* \monImpl{ \bufImpl{ p }{ P }{ \vect m } }{ \mc{M'} }{ \vect n }$.
    By the definition of \nameref{d:coherentSetup}, $\mc{M}$ is synthesized from a global type, so $\mc{M'}$ is not a dependency output (there are never two consecutive dependency outputs).
    Hence, the other cases for $\mc{M'}$ apply and the search for a contradiction continues as usual.
\end{proof}

\begin{lemma}\label{l:finiteReceive}
    Suppose
    \begin{itemize}
        \item $\coherent{ \gtc{G_0} }{ p }{ \mc{M} }{ \RTs }{ \vect n , u }$
        \item with $\gtc{G_0} \ltrans{ \vect \ell } \gtc{G} \ltrans{ k } \gtc{G'}$
        \item and $u = q \send p (\msg j<T_j>)$.
    \end{itemize}
    Then
    \begin{itemize}
        \item $\gtc{G'} = q \send p \gtBraces{ \msg i<T_i> \gtc. \gtc{G'_i} }_{i \in I}$ with $j \in I$
        \item or, for every $k',\gtc{G''}$ such that $\gtc{G'} \ltrans{ k' } \gtc{G''}$, there exists (finite) $\vect \ell'$ such that $\gtc{G''} \ltrans{ \vect \ell' } q \send p \gtBraces{ \msg i<T_i> \gtc. \gtc{G''_i} }_{i \in I}$ with $j \in I$.
    \end{itemize}
\end{lemma}

\begin{proof}
    By the definition of \nameref{d:coherentSetup}, there are $\mbb L,\mbb L'$ such that $\unfold( \gtc{G'} \wrt (p,q)^{\mbb L} ) = q \send p^{\mbb L'} \rtBraces{ \msg i<T_i> \rtc. \rtc{R_i} }_{i \in I}$ with $j \in I$.
    By definition, this relative type is generated in finitely many steps by \cref{alg:relativeProjectionWithLocs} (\cref{li:rpIndepLab}).
    We apply induction on the number $x$ of steps this took.
    In the base case, the thesis holds trivially.

    In the inductive case, the relative projections of all branches of $\gtc{G'}$ onto $(p,q)$ are the same.
    Take any $k',\gtc{G''}$ such that $\gtc{G'} \ltrans{k'} \gtc{G''}$.
    Then $\gtc{G''} \wrt (p,q)^{\mbb L+k'}$ is generated in less than $x$ steps through \cref{alg:relativeProjectionWithLocs} (\cref{li:rpIndepLab}).
    Hence, the thesis follows from the IH.
\end{proof}

\begin{lemma}\label{l:monitorNoOutput}
    \leavevmode
    \begin{itemize}
        \item Suppose
            \begin{itemize}
                \item $\mathcal{R} \satisfies{ \monImpl{ \bufImpl{ p }{ P_0 }{ \epsi } }{ \mc{M_0} }{ \epsi } }{ \RTs_0 }[ p ]$,
                \item where $\coherent{ \gtc{G_0} }{ p }{ \mc{M_0} }{ \RTs_0 }{ \epsi }$
                \item for well-formed $\gtc{G_0}$ with $p \in \prt(\gtc{G_0})$.
            \end{itemize}
        \item Take any $\RTs,\Lbls,\vect m,\vect n,\mc{M},P$ such that
            \begin{itemize}
                \item $( \monImpl{ \bufImpl{ p }{ P }{ \vect m } }{ \mc{M} }{ \vect n } , \RTs , \Lbls ) \in \mathcal{R}$,
                \item $\coherent{ \gtc{G_0} }{ p }{ \mc{M} }{ \RTs }{ \vect n }$,
                \item and $\vect n$ is non-empty.
            \end{itemize}
    \end{itemize}
    If
    \begin{itemize}
        \item $P \ltrans{ p \recv q (\msg j<T_j>) } P'$
        \item and $\unfold(\mc{M}) \neq p \recv q \mBraces{ \msg i<T_i> \mc. \mc{M_i} }_{i \in I}$ for $I \supseteq \braces{j}$,
    \end{itemize}
    then $\unfold(\mc{M}) \in \braces{ p \recv r \mBraces{ \msg k<T_k> \mc. \mc{M_k} }_{k \in K} , p \recv r \mBraces{ k \mc. \mc{M_k} }_{k \in K} , p \send D (\ell) \mc. \mc{M'} \mid r \neq q }$.
\end{lemma}

\begin{proof}
    Assume, toward a contradiction, that $\unfold(\mc{M}) \notin \braces{ p \recv r \mBraces{ \msg k<T_k> \mc. \mc{M_k} }_{k \in K} , p \recv r \mBraces{ k \mc. \mc{M_k} }_{k \in K} , p \send D (\ell) \mc. \mc{M'} \mid r \neq q }$.
    We discuss each possible shape of $\unfold(\mc{M})$ separetely (w.l.o.g., assume $\mc{M}$ does not start with a recursive definition, i.e., $\unfold(\mc{M}) = \mc{M}$).
    \begin{itemize}
        \item
            $\mc{M} = p \send r \mBraces{ \msg k<T_k> \mc. \mc{M_k} }_{k \in K}$.
            By definition of \nameref{d:coherentSetup}, $\unfold\big(\RTs(r)\big) = p \send r^{\mbb L} \rtBraces{ \msg k<T_k> \rtc. \rtc{R_k} }_{k \in K}$ for some $\mbb L$.
            It must then be that $\mbb L$ is the earliest location in all $\RTs$.
            Then, by the progress property of \nameref{d:satisfaction}, the monitored blackbox must eventually do an output transition.
            However, by \bref{i:bInputOutput}{Input/Output}, since $P$ does an input transition, it cannot do an output transition: a contradiction.

        \item
            $\mc{M} = p \recv q \mBraces{ \msg i<T_i> \mc. \mc{M_i} }_{i \in I}$ for $j \notin I$.
            Analogous to the similar case in the proof of \cref{l:satisfactionOutput}.

        \item
            $\mc{M} = \pend$.
            Analogous to the similar case in the proof of \cref{l:satisfactionOutput}.

        \item
            $\mc{M} = \perror$.
            Analogous to the similar case in the proof of \cref{l:satisfactionOutput}.

        \item
            $\mc{M} = \checkmark$.
            It cannot be the case the $\mc{M_0} = \checkmark$, and, by the \nameref{d:ltsNetworks}, $\mc{M} = \checkmark$ can only be reached through a transition labeled $\pend$.
            Hence, the current monitored blackbox must have been reached through \satref{i:sEnd}{End} from $\monImpl{ \bufImpl{ p }{ P }{ \vect m } }{ \pend }{ \vect n}$.
            An $\pend$-transition from this state requires that $\vect n$ is empty, which it is not: a contradiction.
    \end{itemize}
    Hence, the thesis must hold indeed.
\end{proof}

\begin{lemma}\label{l:monitorNotCheck}
    \leavevmode
    \begin{itemize}
        \item Suppose
            \begin{itemize}
                \item $\mathcal{R} \satisfies{ \monImpl{ \bufImpl{ p }{ P_0 }{ \epsi } }{ \mc{M_0} }{ \epsi } }{ \RTs_0 }[ p ]$,
                \item where $\coherent{ \gtc{G_0} }{ p }{ \mc{M_0} }{ \RTs_0 }{ \epsi }$
                \item for well-formed $\gtc{G_0}$ with $p \in \prt(\gtc{G_0})$.
            \end{itemize}
        \item Take any $\RTs,\Lbls,\vect m,\vect n,\mc{M},P$ such that
            \begin{itemize}
                \item $( \monImpl{ \bufImpl{ p }{ P }{ \vect m } }{ \mc{M} }{ \vect n } , \RTs , \Lbls ) \in \mathcal{R}$, $\coherent{ \gtc{G_0} }{ p }{ M }{ \RTs }{ \vect n }$,
                \item and $\vect n$ is non-empty.
            \end{itemize}
    \end{itemize}
    If $P \ltrans{ \tau } P'$, then $\mc{M} \neq \checkmark$.
\end{lemma}

\begin{proof}
    Suppose, toward a contradiction, that $\mc{M} = \checkmark$.
    We have $\mc{M_0} \neq \checkmark$, so this state must have been reached through an $\pend$-transition between $P_0$ and $P$.
    However, by \bref{i:bEnd}{End}, there can be no transitions after an $\pend$-transition: a contradiction.
\end{proof}

\begin{lemma}\label{l:labelOracleDepOut}
    Suppose given $\RTs : \bm{P} \rightarrow \bm{R}$ and $\Lbls : \Pow(\vect {\bm{L}}) \rightarrow \bm{L}$.

    If there is $( q , \rtc{R} ) \in \RTs$ such that $\unfold(\rtc{R}) = (p \lozenge r) \send q^{\mbb L} \rtBraces{ i \rtc. \rtc{R_i} }_{i \in I}$, then, for every $j \in I$,
    \[
        \LO( p , \RTs , \Lbls ) = \LO( p , \RTs \update{q \mapsto \rtc{R_j}} , \Lbls ).
    \]
\end{lemma}

\begin{proof}
    Trivally, by the definition of \nameref{d:labelOracle}, the right-hand-side is a subset of the left-hand-side.
    For the other direction, the update of $\RTs$ might add additional sequences.
    However, since the update to $\RTs$ does not affect any other relative types and $\Lbls$ is not updated, no sequences are removed.
    Hence, the left-hand-side is a subset of the right-hand-side.
\end{proof}

\begin{lemma}\label{l:LOAlpha}
    Suppose given $\RTs : \bm{P} \rightarrow \bm{R}$, $\Lbls : \Pow(\vect {\bm{L}}) \rightarrow \bm{L}$, $\Omega = \LO(p,\RTs,\Lbls)$, and $\Omega' = \Omega(\alpha)$.

    Then all of the following hold:
    \begin{itemize}

        \item
            If $\alpha = p \send q (\msg j<T_j>)$, then $\RTs(q) \unfoldeq p \send q^{\mbb L} \rtBraces{ \msg i<T_i> \rtc. \rtc{R_i} }_{i \in I}$ with $j \in I$, and $\Omega' = \LO(p,\RTs\update{q\mapsto \rtc{R_j}},\Lbls\update{\mbb L\mapsto j})$.

        \item
            If $\alpha = p \recv q (\msg j<T_j>)$, then $\RTs(q) \unfoldeq q \send p^{\mbb L} \rtBraces{ \msg i<T_i> \rtc. \rtc{R_i} }_{i \in I}$ with $j \in I$, and $\Omega' = \LO(p,\RTs\update{q\mapsto \rtc{R_j}},\Lbls\update{\mbb L\mapsto j})$.

        \item
            If $\alpha = p \recv q \Parens{j}$, then $\RTs(q) \unfoldeq (q \lozenge r) \send p^{\mbb L} \rtBraces{ i \rtc. \rtc{R_i} }_{i \in I}$ with $j \in I$.
            If $\not\exists \mbb L' \in \dom(\Lbls).~\mbb L' \overlap \mbb L$, then $\Omega' = \LO(p,\RTs\update{q\mapsto \rtc{R_j}},\Lbls\update{\mbb L\mapsto j})$.
            If $\exists (\mbb L',j) \in \Lbls.~ \mbb L' \overlap \mbb L$, then $\Omega' = \LO(p,\RTs\update{q\mapsto \rtc{R_j}},\Lbls)$.

        \item
            If $\alpha = \pend$, then $\forall (q,\rtc{R}) \in \RTs.~\rtc{R}\unfoldeq\pend$, and $\Omega' = \LO(p,\emptyset,\emptyset)$.

        \item
            If $\alpha = \tau$, then $\Omega' = \Omega = \LO(p,\RTs,\Lbls)$.
    \end{itemize}
\end{lemma}

\begin{proof}
    Keeping in mind \cref{l:labelOracleDepOut}, $\Omega$ is generated with each possible $\alpha$ appearing exactly once under specific conditions and with unique continuation.
\end{proof}

\begin{theorem}[Transparency]\label{t:transparencyBisim}
    Suppose given
    \begin{itemize}
        \item a well-typed global type $\gtc{G_0}$,
        \item a participant $p \in \prt(\gtc{G_0})$, and
        \item a blackbox $P_0$.
    \end{itemize}
    Let
    \begin{itemize}
        \item $\RTs_0 := \braces{ ( q , \gtc{G_0} \wrt (p,q)^{\braces{\epsi}} ) \mid q \in \prt(\gtc{G_0}) \setminus \braces{p} }$, and
        \item $\mc{M_0} := \gtToMon( \gtc{G_0} , p , \prt(\gtc{G_0}) \setminus \braces{p} )$.
    \end{itemize}
    Suppose $\satisfies{ \monImpl{ \bufImpl{ p }{ P_0 }{ \epsi } }{ \mc{M_0} }{ \epsi } }{ \RTs_0 }[ p ]$ minimally (\cref{d:minimalSat}).

    Let $\Omega_0 := \LO(p,\RTs_0,\emptyset)$.
    Then $\monImpl{ \bufImpl{ p }{ P_0 }{ \epsi } }{ \mc{M_0} }{ \epsi } \weakBisim{}{\Omega_0} \bufImpl{ p }{ P_0 }{ \epsi }$.
\end{theorem}

\begin{proof}
    By \nameref{d:satisfaction}, there exists a minimal satisfaction $\mathcal{R}$ at $p$ such that $( \monImpl{ \bufImpl{ p }{ P_0 }{ \epsi } }{ \mc{M_0} }{ \epsi } , \RTs_0 , \emptyset ) \in \mathcal{R}$.
    Let
    \[
        \mathcal{B} := \{ \begin{array}[t]{@{}l@{}}
            ( \monImpl{ \bufImpl{ p }{ P }{ \vect m } }{ \mc{M} }{ \vect n } , \Omega , \bufImpl{ p }{ P }{ \vect n , \vect m } )
            \\
            {} \mid \forall P , M , \vect m , \vect n.~ \exists \RTs , \Lbls.~ \big( \begin{array}[t]{@{}l@{}}
                ( \monImpl{ \bufImpl{ p }{ P }{ \vect m } }{ \mc{M} }{ \vect n } , \RTs , \Lbls ) \in \mathcal{R}
                \\ {} \wedge
                \coherent{G_0}{p}{\mc{M}}{\RTs}{\vect n}
                \\ {} \wedge
                \Omega = \LO( p , \RTs , \Lbls )
        \big) \}. \end{array} \end{array}
    \]
    Clearly, $( \monImpl{ \bufImpl{ p }{ P_0 }{ \epsi } }{ \mc{M_0} }{ \epsi } , \Omega_0 , \bufImpl{ p }{ P_0 }{ \epsi } ) \in \mathcal{B}$, with $P = P_0 , \mc{M} = \mc{M_0} , \vect m = \vect n = \epsi , \RTs = \RTs_0 , \Lbls = \emptyset$, and clearly $\coherent{G_0}{p}{\mc{M_0}}{\RTs_0}{\epsi}$.

    It remains to show that $\mathcal{B}$ is a weak bisimulation.
    Take any $( \net P , \Omega , \net Q ) \in \mathcal{B}$: there are $P , \mc{M} , \vect n , \vect m$ such that $\net P = \monImpl{ \bufImpl{ p }{ P }{ \vect m } }{ \mc{M} }{ \vect n }$, $\net Q = \bufImpl{ p }{ P }{ \vect n , \vect m }$, and there are $\RTs , \Lbls$ such that $( \net P , \RTs , \Lbls ) \in \mathcal{R}$, $\coherent{\gtc{G_0}}{p}{\mc{M}}{\RTs}{\vect n}$, and $\Omega = \LO( p , \RTs , \Lbls )$.
    We show that the two conditions of \cref{d:weakBisim} hold.

    \begin{enumerate}
        \item
            Take any $\net P' , \alpha , \Omega_1$ such that $\net P \bLtrans{\Omega}{ \alpha }{\Omega_1} \net P'$.
            The analysis depends on the rule from \cref{d:enhancedLTS} used to derive the transition (Transition~\ruleLabel{buf}, \ruleLabel{dep}, or~\ruleLabel{no-dep}).
            We never need to read extra messages, so in each case we show the thesis for $\vect b := \epsi$ and $\net P'' := \net P'$.
            That is, in each case we show that there exists $\net Q'$ such $\net Q \bLtrans*{\Omega}{\alpha}{\Omega_1} \net Q'$ and $( \net P' , \Omega_1 , \net Q' ) \in \mathcal{B}$.
            \begin{itemize}
                \item
                    Transition~\ruleLabel{buf-mon}.
                    Then $\alpha \in \braces{ p \recv q (x) , p \recv q \Parens{x} }$, $\net P' = \monImpl{ \bufImpl{ p }{ P }{ \vect m } }{ \mc{M} }{ u , \vect n }$ for $u \in \braces{ q \send p (x) , q \send p \Parens{x} }$, and $\Omega_1 = \Omega(\alpha)$.

                    By \cref{l:LOAlpha}, we have $( q , \rtc{R^q} ) \in \RTs$ where $\unfold(\rtc{R^q})$ is a (dependency) message from $q$ to $p$.
                    For the sake of simplicity, assume w.l.o.g.\ that the message is no dependency.
                    Then $\unfold(\rtc{R^q}) = q \send p^{\mbb L} \rtBraces{ \msg i<T_i> \rtc. \rtc{R^q_i} }_{i \in I}$, $\alpha = p \recv q(x)$, $x = \msg j<T_j>$ with $j \in I$, and $u = q \send p (x)$.
                    Let $\RTs' := \RTs \update{q \mapsto \rtc{R^q_j}}$ and $\Lbls' := \Lbls \update{\mbb L \mapsto j}$.
                    By \satref{i:sInput}{Input}, $( \net P' , \RTs' , \Lbls' ) \in \mathcal{R}$.
                    By \cref{l:LOAlpha}, $\Omega_1 = \LO( p , \RTs' , \Lbls' )$.

                    Since $\mc{M}$ is not updated, and the addition of $u$ to the buffer only adds an input transition to the \nameref{d:coherentSetup} reflected by updated entry for $q$ in $\RTs'$, we have $\coherent{ G_0 }{ p }{ \mc{M} }{ \RTs' }{ u , \vect n }$.
                    Let $\net Q' := \bufImpl{ p }{ P }{ u , \vect n , \vect m }$.
                    By Transition~\ruleLabel{buf}, $\net Q \bLtrans{\Omega}{ \alpha }{\Omega_1} \net Q'$ so $\net Q \bLtrans*{\Omega}{\alpha}{\Omega_1} \net Q'$.
                    Finally, by definition, $( \net P' , \Omega_1 , \net Q' ) \in \mathcal{B}$.

                \item
                    Transition~\ruleLabel{buf-unmon}.
                    This rule does not apply to monitored blackboxes.

                \item
                    Transition~\ruleLabel{dep}.
                    Then $\alpha = \tau$, $\net P \ltrans{ \alpha' } \net P'$ with $\alpha' = s \send r \Parens{j}$ for some $s,r,j$, and $\Omega_1 = \Omega$.

                    This can only have been derived from Transition~\ruleLabel{mon-out-dep}.
                    Then $\mc{M} = p \send (D \cup \braces{r}) (j) \mc. \mc{M'}$, $s = p$, and $\net P \ltrans{ \alpha' } \monImpl{ \bufImpl{ p }{ P }{ \vect m } }{ p \send D (j) \mc. \mc{M'} }{ \vect n } = \net P'$.
                    By \satref{i:sDependencyOutput}{Dependency output}, we have $( r , \rtc{R^r} ) \in \RTs$ with $\unfold(\rtc{R^r}) = (p \lozenge q) \send r^{\mbb L} \rtBraces{ i \rtc. \rtc{R^r_i} }_{i \in I}$ and $j \in I$, and so $( \net P' , \RTs' , \Lbls ) \in \mathcal{R}$ with $\RTs' := \RTs \update{r \mapsto \rtc{R^r_j}}$.

                    In $\RTs'$, only the entry for $r$ has been updated, so $\coherent{ \gtc{G_0} }{ p }{ p \send D (j) \mc. \mc{M'} }{ \RTs' }{ \vect n }$.
                    By \cref{l:labelOracleDepOut}, $\Omega = \LO( p , \RTs' , \Lbls )$.
                    We have $\net Q \bLtrans*{\Omega}{\tau}{\Omega} \net Q$.
                    Finally, by definition, $( \net P' , \Omega , \net Q ) \in \mathcal{B}$.

                \item
                    Transition~\ruleLabel{no-dep}.
                    Then $\alpha \neq s \send r (\ell)$ for any $s,r,\ell$, $\net P \ltrans{ \alpha } \net P'$, and $\Omega_1 = \Omega(\alpha)$.
                    The analysis depends on the derivation of the transition.
                    Some rules are impossible: there is no parallel composition in $\net P$, no rules to derive transitions for buffered blackboxes are possible, and, by \cref{l:satisfactionNoError}, no transitions resulting in an error signal are possible.

                    Some rules require to first unfold recursion in $\mc{M}$, derived by a number of consecutive applications of Transition~\ruleLabel{mon-rec}.
                    We apply induction on this number.
                    The inductive case is trivial by the IH.

                    In the base case, where there are no applications of Transition~\ruleLabel{mon-rec}, we consider each possible rule (Transition~\ruleLabel{mon-out}, \ruleLabel{mon-in}, \ruleLabel{mon-in-dep}, \ruleLabel{mon-tau}, \ruleLabel{mon-out-dep-empty}, and~\ruleLabel{mon-end}).

                    \begin{itemize}
                        \item
                            Transition~\ruleLabel{mon-out}.
                            Then $\alpha = p \send q (\msg j<T_j>)$, $\bufImpl{ p }{ P }{ \vect m } \ltrans{ \alpha } \bufImpl{ p }{ P' }{ \vect m }$, $\mc{M} = p \send q \mBraces{ \msg i<T_i> \mc. \mc{M_i} }_{i \in I}$, $j \in I$, and $\net P' = \monImpl{ \bufImpl{ p }{ P' }{ \vect m } }{ \mc{M_j} }{ \vect n }$.

                            The transition of the buffered blackbox must be due to Transition~\ruleLabel{buf-out}: $P \ltrans{ \alpha } P'$.
                            By \satref{i:sOutput}{Output}, we have $( q , \rtc{R^q} ) \in \RTs$ with $\unfold(\rtc{R^q}) = p \send q^{\mbb L} \rtBraces{ \msg i<T_i> \rtc. \rtc{R^q_i} }_{i \in I}$ and $j \in I$, and so $( \net P' , \RTs' , \Lbls' ) \in \mathcal{R}$ with $\RTs' := \RTs \update{q \mapsto \rtc{R^q_j}}$ and $\Lbls' := \Lbls \update{\mbb L \mapsto j}$.

                            In $\RTs'$, only the entry for $q$ has been updated, so $\coherent{ \gtc{G_0} }{ p }{ M_j }{ \RTs' }{ \vect n }$.
                            By \cref{l:LOAlpha}, $\Omega_1 = \LO( p , \RTs' , \Lbls' )$.
                            Let $\net Q' := \bufImpl{ p }{ P' }{ \vect n , \vect m }$.
                            By Transition~\ruleLabel{buf-out}, $\net Q \ltrans{ \alpha } \net Q'$.
                            Then, by Transition~\ruleLabel{no-dep}, $\net Q \bLtrans{\Omega}{ \alpha }{\Omega_1} \net Q'$ so $\net Q \bLtrans*{\Omega}{ \alpha}{\Omega_1} \net Q'$.
                            Finally, by definition, $( \net P' , \Omega_1 , \net Q' ) \in \mathcal{B}$.

                        \item
                            Transition~\ruleLabel{mon-in}.
                            Then $\alpha = \tau$, $\mc{M} = p \recv q \mBraces{ \msg i<T_i> \mc. \mc{M_i} }_{i \in I}$, $\vect n = \vect n' , u$ where $u = q \send p (\msg j<T_j>)$, $j \in I$, $\net P' = \monImpl{ \bufImpl{ p }{ P }{ u , \vect m } }{ \mc{M_j} }{ \vect n' }$.

                            By \satref{i:sTau}{Tau}, $( \net P' , \RTs , \Lbls ) \in \mathcal{R}$.
                            Since $\mc{M}$ has moved past the input from $q$ ($\RTs(q)$ was already past this point), and it has been removed from $\vect n$, we have $\coherent{ \gtc{G_0} }{ p }{ \mc{M_j} }{ \RTs }{ \vect n' }$.
                            By \cref{l:LOAlpha}, $\Omega_1 = \LO( p , \RTs , \Lbls )$.
                            We have $\net Q \bLtrans*{\Omega}{\tau}{\Omega_1} \net Q$.
                            Finally, by definition, $( \net P' , \Omega_1 , \net Q ) \in \mathcal{B}$.

                        \item
                            Transition~\ruleLabel{mon-in-dep}.
                            Analogous to Transition~\ruleLabel{mon-in}.

                        \item
                            Transition~\ruleLabel{mon-tau}.
                            Then $\alpha = \tau$, $\net P' = \monImpl{ \bufImpl{ p }{ P' }{ \vect m' } }{ \mc{M} }{ \vect n }$, and $\bufImpl{ p }{ P }{ \vect m } \ltrans{ \tau } \bufImpl{ p }{ P' }{ \vect m' }$.

                            By \satref{i:sTau}{Tau}, $( \net P' , \RTs , \Lbls ) \in \mathcal{R}$.
                            By \cref{l:LOAlpha}, $\Omega_1 = \LO( p , \RTs , \Lbls ) = \Omega$.
                            Let $\net Q' := \bufImpl{ p }{ P' }{ \vect n , \vect m' }$.
                            The transition of the buffered blackbox is derived from Transition~\ruleLabel{buf-in}, \ruleLabel{buf-in-dep}, or~\ruleLabel{buf-tau}.
                            In any case, we can add messages to the back of the buffer without affecting the transition, such that $\net Q \ltrans{ \tau } \net Q'$.
                            By Transition~\ruleLabel{no-dep}, $\net Q \bLtrans{\Omega}{\tau}{\Omega_1} \net Q'$ so $\net Q \bLtrans*{\Omega}{\tau}{\Omega_1} \net Q'$.
                            Finally, by definition, $( \net P' , \Omega_1 , \net Q' ) \in \mathcal{B}$.

                        \item
                            Transition~\ruleLabel{mon-out-dep-empty}.
                            Then $\alpha = \tau$, $\mc{M} = p \send \emptyset (\ell) \mc. \mc{M'}$, and $\net P' = \monImpl{ \bufImpl{ p }{ P }{ \vect m } }{ \mc{M'} }{ \vect n }$.

                            By \satref{i:sTau}{Tau}, then $( \net P' , \RTs , \Lbls ) \in \mathcal{R}$.
                            By \cref{l:LOAlpha}, $\Omega_1 = \LO( p , \RTs , \Lbls ) = \Omega$.
                            Since the step from $\mc{M}$ to $\mc{M'}$ has no effect on $\RTs$, $\coherent{ \gtc{G_0} }{ p }{ \mc{M'} }{ \RTs }{ \vect n }$.
                            We have $\net Q \bLtrans*{\Omega}{\tau}{\Omega_1} \net Q$.
                            Then $( \net P' , \Omega_1 , \net Q ) \in \mathcal{B}$.

                        \item
                            Transition~\ruleLabel{mon-end}.
                            Then $\alpha = \pend$, $\mc{M} = \pend$, $\vect n = \epsi$, $\bufImpl{ p }{ P }{ \vect m } \ltrans{\pend} \bufImpl{ p }{ P' }{ \vect m }$, and $\net P' = \monImpl{ \bufImpl{ p }{ P' }{ \vect m } }{ \checkmark }{ \epsi }$.

                            The transition of the buffered blackbox is derived from Transition~\ruleLabel{buf-end}: $P \ltrans{\pend} P'$.
                            Then by the same transition, $\net Q \ltrans{\pend} \bufImpl{ p }{ P' }{ \vect m } =: \net Q'$.
                            By \cref{l:LOAlpha}, $\Omega_1 = \LO( p , \emptyset , \emptyset )$.
                            Then $\net Q \bLtrans*{\Omega}{\pend}{\Omega_1} \net Q'$.
                            By \satref{i:sEnd}{End}, for every $q \in \dom(\RTs)$, $\unfold\big(\RTs(q)\big) = \pend$, and $( \net P' , \emptyset , \emptyset ) \in \mathcal{R}$.
                            Moreover, clearly, $\coherent{ G_0 }{ p }{ \checkmark }{ \emptyset }{ \epsi }$.
                            Then, by definition, $( \net P' , \Omega_1 , \net Q' ) \in \mathcal{B}$.
                    \end{itemize}
            \end{itemize}

        \item
            Take any $\net Q',\alpha,\Omega_1$ such that $\net Q \bLtrans{\Omega}{ \alpha }{\Omega_1} \net Q' \checkmark$.
            The analysis depends on the rule from \cref{d:enhancedLTS} used to derive the transition (Transition~\ruleLabel{buf}, \ruleLabel{dep}, or~\ruleLabel{no-dep}).
            \begin{itemize}
                \item
                    Transition~\ruleLabel{buf-mon}.
                    This rule does not apply to monitored blackboxes.

                \item
                    Transition~\ruleLabel{buf-unmon}.
                    Then $\alpha \in \braces{ p \recv q (x) , p \recv q \Parens{x} }$, $\net Q' = \bufImpl{ p }{ P }{ u , \vect n , \vect m }$ for $u \in \braces{ q \send p (x) , q \send p \Parens(x) }$, and $\Omega_1 = \Omega(\alpha)$.

                    By \cref{l:LOAlpha}, we have $( q , \rtc{R^q} ) \in \RTs$ where $\unfold(\rtc{R^q})$ is a (dependency) message from $q$ to $p$.
                    For the sake of simplicity, assume w.l.o.g.\ that the message is no dependency.
                    Then $\unfold(\rtc{R^q}) = q \send p^{\mbb L} \rtBraces{ \msg i<T_i> \rtc. \rtc{R^q_i} }_{i \in I}$, $\alpha = p \recv q (x)$, $x = \msg j<T_j>$ with $j \in I$, and $u = q \send p (x)$.
                    Let $\RTs' := \RTs \update{q \mapsto \rtc{R^q_j}}$ and $\Lbls' := \Lbls \update{\mbb L \mapsto j}$.
                    Let $\net P' := \monImpl{ \bufImpl{ p }{ P }{ \vect m } }{ \mc{M} }{ u , \vect n }$.
                    By \satref{i:sInput}{Input}, $( \net P' , \RTs' , \Lbls' ) \in \mathcal{R}$.

                    Since $\mc{M}$ is not updated, and the addition of $u$ to the buffer only adds an input to the \nameref{d:coherentSetup} reflected by the updated for $q$ in $\RTs'$, we have $\coherent{ \gtc{G_0} }{ p }{ \mc{M} }{ \RTs' }{ u , \vect n }$.
                    By \cref{l:LOAlpha}, $\Omega_1 = \LO( p , \RTs' , \Lbls' )$.
                    Let $\Omega_2 := \Omega_1$ and $\vect b := \epsi$.
                    We have $\net P \bLtrans*{\Omega}{\vect b}{\Omega} \net P$.
                    By Transition~\ruleLabel{buf}, $\net P \bLtrans{\Omega}{\alpha}{\Omega_2} \net P'$ so $\net P \bLtrans*{\Omega}{\vect b , \alpha}{\Omega_2} \net P'$.
                    Let $\net Q'' := \net Q'$; we have $\net Q' \bLtrans*{\Omega_2}{\vect b}{\Omega_2} \net Q''$.
                    Then, since $\net Q \bLtrans{\Omega}{\alpha}{\Omega_2} \net Q'$, we have $\net Q \bLtrans{\Omega}{\alpha , \vect b}{\Omega_2} \net Q''$.
                    Finally, by definition, $( \net P' , \Omega_2 , \net Q'' ) \in \mathcal{B}$.

                \item
                    Transition~\ruleLabel{dep}.
                    Then $\alpha = \tau$, $\net Q \ltrans{\alpha'} \net Q'$ with $\alpha' = s \send r (\ell)$ for some $s,r,\ell$.
                    There are no rules to derive this transition, so this case does not apply.

                \item
                    Transition~\ruleLabel{no-dep}.
                    Then $\alpha \neq s \send r (\ell)$ for any $s,r,\ell$, $\net Q \ltrans{\alpha} \net Q'$, and $\Omega_1 = \Omega(\alpha)$.
                    The analysis depends on the derivation of the transition.
                    Some rules are impossible: there is no parallel composition in $\net Q$, and no rules to derive transitions for monitored blackboxes are possible.

                    As in case~1 above, we may first need to unfold recursion in $\mc{M}$, which we do inductively.
                    We consider each possible rule (Transition~\ruleLabel{buf-out}, \ruleLabel{buf-in}, \ruleLabel{buf-in-dep}, \ruleLabel{buf-tau}, and~\ruleLabel{buf-end}).
                    \begin{itemize}
                        \item
                            Transition~\ruleLabel{buf-out}.
                            Then $\alpha = p \send q (\msg j<T_j>)$, $\net Q' = \bufImpl{ p }{ P' }{ \vect n , \vect m }$, and $P \ltrans{\alpha} P'$.
                            Then, also by Transition~\ruleLabel{buf-out}, $\bufImpl{ p }{ P }{ \vect m } \ltrans{\alpha} \bufImpl{ p }{ P' }{ \vect m }$.

                            By \cref{l:satisfactionOutput}, $\mc{M} = \mc{M'}$ or $\mc{M} = p \send D (\ell) \mc. \mc{M'}$ where $\mc{M'} = p \send q \mBraces{ \msg i<T_i> \mc. \mc{M'_i} }_{i \in I}$ for some $I \supseteq \braces{j}$.
                            W.l.o.g., assume the latter.
                            Let $\net P' := \monImpl{ \bufImpl{ p }{ P }{ \vect m } }{ \mc{M'} }{ \vect n }$.
                            By induction on the size of $D = \braces{ r_1 , \ldots , r_k }$, we show that $\net P \bLtrans*{\Omega}{\tau }{\Omega} \net P'$, where there exists $\RTs'$ such that $( \net P' , \RTs' , \Lbls ) \in \mathcal{R}$, $\coherent{ \gtc{G_0} }{ p }{ \mc{M'} }{ \RTs' }{ \vect n }$, and $\Omega = \LO( p , \RTs' , \Lbls )$.

                            In the base case, $D = \emptyset$.
                            By Transition~\ruleLabel{mon-out-dep-empty}, $\net P \ltrans{ \tau } \net P'$.
                            Then, by Transition~\ruleLabel{no-dep}, $\net P \bLtrans{\Omega}{\tau}{\Omega} \net P'$ so $\net P \bLtrans*{\Omega}{\tau}{\Omega} \net P'$.
                            Let $\RTs':=\RTs$.
                            By \satref{i:sTau}{Tau}, $( \net P' , \RTs' , \Lbls ) \in \mathcal{R}$.
                            Since the step from $\mc{M}$ to $\mc{M'}$ does not affect any relative types, then also $\coherent{ \gtc{G_0} }{ p }{ \mc{M'} }{ \RTs' }{ \vect n }$.
                            The condition on $\Omega$ holds by \cref{l:LOAlpha}.

                            In the inductive case, $D = D' \cup \braces{r_k}$.
                            By Transition~\ruleLabel{mon-out-dep}, $\net P \ltrans{ p \send r_k \Parens{\ell} } \monImpl{ \bufImpl{ p }{ P }{ \vect m } }{ p \send D' (\ell) \mc. \mc{M'} }{ \vect n }$.
                            Then, by Transition~\ruleLabel{dep}, $\net P \bLtrans{\Omega}{\tau}{\Omega} \monImpl{ \bufImpl{ p }{ P }{ \vect m } }{ p \send D' (\ell) \mc. \mc{M'} }{ \vect n }$.
                            By \satref{i:sDependencyOutput}{Dependency output}, there exists $\RTs'$ such that $( \monImpl{ \bufImpl{ p }{ P }{ \vect m } }{ p \send D' (\ell) \mc. \mc{M'} }{ \vect n } , \RTs' , \Lbls ) \in \mathcal{R}$.
                            Since $\RTs'$ only updates the entry of $r$ (which was a dependency output in $\RTs$), also $\coherent{ \gtc{G_0} }{ p }{ p \send D' (\ell) \mc. \mc{M'} }{ \RTs' , \vect n }$.
                            Using the same reasoning, by \cref{l:labelOracleDepOut}, $\Omega = \LO( p , \RTs' , \Lbls )$.
                            Then, by the IH, $\monImpl{ \bufImpl{ p }{ P }{ \vect m } }{ p \send D' (\ell) \mc. \mc{M'} }{ \vect n } \bLtrans*{\Omega}{\tau}{\Omega} \net P'$ so $\net P \bLtrans*{\Omega}{\tau}{\Omega} \net P'$.

                            Now, let $\net P'' := \monImpl{ \bufImpl{ p }{ P' }{ \vect m } }{ \mc{M'_j} }{ \vect n }$.
                            By Transition~\ruleLabel{mon-out}, $\net P' \ltrans{\alpha} \net P''$.
                            Let $\vect b := \epsi$ and $\Omega_2 := \Omega_1$.
                            By Transition~\ruleLabel{no-dep}, since $\Omega_2 = \Omega(\alpha)$, $\net P' \bLtrans{\Omega}{\alpha}{\Omega_2} \net P''$ so $\net P \bLtrans*{\Omega}{\vect b , \alpha}{\Omega_2} \net P''$.

                            By \satref{i:sOutput}{Output}, we have $( q , \rtc{R^q} ) \in \RTs'$ with $\unfold(\rtc{R^q}) = p \send q^{\mbb{L}} \rtBraces{ \msg i<T_i> \rtc. \rtc{R^q_i} }_{i \in I}$ and $j \in I$.
                            Let $\RTs'' := \RTs' \update{q \mapsto \rtc{R^q_j}}$ and $\Lbls' := \Lbls \update{\mbb L \mapsto j}$.
                            Then $( \net P'' , \RTs'' , \Lbls' ) \in \mathcal{R}$.
                            Since $\RTs''$ only updates the entry for $r$ (which was an output in $\RTs'$), we have $\coherent{ \gtc{G_0} }{ p }{ \mc{M'_j} }{ \RTs'' }{ \vect n }$.
                            By \cref{l:LOAlpha}, $\Omega_2 = \LO( p , \RTs'' , \Lbls' )$.
                            Trivially, $\net Q' \bLtrans*{\Omega_1}{\vect b}{\Omega_2} \net Q'$, so $\net Q \bLtrans*{\Omega}{ \alpha , \vect b }{\Omega_2} \net Q'$.
                            Finally, by definition, $( \net P'' , \Omega_2 , \net Q' ) \in \mathcal{B}$.

                        \item
                            Transition~\ruleLabel{buf-in}.
                            Then $\alpha = \tau$, $\vect n , \vect m = \vect n' , \vect m' , u$ where $u = q \send p (\msg j<T_j>)$, $\net Q' = \bufImpl{ p }{ P' }{ \vect n' , \vect m' }$, and $P \ltrans{ p \recv q (\msg j<T_j>) } P'$.
                            By \cref{l:LOAlpha}, $\Omega_1 = \Omega = \LO( p , \RTs , \Lbls )$.
                            We know $u$ appears in $\vect n$ or $\vect m$.
                            We discuss each case separately.
                            \begin{itemize}
                                \item
                                    We have $u$ appears in $\vect n$.
                                    Then $\vect n = \vect n' , u$ and $\vect m = \vect m'$, and there are no messages in $\vect m$ with sender $q$.

                                    By \cref{l:finiteReceive}, from $\mc{M}$ any path leads to the input by $p$ from $q$ in finitely many steps.
                                    We show by induction on the maximal number of such steps that there are $\vect b , \vect c , \vect d , \Omega_2$ such that $\net P \bLtrans*{\Omega}{\vect b}{\Omega_2} \monImpl{ \bufImpl{ p }{ P }{ \vect d , \vect m } }{ p \recv q \mBraces{ \msg i<T_i> \mc. \mc{M'_i} }_{i \in I} }{ \vect c , u } =: \net P'$ where $j \in I$, $\net Q' \bLtrans*{\Omega_1}{\vect b}{\Omega_2} \bufImpl{ p }{ P' }{ \vect c , \vect d , \vect m } =: \net Q''$, and that there exist $\RTs',\Lbls'$ such that $( \net P' , \RTs' , \Lbls' ) \in \mathcal{R}$, $\coherent{ \gtc{G_0} }{ p }{ p \recv q \mBraces{ \msg i<T_i> \mc. \mc{M'_i} }_{i \in I} }{ \RTs' }{ \vect c , u }$, and $\Omega_2 = \LO( p , \RTs' , \Lbls' )$.

                                    In the base case, $\mc{M} = p \recv q \mBraces{ \msg i<T_i> \mc. \mc{M'_i} }_{i \in I}$ with $j \in I$.
                                    Let $\vect b := \vect d := \epsi , \vect c := \vect n' , \Omega_2 := \Omega$.
                                    Then $\net P' = \net P$ and $\net Q'' = \net Q'$, so the thesis holds trivially.

                                    In the inductive case, by \cref{l:monitorNoOutput}, $\mc{M} \in \braces{ p \recv r \mBraces{ \msg k<T_k> \mc. \mc{M_k} }_{k \in K} , p \recv r \mBraces{ k \mc. \mc{M_k} }_{k \in K} , p \send D (\ell) \mc. \mc{M'} \mid r \neq q }$.
                                    We discuss each case separately.
                                    \begin{itemize}
                                        \item
                                            $\mc{M} = p \recv r \mBraces{ \msg k<T_k> \mc. \mc{M_k} }_{k \in K}$ for $r \neq q$.
                                            This case depends on whether there is a message from $r$ in $\vect n'$.
                                            We discuss each case separately.

                                            \medskip
                                            If there is a message from $r$ in $\vect n'$, then $\vect n' , u = \vect n'' , u , w$ where $w = r \send p (\msg k'<T_{k'}>)$ for $k' \in K$.
                                            Let $\net P'' := \monImpl{ \bufImpl{ p }{ P }{ w , \vect m } }{ \mc{M_{k'}} }{ \vect n'' , u }$.
                                            By Transition~\ruleLabel{mon-in}, $\net P \ltrans{\tau} \net P''$.
                                            By \satref{i:sTau}{Tau}, $( \net P'' , \RTs , \Lbls ) \in \mathcal{R}$.
                                            Since both $\vect n'$ and $\mc{M}$ have correspondingly updated, $\coherent{ \gtc{G_0} }{ p }{ \mc{M_{k'}} }{ \RTs }{ \vect n'' , u }$.
                                            By \cref{l:LOAlpha}, we have $\Omega(\tau) = \Omega$.
                                            By Transition~\ruleLabel{no-dep}, $\net P \bLtrans{\Omega}{\tau}{\Omega} \net P''$.
                                            By the IH, there are $\vect b , \vect c , \vect d , \Omega_2$ such that $\net P'' \bLtrans*{\Omega}{\vect b}{\Omega_2} \net P'$ so $\net P \bLtrans*{\Omega}{\vect b}{\Omega_2} \net P'$, and $\net Q' \bLtrans*{\Omega_1}{\vect b}{\Omega_2} \net Q''$
                                            Moreover, there are $\RTs',\Lbls'$ such that $( \net P' , \RTs' , \Lbls' ) \in \mathcal{R}$, $\coherent{ \gtc{G_0} }{ p }{ p \recv q \mBraces{ \msg i<T_i> \mc. \mc{M'_i} }_{i \in I} }{ \RTs' }{ \vect c , u }$, and $\Omega_2 = \LO( p , \RTs' , \Lbls' )$.

                                            \medskip
                                            If there is no message from $r$ in $\vect n'$, by the definition of \nameref{d:coherentSetup}, we have $( r , \rtc{R^r} ) \in \RTs$ with $\unfold(\rtc{R^r}) = r \send p^{\mbb L} \rtBraces{ \msg k<T_k> \rtc. \rtc{R^r_k} }_{k \in K}$.
                                            Take any $k' \in K$, and let $w := r \send p (\msg k'<T_{k'}>)$.
                                            Let $\RTs'' := \RTs \update{r \mapsto \rtc{R^r_{k'}}}$ and $\Lbls'' := \Lbls \update{\mbb L \mapsto k'}$.
                                            Also, let $\Omega'_2 := \LO( p , \RTs'' , \Lbls'' )$.
                                            Then, by \cref{l:LOAlpha}, $\Omega(p \recv r (\msg k'<T_{k'}>)) = \Omega'_2$.

                                            Let $\net P''_1 := \monImpl{ \bufImpl{ p }{ P }{ \vect m } }{ \mc{M} }{ \vect n' , u , w }$ and $\net Q''' := \bufImpl{ p }{ P' }{ \vect n' , w , \vect m }$.
                                            By \satref{i:sInput}{Input}, $(\net P''_1 , \RTs'' , \Lbls'' ) \in \mathcal{R}$.
                                            Since the buffer and relative types have changed accordingly, $\coherent{ \gtc{G_0} }{ p }{ \mc{M} }{ \RTs'' }{ \vect n' , u , w }$.
                                            By Transition~\ruleLabel{buf}, $\net P \bLtrans{\Omega}{w}{\Omega'_2} \net P''_1$ and $\net Q' \bLtrans{\Omega_1}{w}{\Omega'_2} \net Q'''$.
                                            Let $\net P''_2 := \monImpl{ \bufImpl{ p }{ P }{ w , \vect m } }{ \mc{M_{k'}} }{ \vect n' , u }$.
                                            By Transition~\ruleLabel{mon-in}, $\net P''_1 \ltrans{\tau} \net P''_2$.
                                            By \cref{l:LOAlpha}, $\Omega'_2(\tau) = \Omega'_2$.
                                            Then, by Transition~\ruleLabel{no-dep}, $\net P''_1 \bLtrans{\Omega'_2}{\tau}{\Omega'_2} \net P''_2$.
                                            By \satref{i:sTau}{Tau}, $(\net P''_2 , \RTs'' , \Lbls'' ) \in \mathcal{R}$.
                                            Since the buffer and monitor have changed accordingly, $\coherent{ \gtc{G_0} }{ p }{ \mc{M_{k'}} }{ \RTs'' }{ \vect n' , u }$.

                                            By the IH, there are $\vect b , \vect c , \vect d , \Omega_2$ such that $\net P''_2 \bLtrans*{\Omega'_2}{\vect b}{\Omega_2} \net P'$ so $\net P \bLtrans*{\Omega}{w,\vect b}{\Omega_2} \net P'$, and $\net Q''' \bLtrans*{\Omega'_2}{\vect b}{\Omega_2} \net Q''$ so $\net Q' \bLtrans*{\Omega_1}{w,\vect b}{\Omega_2} \net Q''$.
                                            Moreover, there are $\RTs',\Lbls'$ such that $( \net P' , \RTs' , \Lbls' ) \in \mathcal{R}$, $\coherent{ \gtc{G_0} }{ p }{ p \recv q \mBraces{ \msg i<T_i> \mc. \mc{M'_i} }_{i \in I} }{ \RTs' }{ \vect c , u }$, and $\Omega_2 = \LO( p , \RTs' , \Lbls' )$.

                                        \item
                                            $\mc{M} = p \recv r \mBraces{ k \mc. \mc{M_k} }_{k \in K}$ for $r \neq q$.
                                            This case is analogous to the one above.

                                        \item
                                            $\mc{M} = p \send D (\ell) \mc. \mc{M'}$.
                                            Let $\net P'' := \monImpl{ \bufImpl{ p }{ P }{ \vect m } }{ \mc{M'} }{ \vect n' , u }$.
                                            Similar to the case of Transition~\ruleLabel{buf-out} above, $\net P \bLtrans*{\Omega}{\tau}{\Omega} \net P''$.
                                            Moreover, there are $\RTs'',\Lbls''$ such that $(\net P'' , \RTs'' , \Lbls'' ) \in \mathcal{R}$, and $\coherent{ \gtc{G_0} }{ p }{ \mc{M'} }{ \RTs'' }{ \vect n' , u }$.
                                            By \cref{l:labelOracleDepOut}, $\Omega = \LO( p , \RTs'' , \Lbls'' )$.

                                            By the IH, there are $\vect b , \vect c , \vect d , \Omega_2$ such that $\net P'' \bLtrans*{\Omega}{\vect b}{\Omega_2} \net P'$ so $\net P \bLtrans*{\Omega}{\vect b}{\Omega_3} \net P'$, and $\net Q' \bLtrans*{\Omega_1}{\vect b}{\Omega_2} \net Q''$
                                            Moreover, there are $\RTs',\Lbls'$ such that $( \net P' , \RTs' , \Lbls' ) \in \mathcal{R}$, $\coherent{ \gtc{G_0} }{ p }{ p \recv q \mBraces{ \msg i<T_i> \mc. \mc{M'_i} }_{i \in I} }{ \RTs' }{ \vect c , u }$, and $\Omega_2 = \LO( p , \RTs' , \Lbls' )$.
                                    \end{itemize}

                                    Let $\net P'''_1 := \monImpl{ \bufImpl{ p }{ P }{ \vect d , \vect m , u } }{ \mc{M'_j} }{ \vect c }$.
                                    By Transition~\ruleLabel{mon-in}, $\net P' \ltrans{\tau} \net P'''_1$.
                                    Since the buffer and monitor changed accordingly, $\coherent{ \gtc{G_0} }{ p }{ \mc{M'_j} }{ \RTs' }{ \vect c }$.
                                    By Transition~\ruleLabel{no-dep}, $\net P' \bLtrans{\Omega_2}{\tau}{\Omega_2} \net P'''_1$.
                                    Let $\net P'''_2 := \monImpl{ \bufImpl{ p }{ P' }{ \vect d , \vect m } }{ \mc{M'_j} }{ \vect c }$.
                                    By \cref{l:LOAlpha}, $\Omega_2(\tau) = \Omega_2$.
                                    By Transition~\ruleLabel{buf-in}, $\bufImpl{ p }{ P }{ \vect d , \vect m , u } \ltrans{\tau} \bufImpl{ p }{ P' }{ \vect d , \vect m }$, so, by Transition~\ruleLabel{mon-tau}, $\net P'''_1 \ltrans{\tau} \net P'''_2$.
                                    By Transition~\ruleLabel{no-dep}, $\net P'''_1 \bLtrans{\Omega_2}{\tau}{\Omega_2} \net P'''_2$.
                                    By \satref{i:sTau}{Tau}, $(\net P'''_2 , \RTs' , \Lbls') \in \mathcal{R}$.
                                    We have $\net P \bLtrans*{\Omega}{\vect b,\tau}{\Omega_2} \net P'''_2$ and $\net Q \bLtrans{\Omega}{\tau}{\Omega_1} \net Q' \bLtrans*{\Omega_1}{\vect b}{\Omega_2} \net Q''$.
                                    Recall that $\Omega_2 = \LO( p , \RTs' , \Lbls' )$.
                                    Then, by definition, $( \net P'''_2 , \Omega_2 , \net Q'' ) \in \mathcal{B}$.

                                \item
                                    We have $u$ appears in $\vect m$.
                                    Then $\vect n = \vect n'$ and $\vect m = \vect m' , u$.
                                    Then, by Transition~\ruleLabel{buf-in}, $\bufImpl{ p }{ P }{ \vect m' , u } \ltrans{ \tau } \bufImpl{ p }{ P' }{ \vect m' }$.

                                    Let $\net P' := \monImpl{ \bufImpl{ p }{ P' }{ \vect m' } }{ \mc{M} }{ \vect n }$.
                                    By Transition~\ruleLabel{mon-tau}, $\net P \ltrans{\tau} \net P'$.

                                    By \cref{l:LOAlpha}, $\Omega(\tau) = \Omega$, so $\Omega_1 = \Omega$.
                                    Let $\Omega_2 := \Omega$.
                                    Then, by Transition~\ruleLabel{no-dep}, $\net P \bLtrans{\Omega}{\tau}{\Omega_2} \net P'$ so $\net P \bLtrans*{\Omega}{\tau}{\Omega_2} \net P'$.
                                    By \satref{i:sTau}{Tau}, $( \net P' , \RTs , \Lbls ) \in \mathcal{R}$.
                                    Finally, $\net Q \bLtrans{\Omega}{\tau}{\Omega_1} \net Q' \bLtrans*{\Omega_1}{\epsi}{\Omega_2} \net Q'$.
                                    Then $( \net P' , \Omega_2 , \net Q' ) \in \mathcal{B}$.
                            \end{itemize}

                        \item
                            Transition~\ruleLabel{buf-in-dep}.
                            Analogous to Transition~\ruleLabel{buf-in}.

                        \item
                            Transition~\ruleLabel{buf-tau}.
                            Then $\alpha = \tau$, $\net Q' = \bufImpl{ p }{ P' }{ \vect n , \vect m }$, and $P \ltrans{ \tau } P'$.
                            By Transition~\ruleLabel{buf-tau}, also $\bufImpl{ p }{ P }{ \vect m } \ltrans{ \tau } \bufImpl{ p }{ P' }{ \vect m }$.

                            By \cref{l:monitorNotCheck}, $\mc{M} \neq \checkmark$.
                            Let $\net P' := \monImpl{ \bufImpl{ p }{ P' }{ \vect m } }{ \mc{M} }{ \vect n }$.
                            By Transition~\ruleLabel{mon-tau}, $\net P \ltrans{ \tau } \net P'$.

                            By \cref{l:LOAlpha}, $\Omega(\tau) = \Omega$, so $\Omega_1 = \Omega$.
                            Let $\Omega_2 := \Omega$.
                            Then, by Transition~\ruleLabel{no-dep}, $\net P \bLtrans{\Omega}{ \alpha }{\Omega} \net P'$ so $\net P \bLtrans*{\Omega}{\epsi , \alpha }{\Omega_2} \net P'$.
                            By \satref{i:sTau}{Tau}, $( \net P' , \RTs , \Lbls ) \in \mathcal{R}$.
                            Finally, we have $\net Q \bLtrans{\Omega}{\alpha}{\Omega_1} \net Q' \bLtrans*{\Omega_1}{\epsi}{\Omega_2} \net Q'$.
                            Then $( \net P' , \Omega_2 , \net Q' ) \in \mathcal{B}$.

                        \item
                            Transition~\ruleLabel{buf-end}.
                            Then $\alpha = \pend$, $P \ltrans{ \pend } P'$, and $\net Q' = \bufImpl{ p }{ P' }{ \vect n , \vect m }$.
                            By Transition~\ruleLabel{buf-end}, also $\bufImpl{ p }{ P }{ \vect m } \ltrans{\pend} \bufImpl{ p }{ P' }{ \vect m }$.

                            By \cref{l:satisfactionEnd}, $\vect n = \epsi$, and $\mc{M} = \pend$ or $\mc{M} = p \send D (\ell) \mc. \pend$. W.l.o.g., assume the latter.
                            Let $\net P' := \monImpl{ \bufImpl{ p }{ P }{ \vect m } }{ \pend }{ \epsi }$.
                            Similar to the case for Transition~\ruleLabel{buf-out}, $\net P \bLtrans*{\Omega}{\tau}{\Omega} \net P'$.

                            Let $\net P'' := \monImpl{ \bufImpl{ p }{ P' }{ \vect m } }{ \checkmark }{ \epsi }$.
                            By Transition~\ruleLabel{mon-end}, $\net P' \ltrans{\pend} \net P''$.
                            Let $\vect b := \epsi$ and $\Omega_2 := \LO( p , \emptyset , \emptyset )$.
                            By \cref{l:LOAlpha}, $\Omega(\pend) = \Omega_2 = \Omega_1$.
                            By Transition~\ruleLabel{no-dep}, since $( \pend , \Omega_2 ) \in \Omega$, $\net P' \bLtrans{\Omega}{\pend}{\Omega_2} \net P''$, so $\net P \bLtrans*{\Omega}{\vect b , \pend}{\Omega_2} \net P''$.

                            By \satref{i:sEnd}{End}, for every $q \in \dom(\RTs)$, $\unfold\big(\RTs(q)\big) = \pend$, and $( \net P'' , \emptyset , \emptyset ) \in \mathcal{R}$.
                            Clearly, $\coherent{\gtc{G_0}}{p}{\checkmark}{\emptyset}{\vect n}$.
                            Then, trivially, $\net Q' \bLtrans*{\Omega_1}{\vect b}{\Omega_2} \net Q'$, so $\net Q \bLtrans*{\Omega}{\pend , \vect b}{\Omega_2} \net Q'$.
                            Finally, by definition, $( \net P'' , \Omega_2 , \net Q' ) \in \mathcal{B}$.
                            \qedhere
                    \end{itemize}
            \end{itemize}
    \end{enumerate}
\end{proof}